\documentclass[11pt]{article}
\usepackage{amsmath}
\usepackage{amsthm}
\usepackage{amssymb}
\usepackage{algorithm}
\usepackage{subfig}
\usepackage{color}
\usepackage[english]{babel}
\usepackage{graphicx}
\usepackage{grffile}
\usepackage{wrapfig,epsfig}
\usepackage{epstopdf}
\usepackage{url}
\usepackage{color}
\usepackage{epstopdf}
\usepackage{algpseudocode}
\usepackage[T1]{fontenc}
\usepackage{bbm}
\usepackage{comment}
\usepackage{dsfont}
\usepackage{tikz}
\usepackage{hyperref}  %%% arxiv don't allow this.
\hypersetup{colorlinks=true,citecolor=red,linkcolor=blue} %%% Zhao : maybe we should comment this in submission.
\usetikzlibrary{arrows}
\usepackage[margin=1in]{geometry}
\graphicspath{{./figs/}}
\author{
	Yin Tat Lee \thanks{
	\texttt{yintat@uw.edu}
	University of Washington \& Microsoft Research Redmond. Part of the work done while visiting Simons Institute for the Theory of Computing at Berkeley.}
	\and
	Zhao Song \thanks{
	\texttt{zhaosong@uw.edu}
	University of Washington \& UT-Austin. Part of the work done while visiting University of Washington and hosted by Yin Tat Lee. Part of the work done while visiting Simons Institute for the Theory of Computing at Berkeley and hosted by Peter L. Bartlett, Nikhil Srivastava, Santosh Vempala, and David P. Woodruff.}
	\and
	Qiuyi Zhang \thanks{
	\texttt{qiuyi@math.berkeley.edu}
	University of California, Berkeley.}
}
\date{}

\title{Solving Empirical Risk Minimization in \\the Current Matrix Multiplication Time}

\newtheorem{theorem}{Theorem}[section]
\newtheorem{lemma}[theorem]{Lemma}
\newtheorem{definition}[theorem]{Definition}

\newtheorem{corollary}[theorem]{Corollary}

\newtheorem{fact}[theorem]{Fact}
\newtheorem{remark}[theorem]{Remark}

\newcommand{\wh}{\widehat}
\newcommand{\wt}{\widetilde}
\newcommand{\ov}{\overline}

\newcommand{\R}{\mathbb{R}}

\newcommand{\LHS}{\mathrm{LHS}}

\renewcommand{\varepsilon}{\epsilon}
\renewcommand{\tilde}{\wt}

\renewcommand{\d}{\mathrm{d}}

\DeclareMathOperator*{\E}{{\bf {E}}}

\DeclareMathOperator{\poly}{poly}

\DeclareMathOperator{\new}{new}
\DeclareMathOperator{\pre}{pre}
\DeclareMathOperator{\old}{old}
\DeclareMathOperator{\tr}{tr}
\algnewcommand{\LineComment}[1]{\State \(\triangleright\) #1}

\makeatletter
\newcommand*{\RN}[1]{\expandafter\@slowromancap\romannumeral #1@}
\makeatother

\usepackage{lineno}

\usepackage{etoolbox}

\makeatletter

\newcommand{\define}[4][ignore]{%
  \ifstrequal{#1}{ignore}{}{
  \@namedef{thmtitle@#2}{#1}}%
  \@namedef{thm@#2}{#4}%
  \@namedef{thmtypen@#2}{lemma}%
  \newtheorem{thmtype@#2}[theorem]{#3}%
  \newtheorem*{thmtypealt@#2}{#3~\ref{#2}}%
}

\newcommand{\state}[1]{%
  \@namedef{curthm}{#1}
  \@ifundefined{thmtitle@#1}{
  \begin{thmtype@#1}
    }{
  \begin{thmtype@#1}[\@nameuse{thmtitle@#1}]
  }
    \label{#1}
    \@nameuse{thm@#1}
  \end{thmtype@#1}
  \@ifundefined{thmdone@#1}{
  \@namedef{thmdone@#1}{stated}%
  }{}
}

\newcommand{\restate}[1]{%
  \@namedef{curthm}{#1}
  \@ifundefined{thmtitle@#1}{
    \begin{thmtypealt@#1}
    }{
  \begin{thmtypealt@#1}[\@nameuse{thmtitle@#1}]
  }
    \@nameuse{thm@#1}
  \end{thmtypealt@#1}
  \@ifundefined{thmdone@#1}{
  \@namedef{thmdone@#1}{stated}%
  }{}
}

\newcommand{\thmlabel}[1]{
  \@ifundefined{thmdone@\@nameuse{curthm}}{\label{#1}
    }{\tag*{\eqref{#1}}}
}
\makeatother
\begin{document}
%\linenumbers

\begin{titlepage}
  \maketitle
  \begin{abstract}
Many convex problems in machine learning and computer science share the same form: 
 \begin{align*}
 \min_{x} \sum_{i} f_i( A_i x + b_i),
 \end{align*}
where $f_i$ are convex functions on $\R^{n_i}$ with constant $n_i$, $A_i \in \R^{n_i \times d}$, $b_i \in \R^{n_i}$ and $\sum_i n_i = n$. This problem generalizes linear programming and includes many problems in empirical risk minimization.

In this paper, we give an algorithm that runs in time
\begin{align*}
O^* (  ( n^{\omega} + n^{2.5 - \alpha/2} + n^{2+ 1/6} ) \log (n / \delta)  )
\end{align*}
where $\omega$ is the exponent of matrix multiplication, $\alpha$ is the dual exponent of matrix multiplication, and $\delta$ is the relative accuracy. 
Note that the runtime has only a log dependence on the condition numbers or other data dependent parameters and these are captured in $\delta$. 
For the current bound $\omega \sim 2.38$ [Vassilevska Williams'12, Le Gall'14] and $\alpha \sim 0.31$ [Le Gall, Urrutia'18], our runtime $O^* ( n^{\omega} \log (n / \delta))$ matches the current best for solving a dense least squares regression problem, a special case of the problem we consider. Very recently, [Alman'18] proved that all the current known techniques can not give a better $\omega$ below $2.168$ which is larger than our $2+1/6$.

Our result generalizes the very recent result of solving linear programs in the  current matrix multiplication time [Cohen, Lee, Song'19] to a more broad class of problems. Our algorithm proposes two concepts which are different from [Cohen, Lee, Song'19] :\\
$\bullet$ We give a robust deterministic central path method, whereas the previous one is a stochastic central path which updates weights by a random sparse vector. \\
$\bullet$ We propose an efficient data-structure to maintain the central path of interior point methods even when the weights update vector is dense.

  \end{abstract}
  \thispagestyle{empty}
\end{titlepage}

\section{Introduction}

Empirical Risk Minimization (ERM) problem is a fundamental question in statistical machine learning. There are a huge number of papers that have considered this topic \cite{n83,v92,pj92,n04,bbm05,bb08,njls09,mb11,fgrw12,lsb12,jz13,v13,sz13,db14,dbl14,fgks15,db16,slcny17,zyj17,zx17,zwxxz17,gs17,ms17,ns17,akklns17,c18,jlgj18} as almost all convex optimization machine learning can be phrased in the ERM framework \cite{sb14,v92}. While the statistical convergence properties and generalization bounds for ERM are well-understood, a general runtime bound for general ERM is not known although fast runtime bounds do exist for specific instances \cite{akps19}. 

Examples of applications of ERM include linear regression, LASSO \cite{t96}, elastic net \cite{zh05}, logistic regression \cite{c58,hls13}, support vector machines \cite{cv95}, $\ell_p$ regression \cite{c05,ddhkm09,bcll18,akps19}, quantile regression \cite{k00,kh01,k05}, AdaBoost \cite{fs97}, kernel regression \cite{n64,w64}, and mean-field variational inference \cite{xjr02}.
%%%%%% Zhao : I am a bit confused about the original citations in the above sentence.
%\cite{ss04}, %%% Zhao : why wo we cite that paper for SVM?
%\cite{ss04adaboost}, %%% Zhao : why do we cite that paper for adaboost?
%\cite{hk01}, %%% Zhao : I read the wiki, it seems there are even early paper
%\cite{tfm06}, %%%% Zhao : Why is this the Kernel regression paper to cite?

The classical Empirical Risk Minimization problem is defined as
\begin{align*}
\min_{x} \sum_{i=1}^m f_i ( a_i^\top x + b_i )
\end{align*}
where $f_i : \R \rightarrow \R$ is a convex function, $a_i \in \R^d$, and $b_i \in \R$, $\forall i \in [m]$. Note that this formulation also captures most standard forms of regularization as well.

Letting $y_i = a_i^\top x + b_i$, and $z_i = f_i( a_i^\top x + b_i )$ allows us to rewrite the original problem in the following sense,
\begin{align}\label{eq:apply_main_theorem_and_get_ERM}
\min_{x,y,z} ~ & ~ \sum_{ i = 1 }^m z_i \\
\text{~s.t.~} & ~ A x + b = y \notag \\
& ~ (y_i, z_i) \in K_i = \{ (y_i , z_i) : f_i(y_i) \leq z_i \}, \forall i \in [m] \notag
% & ~ f_i ( y_i ) \leq z_i, \forall i \in [n]
\end{align}
%Nesterov \cite{n98} proved that, for any convex function $f_i$, there exists a low-dimensional convex set $K_i$ such that $f_i (y_i) \leq z_i$ is equivalent to length-2 vector $(y_i, z_i)$ living in $K_i$. 
We can consider a more general version where dimension of $K_i$ can be arbitrary, e.g. $n_i$. Therefore, %reformulating the last constraint as $(y_i,z_i) \in K_i$, 
we come to study the general $n$-variable form $$\min_{ x \in \prod_{i=1}^m K_i, A x = b } c^\top x$$ where $\sum_{i=1}^m n_i = n$. We state our main result for solving the general model.
\begin{theorem}[Main result, informal version of Theorem~\ref{thm:main_result_formal}]\label{thm:main_result_informal}
Given a matrix $A \in \R^{d \times n}$, two vectors $b \in \R^d$, $c \in \R^n$, and $m$ compact convex sets $K_1, K_2, \cdots, K_m$. Assume that there is no redundant constraints and $n_i = O(1)$, $\forall i \in [m]$. There is an algorithm (procedure \textsc{Main} in Algorithm~\ref{alg:main}) that solves
\begin{align*}
\min_{ x \in \prod_{i=1}^m K_i , A x = b } c^\top x
\end{align*}
up to $\delta$ precision and runs in expected time
\begin{align*}
\wt{O} \left( ( n^{\omega+o(1)} + n^{2.5-\alpha/2 + o(1)} + n^{2+1/6 + o(1) } ) \cdot \log ( \frac{n}{\delta} ) \right)
\end{align*}
where $\omega$ is the exponent of matrix multiplication, $\alpha$ is the dual exponent of matrix multiplication. 

For the current value of $\omega \sim 2.38$ \cite{w12,l14} and $\alpha \sim 0.31$ \cite{gu18}, the expected time is simply $n^{\omega + o(1)} \wt{O}(\log ( \frac{n}{\delta} ) )$.
\end{theorem}

\begin{remark}
More precisely, when $n_i$ is super constant, our running time depends polynomially on $\max_{i\in [m]} n_i$ (but not exponential dependence).
\end{remark}

Also note that our runtime depends on diameter, but logarithmically to the diameter. So, it can be applied to linear program by imposing an artificial bound on the solution.

\subsection{Related Work}

First-order algorithms for ERM are well-studied and a long series of accelerated stochastic gradient descent algorithms have been developed and optimized \cite{n98,jz13,xz14,sz14,fgks15,lmh15,mlf15,ay16,rhsps16,s16,ah16,slb17,ms17,lmh17,ljcj17,a17a,a17b,a18a,a18b}. However, these rates depend polynomially on the Lipschitz constant of $\nabla f_i$ and in order to achieve a $\log(1/\epsilon)$ dependence, the runtime will also have to depend on the strong convexity of the $\sum_i f_i$. In this paper, we want to focus on algorithms that depend logarithmically on diameter/smoothness/strong convexity constants, as well as the error parameter $\epsilon$.  Note that gradient descent and a direct application of Newton's method do not belong to these class of algorithms, but for example, interior point method and ellipsoid method does.

%{\color{red}
Therefore, in order to achieve high-accuracy solutions for non-smooth and non strongly convex case, most convex optimization problems will rely on second-order methods, often under the general interior point method (IPM) or some sort of iterative refinement framework. So, we note that our algorithm is thus optimal in this general setting since second-order methods require at least $n^{\omega}$ runtime for general matrix inversion. %Indeed, note that if we can solve linear regression faster than $n^{\omega}$ time, then matrix inversion would also be sped up. }

Our algorithm applies the interior point method framework to solve ERM. The most general interior point methods require $O(\sqrt{n})$-iterations of linear system solves \cite{n98}, requiring a naive runtime bound of $O(n^{\omega + 1/2})$. Using the inverse maintenance technique \cite{v89b,cls18}, one can improve the running time for LP to $O(n^{\omega})$. This essentially implies that almost all convex optimization problems can be solved, up to subpolynomial factors, as fast as linear regression or matrix inversion!

The specific case of $\ell_2$ regression can be solved in $O(n^\omega)$ time since the solution is explicitly given by solving a linear system. In the more general case of $\ell_p$ regression, \cite{bcll18} proposed a $\widetilde{O}_p(n^{|1/2 -1/p|})$-iteration iterative solver with a naive $O(n^\omega)$ system solve at each step. Recently, \cite{akps19} improved the runtime to $\widetilde{O}_p(n^{\max{(\omega,7/3)}})$, which is current matrix multiplication time as $\omega > 7/3$. However, both these results depend exponentially on $p$ and fail to be impressive for large $p$. Otherwise, we are unaware of other ERM formulations that have have general runtime bounds for obtaining high-accuracy solutions.

Recently several works \cite{aw18a,aw18b,a18} try to show the limitation of current known techniques for improving matrix multiplication time. Alman and Vassilevska Williams \cite{aw18b} proved limitations of using the Galactic method applied to many tensors of interest (including Coppersmith-Winograd tensors \cite{cw87}). More recently, Alman \cite{a18} proved that by applying the Universal method on those tensors, we cannot hope to achieve any running time better than $n^{2.168}$ which is already above our $n^{2+1/6}$.  

%Our result can be seen as a generalization of a previous work that solves linear programming in current matrix multiplication time \cite{cls18}. Besides the standard generalization from linear programming to minimizing self-concordant barriers, the two main technical novelties introduced in this paper are 1) instead of moving along the stochastic central path, we propose moving along the robust central path and 2) previous work relied on a deterministic projection maintenance data structure while in this paper, we introduce a randomized projection maintenance data structure.

%The robust central path is defined implicitly via a potential function? (add more intuition and section refs later)
\section{Overview of Techniques}

%{\color{red}
In this section, we discuss the key ideas in this paper. Generalizing the stochastic sparse update approach of \cite{cls18} to our setting is a natural first step to speeding up the matrix-vector multiplication that is needed in each iteration of the interior point method. In linear programs, maintaining approximate complementary slackness means that we maintain $x, s$ to be close multiplicatively to the central path under some notion of distance. However, the generalized notion of complementary slackness requires a barrier-dependent notion of distance. Specifically, if $\phi(x)$ is a barrier function, then our distance is now defined as our function gradient being small in a norm depending on $\nabla^2 \phi(x)$. One key fact of the stochastic sparse update is that the variance introduced does not perturb the approximation too much, which requires understanding the second derivative of the distance function. For our setting, this would require bounding the 4th derivative of $\phi(x)$, which may not exist for self-concordant functions. So, the stochastic approach may not work algorithmically (not just in the analysis) if $\phi(x)$ is assumed to be simply self-concordant. Even when assumptions on the 4th derivative of $\phi(x)$ are made, the analysis will become significantly more complicated due to the 4th derivative terms. To avoid these problems, the main contributions of this paper is to 1) introduce a robust version of the central path and 2) exploit the robustness via sketching to apply the desired matrix-vector multiplication fast.%}

More generally, our main observation is that one can generally speed up an iterative method using sketching if the method is robust in a certain sense. To speed up interior point methods, in Section~\ref{sec:robust_central_path_intro} and \ref{sec:robust_central_path}, we give a robust version of the interior point method; and in Section~\ref{sec:central_path_maintenance}, we give a
data structure to maintain the sketch; and in Section~\ref{sec:combine}, we show how to combine them together. We provide several basic notations and definitions for numerical linear algebra in Section~\ref{sec:preli}. In Section~\ref{sec:initial_point_termination_condition}, we provide some classical lemmas from the literature of interior point methods. In Section~\ref{sec:fastjl}, we prove some basic properties of the sketching matrix. Now, we first begin with an overview of our robust central path and then proceed with an overview of sketching iterative methods.

\subsection{Central Path Method}

We consider the following optimization problem
\begin{equation}
\min_{x\in\prod_{i=1}^{m}K_{i},Ax=b}c^{\top}x\label{eq:origin_problem}
\end{equation}
where $\prod_{i=1}^{m}K_{i}$ is the direct product of $m$ low-dimensional
convex sets $K_{i}$. We let $x_{i}$ be the $i$-th block of $x$
corresponding to $K_{i}$. Interior point methods consider the path of solutions to the following optimization
problem: 
\begin{equation}
x(t) = \arg\min_{Ax=b}c^{\top}x+t\sum_{i=1}^{m}\phi_{i}(x_{i})\label{eq:panelized_problem}
\end{equation}
where $\phi_{i}:K_{i}\rightarrow\R$ are self-concordant barrier functions. This parameterized path is commonly known as the {\it central path}. Many algorithms solve the original problem (\ref{eq:origin_problem}) by following the central path as the path parameter is decreased $t \to 0$. The rate at which we decrease $t$ and subsequently the runtimes of these path-following algorithms are usually governed by the self-concordance properties of the barrier functions we use.

\begin{definition}
We call a function $\phi$ a $\nu$ self-concordant
barrier for $K$ if $\mathrm{dom}\phi=K$ and for any $x\in\mathrm{dom}\phi$
and for any $u\in\R^{n}$
\[
|D^{3}\phi(x)[u,u,u]|\leq2\|u\|_{x}^{3/2}\quad\text{and}\quad\|\nabla\phi(x)\|_{x}^{*}\leq\sqrt{\nu}
\]
where $\|v\|_{x}:=\|v\|_{\nabla^{2}\phi(x)}$ and $\|v\|_{x}^{*}:=\|v\|_{\nabla^{2}\phi(x)^{-1}}$, for any vector $v$.
\end{definition}

\begin{remark}It is known that $\nu\geq1$ for any self-concordant
barrier function.\end{remark}

Nesterov and Nemirovsky showed that for any open convex set $K\subset\R^{n}$,
there is a $O(n)$ self-concordant barrier function \cite{n98}. In
this paper, the convex set $K_{i}$ we considered has $O(1)$ dimension.  While Nesterov and Nemirovsky gave formulas for the universal barrier; in practice, most ERM problems lend themselves to explicit $O(1)$ self-concordant barriers for majority of the convex functions people use. For example, for the set $\{x: \|x\|<1\}$, we use $-\log(1-\|x\|^2)$; for the set $\{x: x>0\}$, we use $-\log(x)$, and so on. That is the reason why we assume the gradient and hessian can be computed in $O(1)$ time. Therefore, in this paper, we assume
a $\nu_{i}$ self-concordant barrier $\phi_{i}$ is provided and that
we can compute $\nabla\phi_{i}$ and $\nabla^{2}\phi_{i}$ in $O(1)$
time. The main result we will use about self-concordance is that the
norm $\|\cdot\|_{x}$ is stable when we change $x$.

\begin{theorem}[Theorem 4.1.6 in \cite{n98}]\label{thm:hessiansc}
If $\phi$ is a self-concordant barrier and if $\|y-x\|_{x}<1$, then
we have : 
\begin{align*}
(1-\|y-x\|_{x})^{2}\nabla^{2}\phi(x)\preceq\nabla^{2}\phi(y)\preceq\frac{1}{(1-\|y-x\|_{x})^{2}}\nabla^{2}\phi(x).
\end{align*}
\end{theorem}

In general, we can simply think of $\phi_{i}$ as a function penalizing
any point $x_{i}\notin K_{i}$. It is known how to transform the original
problem \eqref{eq:origin_problem} by adding $O(n)$ many variables
and constraints so that
\begin{itemize}
\item The minimizer $x(t)$ at $t=1$ is explicitly given.
\item One can obtain an approximate solution of the original problem using
the minimizer at small $t$ in linear time.
\end{itemize}
For completeness, we show how to do it in Lemma \ref{lem:feasible_LP}.
Therefore, it suffices to study how we can move efficiently from $x(1)$
to $x(\epsilon)$ for some tiny $\epsilon$ where $x(t)$ is again the minimizer
of the problem (\ref{eq:panelized_problem}).

\subsection{Robust Central Path}

In the standard interior point method, we use a tight $\ell_2$-bound to control how far we can deviate from $x(t)$ during the entirety of the algorithm. Specifically, if we denote $\gamma_i^t(x_i)$ as the appropriate measure of error (this will be specified later and is often called the Newton Decrement) in each block coordinate $x_i$ at path parameter $t$, then as we let $t\to 0$, the old invariant that we are maintaining is,
\begin{align*}
 \Phi_{\old}^{t}(x) = \sum_{i=1}^m \gamma_i^t(x_i)^2 \leq O(1).
 \end{align*}
It can be shown that a Newton step in the standard direction will allow for us to maintain $\Phi_{\old}^{t}$ to be small even as we decrease $t$ by a multiplicative factor of $O(m^{-1/2})$ in each iteration, thereby giving a standard $O(\sqrt{m})$ iteration analysis. Therefore, the standard approach can be seen as trying to remain within a small $\ell_2$ neighborhood of the central path by centering with Newton steps after making small decreases in the path parameter $t$. Note however that if each $\gamma_i$ can be perturbed by an error that is $\Omega(m^{-1/2})$, $\Phi_{\old}^{t}(x)$ can easily become too large for the potential argument to work.

To make our analysis more robust, we introduce a robust version that maintains the soft-max potential:
\begin{align*}
\Phi_{\new}^t(x) = \sum_{i=1}^m \exp(\lambda \gamma_i^t(x_i)) \leq O(m) 
\end{align*}
for some $\lambda = \Theta(\log m)$. The robust central path is simply the region of all $x$ that satisfies our potential inequality. We will specify the right constants later but we always make $\lambda$ large enough to ensure that $\gamma_i \leq 1$ for all $x$ in the robust central path. Now note that a $\ell_\infty$ perturbation of $\gamma$ translates into a small multiplicative change in $\Phi^t$, tolerating errors on each $\gamma_i$ of up to $O(1/\text{poly}\log(n))$. 

However, maintaining $\Phi_{\new}^t(x) \leq O(m)$ is not obvious because the robust central path is a much wider region of $x$ than the typical $\ell_2$-neighborhood around the central path. We will show later how to modify the standard Newton direction to maintain $\Phi_{\new}^t(x) \leq O(m)$ as we decrease $t$. Specifically, we will show that a variant of gradient descent of $\Phi_{\new}^t$ in the Hessian norm suffices to provide the correct guarantees.

\subsection{Speeding up via Sketching} 

To motivate our sketching algorithm, we consider an imaginary iterative method
\[
z^{(k+1)}\leftarrow z^{(k)}+P\cdot F(z^{(k)})
\]
where $P$ is some dense matrix and $F(z)$ is some simple formula
that can be computed efficiently in linear time. Note that the cost
per iteration is dominated by multiplying $P$ with a vector, which
takes $O(n^{2})$ time. To avoid the cost of multiplication, instead
of storing the solution explicitly, we store it implicitly by $z^{(k)}=P\cdot u^{(k)}$.
Now, the algorithm becomes
\[
u^{(k+1)}\leftarrow u^{(k)}+F(P\cdot u^{(k)}).
\]
This algorithm is as expensive as the previous one except that we
switch the location of $P$. However, if we know the algorithm is
robust under perturbation of the $z^{(k)}$ term in $F(z^{(k)})$,
we can instead do
\[
u^{(k+1)}\leftarrow u^{(k)}+F( R^{\top} R P \cdot u^{(k)} )
\]
for some random Gaussian matrix $R : \R^{b \times n}$. Note that the
matrix $R P$ is fixed throughout the whole algorithm and can
be precomputed. Therefore, the cost of per iteration decreases from
$O(n^{2})$ to $O(nb)$.

For our problem, we need to make two adjustments. First, we need to
sketch the change of $z$, that is $F(P \cdot u^{(k)})$, instead of
$z^{(k)}$ directly because the change of $z$ is smaller and this
creates a smaller error. Second, we need to use a fresh random $R$
every iteration to avoid the randomness dependence issue in the proof.
For the imaginary iterative process, it becomes
\begin{align*}
\overline{z}^{(k+1)} & \leftarrow\overline{z}^{(k)}+R^{(k)\top}R^{(k)}P\cdot F(\overline{z}^{(k)}),\\
u^{(k+1)} & \leftarrow u^{(k)}+F(\overline{z}^{(k)}).
\end{align*}
After some iterations, $\overline{z}^{(k)}$ becomes too far from
$z^{(k)}$ and hence we need to correct the error by setting $z^{(k)}=P\cdot u^{(k)}$,
which zeros the error.

Note that the algorithm explicitly maintains the approximate vector
$\overline{z}$ while implicitly maintaining the exact vector $z$ by
$P u^{(k)}$. This is different from the classical way to sketch Newton
method \cite{pw16,pw17}, which is to simply run $z^{(k+1)}\leftarrow z^{(k)}+R^{\top} R P\cdot F(z^{(k)})$ or use another way to subsample and approximate $P$. Such a scheme relies on the iteration method to fix the error
accumulated in the sketch, while we are actively fixing the error
by having both the approximate explicit vector $\overline{z}$ and
the exact implicit vector $z$.

Without precomputation, the cost of computing $R^{(k)} P$ is in fact higher
than that of $P\cdot F(z^{(k)})$. The first one involves multiplying
multiple vectors with $P$ and the second one involves multiplying 1
vector with $P$. However, we can precompute 
$[ R^{(1)\top}; R^{(2)\top}; \cdots; R^{(T)\top} ]^\top \cdot P$
by fast matrix multiplication. This decreases the cost of multiplying
$1$ vector with $P$ to $n^{\omega-1}$ per vector. This is a huge
saving from $n^{2}$. In our algorithm, we end up using only $\tilde{O}(n)$
random vectors in total and hence the total cost is still roughly
$n^{\omega}$.

\subsection{Maintaining the Sketch}

The matrix $P$ we use in interior point methods is of the form
\[
P=\sqrt{W}A^{\top}(AWA^{\top})^{-1}A\sqrt{W}
\]
where $W$ is some block diagonal matrix. \cite{cls18} showed
one can approximately maintain the matrix $P$ with total cost $\tilde{O}(n^{\omega})$
across all iterations of interior point method. However, the cost
of applying the dense matrix $P$ with a vector $z$ is roughly $O(n\|z\|_{0})$
which is $O(n^{2})$ for dense vectors. Since interior point methods
takes at least $\sqrt{n}$ iterations in general, this gives a total
runtime of $O(n^{2.5})$. The key idea in \cite{cls18} is that one can
design a stochastic interior point method such that each step only
need to multiply $P$ with a vector of density $\tilde{O}(\sqrt{n})$.
This bypasses the $n^{2.5}$ bottleneck.

In this paper, we do not have this issue because we only need to compute
$RPz$ which is much cheaper than $Pz$. We summarize why it suffices to maintain
$RP$ throughout the algorithm. In general, for interior point method,
the vector $z$ is roughly an unit vector and since $P$ is an orthogonal
projection, we have $\|Pz\|_{2}=O(1)$. One simple insight we have
is that if we multiply a random $\sqrt{n}\times n$ matrix $R$ with
values $\pm\frac{1}{\sqrt{n}}$ by $Pz$, we have $\|RPz\|_{\infty}=\tilde{O}(\frac{1}{\sqrt{n}})$
(Lemma \ref{lem:sketch_vector}). Since there are $\tilde{O}(\sqrt{n})$ iterations in
interior point method, the total error is roughly $\tilde{O}(1)$
in a correctly reweighed $\ell_{\infty}$ norm. In Section \ref{sec:robust_central_path}, we
showed that this is exactly what interior point method needs for convergence.
Furthermore, we note that though each step needs to use a fresh random
matrix $R_{l}$ of size $\sqrt{n}\times n$, the random matrices
$[ R_{1}^\top ; R_{2}^\top ; \cdots ; R_{T}^\top ]^\top$ we need can all fit into $\tilde{O}(n)\times n$
budget. Therefore, throughout the algorithm, we simply need to maintain
the matrix $[R_{1}^\top ; R_{2}^\top ; \cdots ; R_{T}^\top ]^\top P$ which can be done with total
cost $\tilde{O}(n^{\omega})$ across all iterations using idea similar
to \cite{cls18}.

The only reason the data structure looks complicated is that when
the block matrix $W$ changes in different location in $\sqrt{W}A^{\top}(AWA^{\top})^{-1}A\sqrt{W}$,
we need to update the matrix $[R_{1};R_{2};\cdots;R_{T}]P$ appropriately.
This gives us few simple cases to handle in the algorithm and in the
proof. For the intuition on how to maintain $P$ under $W$ change,
see \cite[Section 2.2 and 5.1]{cls18}.

\subsection{Fast rectangular matrix multiplication}
Given two size $n \times n$ matrices, the time of multiplying them is $n^{2.81} < n^3$ by applying Strassen's original algorithm \cite{s69}. The current best running time takes $n^{\omega}$ time where $\omega < 2.373$ \cite{w12,l14}. One natural extension of multiplying two square matrices is multiplying two rectangular matrices. What is the running time of multiplying one $n \times n^a$ matrix with another $n^a \times n$ matrix? Let $\alpha$ denote the largest upper bound of $a$ such that multiplying two rectangular matrices takes $n^{2+o(1)}$ time. The $\alpha$ is called the dual exponent of matrix multiplication, and the state-of-the-art result is $\alpha = 0.31$ \cite{gu18}. We use the similar idea as \cite{cls18} to delay the low-rank update when the rank is small.

\section{Preliminaries}\label{sec:preli}

Given a vector $x \in \R^n$ and $m$ compact convex sets $K_1 \subset \R^{n_1}, K_2 \subset \R^{n_2}, \cdots, K_m \subset \R^{n_m}$ with $\sum_{i=1}^m n_i = n$. We use $x_i$ to denote the $i$-th block of $x$, then $x \in \prod_{i=1}^m K_i$ if $x_i \in K_i$, $\forall i \in [m]$. 

We say a block diagonal matrix $A \in \oplus_{i=1}^m \R^{n_i \times n_i}$ if $A$ can be written as
\begin{align*}
A = \begin{bmatrix}
A_1 & & & \\
& A_2 & & \\
& & \ddots &  \\
& & & A_m
\end{bmatrix}
\end{align*}
where $A_1 \in \R^{n_1 \times n_1}$, $A_2 \in \R^{n_2 \times n_2}$, and $A_m \in \R^{n_m \times n_m}$. For a matrix $A$, we use $\| A \|_F$ to denote its Frobenius norm and use $\| A \|$ to denote its operator norm. There are some trivial facts $\| A B \|_2 \leq \| A \|_2 \cdot \| B\|_2$ and $\| A B \|_F \leq \| A \|_F \cdot \| B \|_2$.

For notation convenience, we assume the number of variables $n\geq10$ and there are no redundant constraints. In particular, this implies that the constraint matrix $A$ is full rank.

For a positive integer $n$, let $[n]$ denote the set $\{1,2,\cdots,n\}$.

For any function $f$, we define $\wt{O}(f)$ to be $f\cdot \log^{O(1)}(f)$. In addition to $O(\cdot)$ notation, for two functions $f,g$, we use the shorthand $f\lesssim g$ (resp. $\gtrsim$) to indicate that $f\leq C g$ (resp. $\geq$) for some absolute constant $C$. For any function $f$, we use $\mathrm{dom} f$ to denote the domain of function $f$.

For a vector $v$, We denote $\| v \|$ as the standard Euclidean norm of $v$ and for a symmetric PSD matrix $A$, we let $\|v\|_A = (v^\top A v)^{1/2}$. For a convex function $f(x)$ that is clear from context, we denote $\|v\|_x = \|v\|_{\nabla^2 f(x)}$ and $\|v \|_x^* = \|v\|_{\nabla^2f(x)^{-1}}$.

\global\long\def\ideal{\mathrm{ideal}}

\section{Robust Central Path}\label{sec:robust_central_path_intro}

In this section we show how to move move efficiently from $x(1)$
to $x(\epsilon)$ for some tiny $\epsilon$ by staying on a robust version of the central path. Because we are maintaining values that are slightly off-center, we show that our analysis still goes through despite $\ell_\infty$ perturbations on the order of $O(1/\text{poly}\log(n))$.

\subsection{Newton Step}

To follow the path $x(t)$, we consider the optimality condition of
(\ref{eq:panelized_problem}):
\begin{align*}
s/t+\nabla\phi(x) & =0,\\
Ax & =b,\\
A^{\top}y+s & =c
\end{align*}
where $\nabla\phi(x)=(\nabla\phi_{1}(x_{1}),\nabla\phi_{2}(x_{2}),\cdots,\nabla\phi_{m}(x_{m}))$.
To handle the error incurred in the progress, we consider the perturbed
central path
\begin{align*}
s/t+\nabla\phi(x) & =\mu,\\
Ax & =b,\\
A^{\top}y+s & =c
\end{align*}
where $\mu$ represent the error between the original central path
and our central path. Each iteration, we decrease $t$ by a certain
factor. It may increase the error term $\mu$. Therefore, we need
a step to decrease the norm of $\mu$. The Newton method to move $\mu$
to $\mu+h$ is given by
\begin{align*}
\frac{1}{t}\cdot\delta_{s}^{\ideal}+\nabla^{2}\phi(x)\cdot\delta_{x}^{\ideal}= & ~h,\\
A\delta_{x}^{\ideal}= & ~0,\\
A^{\top}\delta_{y}^{\ideal}+\delta_{s}^{\ideal}= & ~0
\end{align*}
where $\nabla^{2}\phi(x)$ is a block diagonal matrix with the $i$-th
block is given by $\nabla^{2}\phi_{i}(x_{i})$. Letting $W=(\nabla^{2}\phi(x))^{-1}$,
we can solve this: 
\begin{align*}
\delta_{y}^{\ideal}= & ~-t\cdot\left(AWA^{\top}\right)^{-1}AWh,\\
\delta_{s}^{\ideal}= & ~t\cdot A^{\top}\left(AWA^{\top}\right)^{-1}AWh,\\
\delta_{x}^{\ideal}= & ~Wh-WA^{\top}\left(AWA^{\top}\right)^{-1}AWh.
\end{align*}

We define projection matrix $P\in\R^{n\times n}$ as follows 
\[
P=W^{1/2}A^{\top}\left(AWA^{\top}\right)^{-1}AW^{1/2}
\]
and then we rewrite them 
\begin{align}
\delta_{x}^{\ideal}= & ~W^{1/2}(I-P)W^{1/2}\delta_{\mu},\label{eq:x_ideal}\\
\delta_{s}^{\ideal}= & ~tW^{-1/2}PW^{1/2}\delta_{\mu}.\label{eq:s_ideal}
\end{align}
One standard way to analyze the central path is to measure the error
by $\|\mu\|_{\nabla^{2}\phi(x)^{-1}}$ and uses the step induced by
$h=-\mu$. One can easily prove that if $\|\mu\|_{\nabla^{2}\phi(x)^{-1}}<\frac{1}{10}$,
one step of Newton step decreases the norm by a constant factor. Therefore,
one can alternatively decrease $t$ and do a Newton step to follow
the path.

\subsection{Robust Central Path Method}

In this section, we develop a central path method that is robust under
certain $\ell_{\infty}$ perturbations. Due to the $\ell_{\infty}$
perturbation, we measure the error $\mu$ by a soft max instead of
the $\ell_{2}$ type potential:

\begin{definition}For each $i\in[m]$, let $\mu_{i}^{t}(x,s)\in\R^{n_{i}}$
and $\gamma_{i}^{t}(x,s)\in\R$ be defined as follows: 
\begin{align}
\mu_{i}^{t}(x,s)= & ~s_{i}/t+\nabla\phi_{i}(x_{i}),\label{eq:def_mu_i_t}\\
\gamma_{i}^{t}(x,s)= & ~\|\mu_{i}^{t}(x,s)\|_{\nabla^{2}\phi_{i}(x_{i})^{-1}},\label{eq:def_gamma_i_t}
\end{align}
and we define potential function $\Phi$ as follows:
\begin{align*}
\Phi^{t}(x,s)= & \ \sum_{i=1}^{m}\exp(\lambda\gamma_{i}^{t}(x,s))
\end{align*}
where $\lambda=O(\log m)$.\end{definition}
The {\it robust central path} is the region $(x,s)$ that satisfies $\Phi^t(x,s) \leq O(m)$. To run our convergence argument, we will be setting $\lambda$ appropriately so that staying on the robust central path will guarantee a $\ell_\infty$ bound on $\gamma$. Then, we will show how to maintain
$\Phi^{t}(x,s)$ to be small throughout the algorithm while decreasing $t$, always staying on the robust central path. This is broken into a two step analysis: the progress step (decreasing $t$) and the centering step (moving $x,s$ to decrease $\gamma$). 

It is important to note that to follow the robust central path, we no longer pick the standard Newton direction by setting $h = -\mu$. To explain how we pick our centering step, suppose we can move $\mu\to\mu+h$
arbitrarily with the only restriction on the distance $\|h\|_{\nabla^{2}\phi(x)^{-1}}=\alpha$.
Then, the natural step would be
\[
h=\arg\min_{\|h\|_{\nabla^{2}\phi(x)^{-1}}=\alpha}\left\langle \nabla f(\mu(x,s)),h\right\rangle 
\]
where $f(\mu)=\sum_{i=1}^{m}\exp(\lambda\|\mu\|_{\nabla^{2}\phi_{i}(x_{i})^{-1}})$.
Note that 
\[
\nabla f(\mu^{t}(x,s))_{i}=\lambda\exp(\lambda\gamma_{i}^{t}(x,s))/\gamma_{i}^{t}(x,s)\cdot\nabla^{2}\phi_{i}(x_{i})^{-1}\mu_{i}^{t}(x,s).
\]
Therefore, the solution for the minimization problem is
\[
h_{i}^{\ideal}=-\alpha\cdot c_{i}^{t}(x,s)^{\ideal}\mu_{i}^{t}(x,s)\in\R^{n_{i}},
\]
where $\mu_{i}^{t}(x,s)\in\R^{n_{i}}$ is defined as Eq.~\eqref{eq:def_mu_i_t}
and $c_{i}^{t}(x,s)\in\R$ is defined as 
\[
c_{i}^{t}(x,s)^{\ideal}=\frac{\exp(\lambda\gamma_{i}^{t}(x,s))/\gamma_{i}^{t}(x,s)}{(\sum_{i=1}^{m}\exp(2\lambda\gamma_{i}^{t}(x,s)))^{1/2}}.
\]
Eq.~\eqref{eq:x_ideal} and Eq.~\eqref{eq:s_ideal} gives the corresponding
ideal step on $x$ and $s$.

Now, we discuss the perturbed version of this algorithm. Instead of
using the exact $x$ and $s$ in the formula of $h$, we use a $\ov x$
which is approximately close to $x$ and a $\overline{s}$ which is
close to $s$. Precisely, we have
\begin{equation}
h_{i}=-\alpha\cdot c_{i}^{t}(\overline{x},\overline{s})\mu_{i}^{t}(\overline{x},\overline{s})\label{eq:def_h_i_t}
\end{equation}
where
\begin{equation}
c_{i}^{t}(x,s)=\begin{cases}
\frac{\exp(\lambda\gamma_{i}^{t}(x,s))/\gamma_{i}^{t}(x,s)}{(\sum_{i=1}^{m}\exp(2\lambda\gamma_{i}^{t}(x,s)))^{1/2}} & \text{if }\gamma_{i}^{t}(x,s)\geq96\sqrt{\alpha}\\
0 & \text{otherwise}
\end{cases}.\label{eq:def_c_i_t}
\end{equation}
Note that our definition of $c_i^{t}$ ensures that $c_i^{t}(x,s)\leq\frac{1}{96\sqrt{\alpha}}$
regardless of the value of $\gamma_i^t(x,s)$. This makes sure we do not move too much in
any coordinates and indeed when $\gamma_i^t$ is small, it is fine to set $c_i^t = 0$. Furthermore, for the formula on $\delta_{x}$ and $\delta_{s}$,
we use some matrix $\wt V$ that is close to $(\nabla^{2}\phi(x))^{-1}$.
Precisely, we have
\begin{align}
\delta_{x}= & ~\wt V^{1/2}(I-\wt P)\wt V^{1/2}h,\label{eq:delta_x}\\
\delta_{s}= & ~t\cdot\wt V^{-1/2}\wt P~\wt V^{1/2}h.\label{eq:delta_s}
\end{align}
where
\[
\wt P=\wt V^{1/2}A^{\top}(A\wt VA^{\top})^{-1}A\wt V^{1/2}.
\]

%{[}TODO: Make the following algorithm box formal.{]} %% Zhao: done in combine section
Here we give a quick summary of our algorithm. (The more detailed of our algorithm can be found in Algorithm~\ref{alg:robust_central_path} and \ref{alg:main} in Section~\ref{sec:combine}.)
\begin{itemize}
\item $\textsc{RobustIPM}$($A,b,c,\phi,\delta$)
\begin{itemize}
\item $\lambda=2^{16}\log(m)$, $\alpha=2^{-20}\lambda^{-2}$, $\kappa=2^{-10}\alpha$.
\item $\delta = \min(\frac{1}{\lambda},\delta)$.
\item $\nu = \sum_{i=1}^m \nu_{i}$ where $\nu_{i}$ are the self-concordant parameters
of $\phi_{i}$.
\item Modify the convex problem and obtain an initial $x$ and $s$ according
to Lemma \ref{lem:feasible_LP}.
\item $t=1$.
\item While $t>\frac{\delta^{2}}{4\nu}$
\begin{itemize}
\item Find $\overline{x}$ and $\overline{s}$ such that $\|\ov x_{i}-x_{i}\|_{\ov x_{i}}<\alpha$
and $\|\ov s_{i}-s_{i}\|_{\ov x_{i}}^{*}<t\alpha$ for all $i$.
\item Find $\wt V_{i}$ such that $(1-\alpha)(\nabla^{2}\phi_{i}(\overline{x}_{i}))^{-1}\preceq\wt V_{i}\preceq(1+\alpha)(\nabla^{2}\phi_{i}(\overline{x}_{i}))^{-1}$
for all $i$.
\item Compute $h=-\alpha\cdot c_{i}^{t}(\overline{x},\overline{s})\mu_{i}^{t}(\overline{x},\overline{s})$
where
\[
c_{i}^{t}(\overline{x},\overline{s})=\begin{cases}
\frac{\exp(\lambda\gamma_{i}^{t}(\overline{x},\overline{s}))/\gamma_{i}^{t}(\overline{x},\overline{s})}{(\sum_{i=1}^{m}\exp(2\lambda\gamma_{i}^{t}(\overline{x},\overline{s})))^{1/2}} & \text{if }\gamma_{i}^{t}(\overline{x},\overline{s})\geq96\sqrt{\alpha}\\
0 & \text{otherwise}
\end{cases}.
\]
and $\mu_{i}^{t}(\overline{x},\overline{s})=~\overline{s}_{i}/t+\nabla\phi_{i}(\overline{x}_{i})$
and $\gamma_{i}^{t}(\overline{x},\overline{s})=\|\mu_{i}^{t}(\overline{x},\overline{s})\|_{\nabla^{2}\phi_{i}(\overline{x}_{i})^{-1}}$
\item Let $\wt P=\wt V^{1/2}A^{\top}(A\wt VA^{\top})^{-1}A\wt V^{1/2}$.
\item Compute $\delta_{x}=~\wt V^{1/2}(I-\wt P)\wt V^{1/2}h$ and $\delta_{s}=~t\cdot\wt V^{-1/2}\wt P~\wt V^{1/2}h.$
\item Move $x\leftarrow x+\delta_{x}$, $s\leftarrow s+\delta_{s}$.
\item $t^{\new}=(1-\frac{\kappa}{\sqrt{\nu}})t$.
\end{itemize}
\item Return an approximation solution of the convex problem according to
Lemma \ref{lem:feasible_LP}.
\end{itemize}
\end{itemize}
\begin{theorem}[Robust Interior Point Method]\label{thm:roubst_IPM} Consider a convex problem
$\min_{Ax=b,x\in\prod_{i=1}^{m}K_{i}}c^{\top}x$ where $K_{i}$ are
compact convex sets. For each $i \in [m]$, we are given a $\nu_{i}$-self
concordant barrier function $\phi_{i}$ for $K_{i}$. Let $\nu = \sum_{i=1}^m \nu_i$. Also, we are
given $x^{(0)}=\arg\min_{x}\sum_{i=1}^m \phi_{i}(x_{i})$. Assume that
\begin{enumerate}
\item Diameter of the set: For any $x \in \prod_{i=1}^m K_{i}$, we have that $\|x\|_{2}\leq R$.
\item Lipschitz constant of the program: $\|c\|_{2}\leq L$.
\end{enumerate}
Then, the algorithm $\textsc{RobustIPM}$ finds a vector $x$ such
that 
\begin{align*}
c^{\top}x & \leq\min_{Ax=b,x\in\prod_{i=1}^m K_{i}}c^{\top}x+LR\cdot\delta,\\
\|Ax-b\|_{1} & \leq3\delta\cdot\left(R\sum_{i,j}|A_{i,j}|+\|b\|_{1}\right),\\
x & \in\prod_{i=1}^m K_{i}.
\end{align*}
in $O(\sqrt{\nu}\log^{2}m\log(\frac{\nu}{\delta}))$ iterations.

\end{theorem}

\begin{proof}Lemma \ref{lem:feasible_LP} shows that the initial
$x$ and $s$ satisfies
\[
\|s+\nabla\phi(x)\|_{x}^{*}\leq\delta\leq\frac{1}{ \lambda}
\]
where the last inequality is due to our step $\delta\leftarrow\min(\frac{1}{\lambda},\delta)$.
This implies that $\gamma_{i}^{1}(x,s)=\|s_{i}+\nabla\phi_{i}(x_{i})\|_{x_{i}}^{*}\leq\frac{1}{ \lambda}$
and hence $\Phi^{1}(x,s)\leq e\cdot m\leq80\frac{m}{\alpha}$ for
the initial $x$ and $s$. Apply Lemma \ref{lem:robust_central_path_part3}
repetitively, we have that $\Phi^{t}(x,s)\leq80\frac{m}{\alpha}$
during the whole algorithm. In particular, we have this at the end
of the algorithm. This implies that
\[
\|s_{i}+\nabla\phi_{i}(x_{i})\|_{x_{i}}^{*}\leq\frac{\log(80\frac{m}{\alpha})}{\lambda}\leq1
\]
at the end. Therefore, we can apply Lemma \ref{lem:apx_center_imply_gap}
to show that
\[
\left\langle c,x\right\rangle \leq\left\langle c,x^{*}\right\rangle +4t\nu\leq\left\langle c,x^{*}\right\rangle +\delta^{2}
\]
where we used the stop condition for $t$ at the end. Note that this
guarantee holds for the modified convex program. Since the error is
$\delta^{2}$, Lemma \ref{lem:feasible_LP} shows how to get an approximate
solution for the original convex program with error $LR\cdot\delta$.

The number of steps follows from the fact we decrease $t$ by $1-\frac{1}{\sqrt{\nu}\log^{2}m}$
factor every iteration.

\end{proof}

%\newpage
\addcontentsline{toc}{section}{References}
\bibliographystyle{alpha}%{colt2019}
\bibliography{ref}

\newpage
\appendix
\section*{Appendix}
\section{Robust Central Path}\label{sec:robust_central_path}

The goal of this section is to analyze robust central path. We provide an outline in Section~\ref{sec:robust_central_path_outline}. In Section~\ref{sec:robust_central_path_changes_in_mu_gamma}, we bound the changes in $\mu$ and $\gamma$. In Section~\ref{sec:robust_central_path_changes_in_xs}, we analyze the changes from $(x,x,s)$ to $(x^{\new},x,s^{\new})$. In Section~\ref{sec:robust_central_path_changes_in_z}, we analyze the changes from $(x^{\new},x,s^{\new})$ to $(x^{\new},x^{\new},s^{\new})$. We bound the changes in $t$ in Section~\ref{sec:robust_central_path_changes_in_t}. Finally, we analyze entire changes of potential function in Section~\ref{sec:robust_central_path_changes_in_potential}.

\subsection{Outline of Analysis}\label{sec:robust_central_path_outline}

\begin{table}
\begin{centering}
\begin{tabular}{|l|l|l|}
\hline 
Statement & Section & Parameters\tabularnewline
\hline 
Lemma~\ref{lem:changes_in_mu} & Section~\ref{sec:robust_central_path_changes_in_mu_gamma} & $\mu_{i}^{t}(x,s)\rightarrow\mu_{i}^{t}(x^{\new},s^{\new})$\tabularnewline
\hline 
Lemma~\ref{lem:changes_in_gamma} & Section~\ref{sec:robust_central_path_changes_in_mu_gamma} & $\gamma_{i}^{t}(x,x,s)\rightarrow\gamma_{i}^{t}(x^{\new},x,s^{\new})$\tabularnewline
\hline 
Lemma~\ref{lem:changes_in_xs} & Section~\ref{sec:robust_central_path_changes_in_xs} & $\Phi(x,x,s)\rightarrow\Phi(x^{\new},x,s^{\new})$\tabularnewline
\hline 
Lemma~\ref{lem:changes_in_z} & Section~\ref{sec:robust_central_path_changes_in_z} & $\Phi(x^{\new},x,s^{\new})\rightarrow\Phi(x^{\new},x^{\new},s^{\new})$\tabularnewline
\hline 
Lemma~\ref{lem:changes_in_t} & Section~\ref{sec:robust_central_path_changes_in_t} & $\Phi^{t}\rightarrow\Phi^{t^{\new}}$\tabularnewline
\hline 
Lemma~\ref{lem:robust_central_path_part3} & Section~\ref{sec:robust_central_path_changes_in_potential} & $\Phi^{t}(x,s)\rightarrow\Phi^{t^{\new}}(x^{\new},s^{\new})$\tabularnewline
\hline 
\end{tabular}
\par\end{centering}
\caption{Bounding the changes of different variables}
\end{table}

Basically, the main proof is just a simple calculation on how $\Phi^{t}(x,s)$
changes during 1 iteration. It could be compared to the proof of $\ell_\infty$ potential reduction arguments for the convergence of long-step interior point methods, although the main difficulty arises from the perturbations from stepping using $\ov{x},\ov{s}$ instead of $x,s$.

To organize the calculations, we note that the term $\gamma_{i}^{t}(x,s)=~\|\mu_{i}^{t}(x,s)\|_{\nabla^{2}\phi_{i}(x_{i})^{-1}}$
has two terms involving $x$, one in the $\mu$ term and one in the Hessian. Hence, we separate how different $x$
affect the potential by defining
\begin{align*}
\gamma_{i}^{t}(x,z,s) & =~\|\mu_{i}^{t}(x,s)\|_{\nabla^{2}\phi_{i}(z_{i})^{-1}},\\
\Phi^{t}(x,z,s) & =\ \sum_{i=1}^{m}\exp(\lambda\gamma_{i}^{t}(x,z,s)).
\end{align*}
One difference between our proof and standard $\ell_2$ proofs of interior point is that we assume the barrier function is decomposable. We define $\alpha_{i}=\|\delta_{x,i}\|_{\ov{x}_{i}}$
is the ``step'' size of the coordinate $i$. One crucial fact we
are using is that sum of squares of the step sizes is small.

\begin{lemma}\label{lem:alpha_i}
Let $\alpha$ denote the parameter in \textsc{RobustIPM}. For all $i\in [m]$, let $\alpha_{i}=\|\delta_{x,i}\|_{\ov{x}_{i}}$.
Then,
\[
\sum_{i=1}^m \alpha_{i}^{2}\leq4\alpha^{2}.
\]
\end{lemma}

\begin{proof}Note that
\[
\sum_{i=1}^m \alpha_{i}^{2}=\|\delta_{x}\|_{\ov{x}}^{2}=h^{\top}\wt V^{1/2}(I-\wt P)\wt V^{1/2}\nabla^{2}\phi(\ov{x})\wt V^{1/2}(I-\wt P)\wt V^{1/2}h.
\]
Since $(1-\alpha)(\nabla^{2}\phi_{i}(\ov{x}_{i}))^{-1}\preceq\wt V_{i}\preceq(1+\alpha)(\nabla^{2}\phi_{i}(\ov{x}_{i}))^{-1}$,
we have that
\[
(1-\alpha)(\nabla^{2}\phi(\ov{x}))^{-1}\preceq\wt V\preceq(1+\alpha)(\nabla^{2}\phi(\ov{x}))^{-1}.
\]
Using $\alpha\leq\frac{1}{10000}$, we have that
\[
\sum_{i=1}^m \alpha_{i}^{2} \leq2h^{\top}\wt V^{1/2}(I-\wt P)(I-\wt P)\wt V^{1/2}h\leq2h^{\top}\wt Vh
\]
where we used that $I-\wt P$ is an orthogonal projection at the end.
Finally, we note that
\begin{align*}
h^{\top}\wt Vh 
\leq & ~ 2 \sum_{i=1}^m \|h_{i}\|_{\ov{x}_{i}}^{*2}\\
 = & ~ 2\alpha^{2}\sum_{i=1}^m c_{i}^{t}(\ov{x},\ov{s})^{2}\|\mu_{i}^{t}(\ov{x},\ov{s})\|_{\ov{x}_{i}}^{*2}\\
\leq & ~ 2\alpha^{2}\sum_{i=1}^m \frac{ \exp (2 \lambda \gamma_i^t ( \ov{x}, \ov{s} ) ) /  \gamma_i^t ( \ov{x} , \ov{s} )^2 }{ \sum_{i=1}^m \exp( 2 \lambda \gamma_i^t ( \ov{x}, \ov{s} ) )^{1/2} }  \|\mu_{i}^{t}(\ov{x},\ov{s})\|_{\ov{x}_{i}}^{*2}\\
 = & ~ 2\alpha^{2}\frac{ \sum_{i=1}^m \exp(2\lambda\gamma_{i}^{t}(\ov{x},\ov{s}))}{\sum_{i=1}^m \exp(2\lambda\gamma_{i}^{t}(\ov{x},\ov{s}))} \\
 = & ~ 2\alpha^{2} 
\end{align*}
where the second step follows from definition of $h_i$ \eqref{eq:def_h_i_t}, the third step follows from definition $c_i^t$ \eqref{eq:def_c_i_t}, the fourth step follows from definition of $\gamma_i^t$ \eqref{eq:def_gamma_i_t}.

Therefore, putting it all together, we can show
\begin{align*}
\sum_{i=1}^m \alpha_i^2 \leq 4 \alpha^2.
\end{align*}

\end{proof}

\subsection{Changes in $\mu$ and $\gamma$}\label{sec:robust_central_path_changes_in_mu_gamma}

We provide basic lemmas that bound changes in $\mu, \gamma$ due to the centering steps.

\begin{lemma}[Changes in $\mu$]\label{lem:changes_in_mu}For all
$i\in[m]$, let 
\[
\mu_{i}^{t}(x^{\new},s^{\new})=\mu_{i}^{t}(x,s)+h_{i}+\epsilon_{i}^{(\mu)}.
\]
Then, $\|\epsilon_{i}^{(\mu)}\|_{ x_{i}}^{*}\leq10\alpha\cdot\alpha_{i}$.
\end{lemma} \begin{proof}

Let $x^{(u)}=ux^{\new}+(1-u)x$ and $\mu_{i}^{\new}=\mu_{i}^{t}(x^{\new},s^{\new})$.
The definition of $\mu$ (\ref{eq:def_mu_i_t}) shows that
\begin{align*}
\mu_{i}^{\new} & =\mu_{i}+\frac{1}{t}\delta_{s,i}+\nabla\phi_{i}(x_{i}^{\new})-\nabla\phi_{i}(x_{i})\\
 & =\mu_{i}+\frac{1}{t}\delta_{s,i}+\int_{0}^{1}\nabla^{2}\phi_{i}(x_{i}^{(u)})\delta_{x,i}\, \d u\\
 & =\mu_{i}+\frac{1}{t}\delta_{s,i}+\nabla^{2}\phi_{i}(\ov{x}_{i})\delta_{x,i}+\int_{0}^{1}\left(\nabla^{2}\phi_{i}(x_{i}^{(u)})-\nabla^{2}\phi_{i}(\ov{x}_{i})\right)\delta_{x,i} \,  \d u.
\end{align*}
By the definition of $\delta_{x}$ and $\delta_{s}$ (\ref{eq:delta_x})
and (\ref{eq:delta_s}), we have that $\frac{1}{t}\delta_{s,i}+\wt V_{i}^{-1}\delta_{x,i}=h_{i}.$
Hence, we have
\[
\mu_{i}^{\new}=\mu_{i}+h_{i}+\epsilon_{i}^{(\mu)}
\]
where 
\begin{equation}
\epsilon_{i}^{(\mu)}=\int_{0}^{1}\left(\nabla^{2}\phi_{i}(x_{i}^{(u)})-\nabla^{2}\phi_{i}(\ov{x}_{i})\right)\delta_{x,i}\, \d u+(\nabla^{2}\phi_{i}(\ov{x}_{i})-\wt V_{i}^{-1})\delta_{x,i}.\label{eq:change_mu_1}
\end{equation}

To bound $\epsilon_{i}^{(\mu)}$, we note that
\begin{align*}
\|x_{i}^{(t)}-\ov x_{i}\|_{\ov x_{i}}\leq\|x_{i}^{(t)}-x_{i}\|_{\ov x_{i}}+\|x_{i}-\ov x_{i}\|_{\ov x_{i}}\leq\|\delta_{x,i}\|_{\ov x_{i}}+\alpha = \alpha_i + \alpha  \leq3\alpha
\end{align*}
where the first step follows from triangle inequality, the third step follows from definition of $\alpha_i$ (Lemma~\ref{lem:alpha_i}), and the last step follows from $\alpha_i \leq 2 \alpha$ (Lemma~\ref{lem:alpha_i}). 

Using $\alpha\leq\frac{1}{100}$,
Theorem~\ref{thm:hessiansc} shows that
\[
-7\alpha\cdot\nabla^{2}\phi_{i}(\ov x_{i})\preceq\nabla^{2}\phi_{i}(x_{i}^{(u)})-\nabla^{2}\phi_{i}(\ov{x}_{i})\preceq7\alpha\cdot\nabla^{2}\phi_{i}(\ov x_{i}).
\]
Equivalently, we have
\[
(\nabla^{2}\phi_{i}(x_{i}^{(u)})-\nabla^{2}\phi_{i}(\ov{x}_{i})) \cdot (\nabla^{2}\phi_{i}(\ov x_{i}))^{-1} \cdot (\nabla^{2}\phi_{i}(x_{i}^{(u)})-\nabla^{2}\phi_{i}(\ov{x}_{i}))\preceq(7\alpha)^{2}\cdot\nabla^{2}\phi_{i}(\ov x_{i}).
\]
Using this, we have
\begin{align}
\left\| \int_{0}^{1}\left(\nabla^{2}\phi_{i}(x_{i}^{(u)})-\nabla^{2}\phi_{i}(\ov{x}_{i})\right)\delta_{x,i} \, \d u \right\|_{\ov{x}_{i}}^{*} & \leq\int_{0}^{1} \left\|\left(\nabla^{2}\phi_{i}(x_{i}^{(u)})-\nabla^{2}\phi_{i}(\ov{x}_{i})\right)\delta_{x,i} \right\|_{\ov{x}_{i}}^{*} \,\d u\nonumber \\
 & \leq7\alpha\|\delta_{x,i}\|_{\ov{x}_{i}}=7\alpha\cdot\alpha_{i},\label{eq:change_mu_2}
\end{align}
where the last step follows from definition of $\alpha_i$ (Lemma~\ref{lem:alpha_i}).

For the other term in $\epsilon_{i}^{(\mu)}$, we note that
\[
(1-2\alpha) \cdot (\nabla^{2}\phi_{i}(\ov{x}_{i}))\preceq\wt V_{i}^{-1}\preceq(1+2\alpha) \cdot (\nabla^{2}\phi_{i}(\ov{x}_{i})).
\]
Hence, we have
\begin{equation}
\left\| (\nabla^{2}\phi_{i}(\ov{x}_{i})-\wt V_{i}^{-1})\delta_{x,i} \right\|_{\ov{x}_{i}}^{*}\leq2\alpha\|\delta_{x,i}\|_{\ov{x}_{i}}=2\alpha\cdot\alpha_{i}.\label{eq:change_mu_3}
\end{equation}
Combining (\ref{eq:change_mu_1}), (\ref{eq:change_mu_2}) and (\ref{eq:change_mu_3}),
we have
\[
\|\epsilon_{i}^{(\mu)}\|_{\ov x_{i}}^{*}\leq9\alpha\cdot\alpha_{i}.
\]
Finally, we use the fact that $x_{i}$ and $\ov{x}_{i}$ are
$\alpha$ close and hence again by self-concordance, $\|\epsilon_{i}^{(\mu)}\|_{x_{i}}^{*}\leq 10 \alpha\cdot\alpha_{i}$.

\end{proof}

Before bounding the change of $\gamma$, we first prove a helper lemma:

\begin{lemma}\label{lem:changes_in_gamma_help}For all $i\in[m]$,
we have
\[
\|\mu_{i}^{t}(x,s)-\mu_{i}^{t}(\ov{x},\ov{s})\|_{x_{i}}^{*}\leq4\alpha.
\]
\end{lemma}

\begin{proof}Note that
\[
\|\mu_{i}^{t}(x,s)-\mu_{i}^{t}(\ov{x},\ov{s})\|_{\ov x_{i}}^{*}=\frac{1}{t}\|s_{i}-\ov{s}_{i}\|_{\ov x_{i}}^{*}+\|\nabla\phi_{i}(x_{i})-\nabla\phi_{i}(\ov{x}_{i})\|_{\ov x_{i}}^{*}.
\]
For the first term, we have $\|s_{i}-\ov{s}_{i}\|_{\ov x_{i}}^{*}\leq t\alpha$.

For the second term, let $x_{i}^{(u)}=ux_{i}+(1-u)\ov{x}_{i}$.
Since $x_{i}$ is close enough to $\ov{x}_{i}$, Theorem~\ref{thm:hessiansc}
shows that $\nabla^{2}\phi_{i}(x_{i}^{(u)})\preceq2\cdot\nabla^{2}\phi_{i}(\ov{x}_{i})$.
Hence, we have 
\begin{align*}
\| \nabla\phi_{i}(x_{i})-\nabla\phi_{i}(\ov{x}_{i}) \|_{\ov x_{i}}^{*}= \left\|\int_{0}^{1}\nabla^{2}\phi_{i}(x_{i}^{(u)})\cdot(x_{i}-\ov{x}_{i})du \right\|_{\ov x_{i}}^{*}\leq2\|x_{i}-\ov{x}_{i}\|_{\ov x_{i}}=2\alpha.
\end{align*}
Hence, we have $\|\mu_{i}^{t}(x,s)-\mu_{i}^{t}(\ov{x},\ov{s})\|_{\ov{x}_{i}}^{*}\leq3\alpha$
and using again $x_{i}$ is close enough to $\ov{x}_{i}$ to
get the final result.

\end{proof}

\begin{lemma}[Changes in $\gamma$]\label{lem:changes_in_gamma}For
all $i\in[m]$, let 
\[
\gamma_{i}^{t}(x^{\new},x,s^{\new})\leq(1-\alpha\cdot c_{i}^{t}(\ov{x},\ov{s}))\gamma_{i}^{t}(x,x,s)+\epsilon_{i}^{(\gamma)}.
\]
then $\epsilon_{i}^{(\gamma)}\leq10\alpha\cdot(\alpha c_{i}^{t}(\ov{x},\ov{s})+\alpha_{i})$.
Furthermore, we have $|\gamma_{i}^{t}(x^{\new},x,s^{\new})-\gamma_{i}^{t}(x,x,s)|\leq3\alpha.$
\end{lemma}

\begin{proof}

For the first claim, Lemma \ref{lem:changes_in_mu}, the definition
of $\gamma$ (\ref{eq:def_gamma_i_t}), $h$ (\ref{eq:def_h_i_t})
and $c$ (\ref{eq:def_c_i_t}) shows that 
\begin{align*}
\gamma_{i}^{t}(x^{\new},x,s^{\new}) & =\|\mu_{i}^{t}(x,s)+h_{i}+\epsilon_{i}^{(\mu)}\|_{x_{i}}^{*}\\
 & =\|(1-\alpha\cdot c_{i}^{t}(\ov{x},\ov{s}))\mu_{i}^{t}(x,s)+\epsilon_{i}\|_{x_{i}}^{*}
\end{align*}
where $\epsilon_{i}=\alpha\cdot c_{i}^{t}(\ov{x},\ov{s})(\mu_{i}^{t}(x,s)-\mu_{i}^{t}(\ov{x},\ov{s}))+\epsilon_{i}^{(\mu)}.$

From the definition of $c_{i}^{t}$, we have that $c_{i}^{t}\leq\frac{1}{96\sqrt{\alpha}}\leq\frac{1}{\alpha}$
and hence $0\leq1-\alpha\cdot c_{i}^{t}(\ov{x},\ov{s})\leq1$.
Therefore, we have
\begin{align}
\gamma_{i}^{t}(x^{\new},x,s^{\new})\leq & (1-\alpha\cdot c_{i}^{t}(\ov{x},\ov{s}))\gamma_{i}^{t}(x,x,s)+\|\epsilon_{i}\|_{x_{i}}^{*}.\label{eq:gamma_change_gamma}
\end{align}
Now, we bound $\|\epsilon_{i}\|_{x_{i}}^{*}$:
\begin{align}
\|\epsilon_{i}\|_{x_{i}}^{*} & \leq\alpha c_{i}^{t}(\ov{x},\ov{s})\cdot\|\mu_{i}^{t}(x,s)-\mu_{i}^{t}(\ov{x},\ov{s})\|_{x_{i}}^{*}+\|\epsilon_{i}^{(\mu)}\|_{x_{i}}^{*}\nonumber \\
 & \leq4\alpha^{2}c_{i}^{t}(\ov{x},\ov{s})+10\alpha\cdot\alpha_{i}\label{eq:gamma_change_eps}
\end{align}
where we used Lemma \ref{lem:changes_in_gamma_help} and Lemma \ref{lem:changes_in_mu}
at the end.

For the second claim, we have
\[
\left|\gamma_{i}^{t}(x^{\new},x,s^{\new})-\gamma_{i}^{t}(x,x,s)\right|\leq\|h_{i}+\epsilon_{i}^{(\mu)}\|_{x_{i}}^{*}\leq2\alpha+10\alpha\cdot\alpha_{i}
\]
where we used (\ref{eq:gamma_change_eps}) and that $\|h_{i}\|_{x_{i}}^{*}\leq2\|h\|_{\ov{x}}^{*}\leq2\alpha$.
From Lemma \ref{lem:alpha_i} and that $\alpha\leq\frac{1}{10000}$,
we have $10\alpha\cdot\alpha_{i}\leq20\alpha^{2}\leq\alpha$.

\end{proof}

\subsection{Movement from $(x,x,s)$ to $(x^{\protect\new},x,s^{\protect\new})$}\label{sec:robust_central_path_changes_in_xs}

In the previous section, we see that $\gamma_i$ will be expected to decrease by a factor of $\alpha \cdot c_i^t$ up to some small perturbations. We show that our potential $\Phi^t$ will therefore decrease significantly.  

\begin{lemma}[Movement along the first and third parameters]\label{lem:changes_in_xs}Assume
that $\gamma_{i}^{t}(x,x,s)\leq1$ for all $i$. We
have 
\[
\Phi^{t}(x^{\new},x,s^{\new})\leq\Phi^{t}(x,x,s)-\frac{\alpha\lambda}{5} \left(\sum_{i=1}^m \exp(2\lambda\gamma_{i}^{t}(\ov{x},\ov{s})) \right)^{1/2}+\sqrt{m}\lambda\cdot\exp(192\lambda\sqrt{\alpha}).
\]
Note that $\gamma$ is a function that has three inputs. We use $\gamma(x,s)$ to denote $\gamma(x,x,s)$ for simplicity. 
%\Zhao{Richard: please check this sentence!}
\end{lemma}

\begin{proof}Let $\Phi^{\new}=\Phi^{t}(x^{\new},x,s^{\new})$, $\Phi=\Phi^{t}(x,x,s)$,
\[
\gamma^{(u)}=u\gamma_{i}^{t}(x^{\new},x,s^{\new})+(1-u)\gamma_{i}^{t}(x,x,s).
\]
Then, we have that
\[
\Phi^{\new}-\Phi= \sum_{i=1}^m (e^{\lambda\gamma_{i}^{(1)}}-e^{\lambda\gamma_{i}^{(0)}})=\lambda \sum_{i=1}^m e^{\lambda\gamma_{i}^{(\zeta)}}(\gamma_{i}^{(1)}-\gamma_{i}^{(0)})
\]
for some $0\leq\zeta\leq1$. Let $v_{i}=\gamma_{i}^{(1)}-\gamma_{i}^{(0)}$.
Lemma \ref{lem:changes_in_gamma} shows that 
\[
v_{i}\leq-\alpha\cdot c_{i}^{t}(\ov{x},\ov{s})\cdot\gamma_{i}^{t}(x,x,s)+\epsilon_{i}^{(\gamma)}=-\alpha\cdot c_{i}^{t}(\ov{x},\ov{s})\cdot\gamma_{i}^{(0)}+\epsilon_{i}^{(\gamma)}
\]
and hence
\begin{align}
\frac{\Phi^{\new}-\Phi}{\lambda} & \leq-\alpha\sum_{i=1}^m c_{i}^{t}(\ov{x},\ov{s})\cdot\gamma_{i}^{(0)}\exp(\lambda\gamma_{i}^{(\zeta)}) + \sum_{i=1}^m \epsilon_{i}^{(\gamma)}\exp(\lambda\gamma_{i}^{(\zeta)}).\label{eq:change_xs_main}
\end{align}

To bound the first term in (\ref{eq:change_xs_main}), we first relate
$\gamma_{i}^{(0)}$, $\gamma_{i}^{(\zeta)}$ and $\gamma_{i}^{t}(\ov{x},\ov{s})$ .
Lemma \ref{lem:changes_in_gamma} shows that
\begin{equation}
|\gamma_{i}^{(0)}-\gamma_{i}^{(\zeta)}|\leq|\gamma_{i}^{(0)}-\gamma_{i}^{(1)}|\leq3\alpha.\label{eq:gamma_0_zeta}
\end{equation}
Finally, we have
\begin{align}\label{eq:gamma_r}
\left|\gamma_{i}^{t}(\ov{x},\ov{s})-\gamma_{i}^{(0)}\right| 
= & ~ \left|\gamma_{i}^{t}(\ov{x},\ov{x},\ov{s})-\gamma_{i}^{t}(x,x,s)\right| \notag \\
\leq & ~ \left|\gamma_{i}^{t}(\ov{x},\ov{x},\ov{s})-\gamma_{i}^{t}(x,\ov{x},s)\right|+\left|\gamma_{i}^{t}(x,\ov{x},s)-\gamma_{i}^{t}(x,x,s)\right|\nonumber \\
\leq & ~ \|\mu_{i}^{t}(x,s)-\mu_{i}^{t}(\ov{x},\ov{s})\|_{\ov{x}_{i}}^{*}+\left|\|\mu_{i}^{t}(x,s)\|_{\ov{x}_{i}}^{*}-\|\mu_{i}^{t}(x,s)\|_{x_{i}}^{*}\right|\nonumber \\
\leq & ~ 2\|\mu_{i}^{t}(x,s)-\mu_{i}^{t}(\ov{x},\ov{s})\|_{x_{i}}^{*} +\left|\|\mu_{i}^{t}(x,s)\|_{\ov{x}_{i}}^{*}-\|\mu_{i}^{t}(x,s)\|_{x_{i}}^{*}\right|\nonumber \\
\leq & ~ 2\|\mu_{i}^{t}(x,s)-\mu_{i}^{t}(\ov{x},\ov{s})\|_{x_{i}}^{*}+2\alpha\|\mu_{i}^{t}(x,s)\|_{x_{i}}^{*}\nonumber \\
\leq & ~ 8\alpha+2\alpha = 10\alpha 
\end{align}
where the first step follows from definition, the second and third step follows from triangle inequality, the fourth step follows from $ \|\mu_{i}^{t}(x,s)-\mu_{i}^{t}(\ov{x},\ov{s})\|_{ \ov{x}_{i} }^{*} \leq 2\|\mu_{i}^{t}(x,s)-\mu_{i}^{t}(\ov{x},\ov{s})\|_{x_{i}}^{*}$, the fifth step follows from self-concordance, the sixth step follows from Lemma~\ref{lem:changes_in_gamma_help} and that $\|\mu_{i}^{t}(x,s)\|_{x_{i}}^{*}=\gamma_{i}^{t}(x,x,s)\leq1$ for all $i$
%\Zhao{Where do we need it?}.

Using (\ref{eq:gamma_0_zeta}) and (\ref{eq:gamma_r}), we have
\begin{align}
 & ~ \sum_{i=1}^m c_{i}^{t}(\ov{x},\ov{s})\cdot\gamma_{i}^{(0)}\exp(\lambda\gamma_{i}^{(\zeta)} )\nonumber \\
 = & ~ \sum_{i=1}^m c_{i}^{t}(\ov{x},\ov{s})\cdot\gamma_{i}^{(0)}\exp(\lambda \gamma_i^t ( \ov{x}, \ov{s} ) - \lambda \gamma_i^t ( \ov{x}, \ov{s} ) + \lambda \gamma_i^{(0)} - \lambda \gamma_i^{(0)} + \lambda \gamma_i^{(\zeta)} )\nonumber \\
\geq & ~ \sum_{i=1}^m c_{i}^{t}(\ov{x},\ov{s})\cdot\gamma_{i}^{(0)}\exp(\lambda\gamma_{i}^{t}(\ov{x},\ov{s})-13\lambda\alpha)\nonumber \\
\geq & ~ \frac{1}{2}\sum_{i=1}^m c_{i}^{t}(\ov{x},\ov{s})\cdot\gamma_{i}^{(0)}\exp(\lambda\gamma_{i}^{t}(\ov{x},\ov{s}))\nonumber \\
\geq & ~ \frac{1}{2}\sum_{i=1}^m c_{i}^{t}(\ov{x},\ov{s})\cdot\gamma_{i}^{t}(\ov{x},\ov{s})\exp(\lambda\gamma_{i}^{t}(\ov{x},\ov{s}))-3\alpha\sum_{i=1}^m c_{i}^{t}(\ov{x},\ov{s})\exp(\lambda\gamma_{i}^{t}(\ov{x},\ov{s})).\label{eq:change_xs_main_1}
\end{align}
where the third step follows from $ \exp(-13 \lambda \alpha) \geq 1/2 $, and the last step follows from \eqref{eq:gamma_0_zeta}.

For the first term in (\ref{eq:change_xs_main_1}), we have
\begin{align*}
& ~ \sum_{i=1}^m c_{i}^{t}(\ov{x},\ov{s})\cdot\gamma_{i}^{t}(\ov{x},\ov{s})\exp(\lambda\gamma_{i}^{t}(\ov{x},\ov{s})) \\
= & ~ \sum_{\gamma_{i}^{t}(\ov{x},\ov{s})\geq96\sqrt{\alpha}}\frac{\exp(2\lambda\cdot\gamma_{i}^{t}(\ov{x},\ov{s}))}{(\sum_{i=1}^m \exp(2\lambda\gamma_{i}^{t}(\ov{x},\ov{s})))^{1/2}}\\
= & ~ \sum_{i = 1}^m \frac{\exp(2\lambda\cdot\gamma_{i}^{t}(\ov{x},\ov{s}))}{(\sum_{i=1}^m \exp(2\lambda\gamma_{i}^{t}(\ov{x},\ov{s})))^{1/2}} - \sum_{\gamma_{i}^{t}(\ov{x},\ov{s}) < 96\sqrt{\alpha}}\frac{\exp(2\lambda\cdot\gamma_{i}^{t}(\ov{x},\ov{s}))}{(\sum_{i=1}^m \exp(2\lambda\gamma_{i}^{t}(\ov{x},\ov{s})))^{1/2}} \\
\geq & ~ \left( \sum_{i=1}^m \exp(2\lambda\gamma_{i}^{t}(\ov{x},\ov{s})) \right)^{1/2}-\frac{m\cdot\exp(192\lambda\cdot\sqrt{\alpha})}{(\sum_{i=1}^m \exp(2\lambda\gamma_{i}^{t}(\ov{x},\ov{s})))^{1/2}}.
\end{align*}
So, if $\sum_{i=1}^m \exp(2\lambda\gamma_{i}^{t}(\ov{x},\ov{s}))\geq m\cdot\exp(192\lambda\cdot\sqrt{\alpha})$,
we have 
\[
\sum_{i=1}^m c_{i}^{t}(\ov{x},\ov{s})\cdot\gamma_{i}^{t}(\ov{x},\ov{s})\exp(\lambda\gamma_{i}^{t}(\ov{x},\ov{s}))\geq \left(\sum_{i=1}^m \exp(2\lambda\gamma_{i}^{t}(\ov{x},\ov{s})) \right)^{1/2}-\sqrt{m}\cdot\exp(192\lambda\sqrt{\alpha}).
\]
Note that if $\sum_{i=1}^m \exp(2\lambda\gamma_{i}^{t}(\ov{x},\ov{s}))\leq m\cdot\exp(192\lambda\cdot\sqrt{\alpha})$,
this is still true because left hand side is lower bounded by $0$.
For the second term in (\ref{eq:change_xs_main_1}), we have
\begin{align*}
\sum_{i=1}^m c_{i}^{t}(\ov{x},\ov{s})\exp(\lambda\gamma_{i}^{t}(\ov{x},\ov{s})) 
= & ~ \sum_{\gamma_{i}^{t}(\ov{x},\ov{s})\geq96\sqrt{\alpha}}\frac{\exp(\lambda\cdot\gamma_{i}^{t}(\ov{x},\ov{s}))/\gamma_{i}^{t}(\ov{x},\ov{s})}{(\sum_{i=1}^m \exp(2\lambda\gamma_{i}^{t}(\ov{x},\ov{s})))^{1/2}}\exp(\lambda\gamma_{i}^{t}(\ov{x},\ov{s}))\\
\leq & ~ \frac{1}{96\sqrt{\alpha}} \sum_{\gamma_{i}^{t}(\ov{x},\ov{s})\geq96\sqrt{\alpha}}\frac{\exp(2\lambda\cdot\gamma_{i}^{t}(\ov{x},\ov{s}))}{(\sum_{i=1}^m \exp(2\lambda\gamma_{i}^{t}(\ov{x},\ov{s})))^{1/2}} \\
\leq & ~ \frac{1}{96\sqrt{\alpha}} \sum_{i=1}^m \frac{\exp(2\lambda\cdot\gamma_{i}^{t}(\ov{x},\ov{s}))}{(\sum_{i=1}^m \exp(2\lambda\gamma_{i}^{t}(\ov{x},\ov{s})))^{1/2}} \\
= & ~ \frac{1}{96\sqrt{\alpha}} \left(\sum_{i=1}^m \exp(2\lambda\gamma_{i}^{t}(\ov{x},\ov{s})) \right)^{1/2}.
\end{align*}
where the second step follows $\frac{1}{\gamma_i^t( \ov{x} , \ov{s} )} \leq \frac{1}{ 96 \sqrt{\alpha} }$, and the third step follows from each term in the summation is non-negative.

Combining the bounds for both first and second term in (\ref{eq:change_xs_main_1}),
we have
\begin{align}
\sum_{i=1}^m c_{i}^{t}(\ov{x},\ov{s})\cdot\gamma_{i}^{(0)}\exp(\lambda\gamma_{i}^{(\zeta)})
\geq & ~ \frac{1}{2}\left( \left(\sum_{i=1}^m \exp(2\lambda\gamma_{i}^{t}(\ov{x},\ov{s})) \right)^{1/2}-\sqrt{m}\cdot\exp(192\lambda\sqrt{\alpha})\right)\nonumber \\
 & ~ -\frac{3\alpha}{96\sqrt{\alpha}} \left(\sum_{i=1}^m \exp(2\lambda\gamma_{i}^{t}(\ov{x},\ov{s})) \right)^{1/2}\nonumber \\
\geq & ~ \frac{2}{5} \left(\sum_{i=1}^m \exp(2\lambda\gamma_{i}^{t}(\ov{x},\ov{s})) \right)^{1/2}-\sqrt{m}\cdot\exp(192\lambda\sqrt{\alpha}).\label{eq:change_xs_main_1_result}
\end{align}
where the last step follows from $\frac{1}{2} - \frac{3\alpha}{96\sqrt{\alpha}} \geq \frac{1}{2} - \frac{3}{96} = \frac{45}{96} \geq \frac{2}{5}$.

For the second term in (\ref{eq:change_xs_main}), we note that $|\gamma_{i}^{(\zeta)}-\gamma_{i}^{t}(\ov{x},\ov{s})|\leq13\alpha\leq\frac{1}{2\lambda}$
by (\ref{eq:gamma_0_zeta}) and (\ref{eq:gamma_r}). Hence,
\[
\sum_{i=1}^m \epsilon_{i}^{(\gamma)}\exp(\lambda\gamma_{i}^{(\zeta)})\leq2\sum_{i=1}^m \epsilon_{i}^{(\gamma)}\exp(\lambda\gamma_{i}^{t}(\ov{x},\ov{s})).
\]
Now, we use $\epsilon_{i}^{(\gamma)}\leq10\alpha\cdot(\alpha c_{i}^{t}(\ov{x},\ov{s})+\alpha_{i})$
(Lemma \ref{lem:changes_in_gamma}) to get
\begin{align*}
\sum_{i=1}^m \epsilon_{i}^{(\gamma)}\exp(\lambda\gamma_{i}^{(\zeta)}) 
\leq & ~ 20 \alpha \sum_{i=1}^m (\alpha c_{i}^{t}(\ov{x},\ov{s})+\alpha_{i})\cdot\exp(\lambda\gamma_{i}^{t}(\ov{x},\ov{s}))\\
\leq & ~ 20 \alpha \left(\sum_{i=1}^m (\alpha c_{i}^{t}(\ov{x},\ov{s})+\alpha_{i})^{2} \right)^{1/2} \left(\sum_{i=1}^m \exp(2\lambda\gamma_{i}^{t}(\ov{x},\ov{s})) \right)^{1/2}.
\end{align*}
where the last step follows from Cauchy-Schwarz inequality.

Note that by using Cauchy-Schwarz, 
\begin{align*}
\left( \sum_{i=1}^m (\alpha c_{i}^{t}(\ov{x},\ov{s})+\alpha_{i})^{2} \right)^{1/2} 
\leq & ~ \alpha \left(\sum_{i=1}^m c_{i}^{t}(\ov{x},\ov{s})^{2} \right)^{1/2} + \left(\sum_{i=1}^m \alpha_{i}^{2} \right)^{1/2}\\
\leq & ~ \alpha\cdot\frac{1}{96\sqrt{\alpha}}+2\alpha\leq\frac{\sqrt{\alpha}}{90}.
\end{align*}
where we used the definition of $c_{i}^{t}$, Lemma \ref{lem:alpha_i}
and $\alpha\leq\frac{1}{2^{24}}$. Together, we conclude 
\begin{equation}
\sum_{i=1}^m \epsilon_{i}^{(\gamma)}\exp(\lambda\gamma_{i}^{(\zeta)})\leq\frac{1}{5}\alpha\left( \sum_{i=1}^m \exp(2\lambda\gamma_{i}^{t}(\ov{x},\ov{s})) \right)^{1/2}.\label{eq:change_xs_main_2}
\end{equation}

Combining (\ref{eq:change_xs_main_1_result}) and (\ref{eq:change_xs_main_2})
to (\ref{eq:change_xs_main}) gives
\begin{align*}
\frac{\Phi^{\new}-\Phi}{\lambda} 
\leq & ~ - \frac{2}{5} \alpha \left(\sum_{i=1}^m \exp(2\lambda\gamma_{i}^{t}(\ov{x},\ov{s})) \right)^{1/2}+\sqrt{m}\cdot\exp(192\lambda\sqrt{\alpha})+\frac{1}{5}\alpha \left(\sum_{i=1}^m \exp(2\lambda\gamma_{i}^{t}(\ov{x},\ov{s})) \right)^{1/2}\\
 = & ~ - \frac{1}{5} \alpha \left(\sum_{i=1}^m \exp(2\lambda\gamma_{i}^{t}(\ov{x},\ov{s})) \right)^{1/2}+\sqrt{m}\cdot\exp(192\lambda\sqrt{\alpha}).
\end{align*}
where the last step follows from merging the first term with the third term.
\end{proof}

\subsection{Movement from $(x^{\protect\new},x,s^{\protect\new})$
to $(x^{\protect\new},x^{\protect\new},s^{\protect\new})$}\label{sec:robust_central_path_changes_in_z}

Next, we must analyze the potential change when we change the second
term.

\begin{lemma}[Movement along the second parameter]\label{lem:changes_in_z}
Assume that $\|\gamma^{t}(x,x,s)\|_{\infty} \leq 1$. Then we have
\[
\Phi^{t}(x^{\new},x^{\new},s^{\new})\leq\Phi^{t}(x^{\new},x,s^{\new})+12\alpha(\|\gamma^{t}(x,x,s)\|_{\infty}+3\alpha)\lambda\left(\sum_{i=1}^{m}\exp(2\lambda\gamma_{i}^{t}(x,x,s))\right)^{1/2}.
\]
\end{lemma}

\begin{proof} We can upper bound $\Phi^{t}(x^{\new},x^{\new},s^{\new})$
as follows 
\begin{align*}
\Phi^{t}(x^{\new},x^{\new},s^{\new}) & =\sum_{i=1}^{m}\exp(\lambda\gamma_{i}^{t}(x^{\new},x^{\new},s^{\new}))\\
 & \leq\sum_{i=1}^{m}\exp(\lambda\gamma_{i}^{t}(x^{\new},x,s^{\new})(1+2\alpha_{i})).
\end{align*}
where the second step follows from $\gamma_{i}^{t}(x^{\new},x^{\new},s^{\new})\leq\gamma_{i}^{t}(x^{\new},x,s^{\new})\cdot(1+2\alpha_{i})$
by self-concordance (Theorem~\ref{thm:hessiansc}) and $\|x_{i}^{\new}-x_{i}\|_{x_{i}}\leq2\|x_{i}^{\new}-x_{i}\|_{\ov{x}_{i}}\leq2\alpha_{i}$.

Now, by Lemma~\ref{lem:changes_in_gamma}, we note that $\gamma_{i}^{t}(x^{\new},x,s^{\new})\leq\gamma_{i}^{t}(x,x,s)+3\alpha\leq1+3\alpha$
and that $\alpha\leq\frac{1}{100\lambda}$. Hence, by a simple taylor expansion, we have
\begin{align*}
 & \Phi^{t}(x^{\new},x^{\new},s^{\new})\\
\leq & \sum_{i=1}^{m}\exp(\lambda\gamma_{i}^{t}(x^{\new},x,s^{\new}))+3\sum_{i=1}^{m}\alpha_{i}\exp(\lambda\gamma_{i}^{t}(x^{\new},x,s^{\new}))\gamma_{i}^{t}(x^{\new},x,s^{\new}).
\end{align*}
Finally, we bound the last term by
\begin{align*}
 & \sum_{i=1}^{m}\exp(\lambda\gamma_{i}^{t}(x^{\new},x,s^{\new}))\gamma_{i}^{t}(x^{\new},x,s^{\new})\alpha_{i}\\
\leq & ~ \sum_{i=1}^{m}\exp(\lambda\gamma_{i}^{t}(x,x,s)+3\lambda\alpha)(\gamma_{i}^{t}(x,x,s)+3\alpha)\alpha_{i}\\
\leq & ~ 2(\|\gamma^{t}(x,x,s)\|_{\infty}+3\alpha)\sum_{i=1}^{m}\exp(\lambda\gamma_{i}^{t}(x,x,s))\alpha_{i}\\
\leq & ~ 2(\|\gamma^{t}(x,x,s)\|_{\infty}+3\alpha) \left( \sum_{i=1}^{m}\exp( 2 \lambda\gamma_{i}^{t}(x,x,s)) \right)^{1/2} \left( \sum_{i=1}^m \alpha_{i}^2 \right)^{1/2} \\
\leq & ~ 4\alpha(\|\gamma^{t}(x,x,s)\|_{\infty}+3\alpha)\left(\sum_{i=1}^{m}\exp(2\lambda\gamma_{i}^{t}(x,x,s))\right)^{1/2},
\end{align*}
where the first step follows from $\lambda \gamma_i^t(x^{\new},x,s^{\new}) \leq \exp( \lambda \gamma_i^t(x,x,s) + 3 \lambda \alpha )$, the second step follows $\exp(3 \lambda \alpha) \leq 2$, the third step follows from Cauchy-Schwarz inequality, the last step follows from $\sum_{i=1}^m \alpha_{i}^{2} \leq 4\alpha^{2}$.

\end{proof}

\subsection{Movement of $t$}\label{sec:robust_central_path_changes_in_t}

Lastly, we analyze the effect of setting $t\to t^{\new}$.

\begin{lemma}[Movement in $t$]\label{lem:changes_in_t}For any
$x,s$ such that $\gamma_{i}^{t}(x,s)\leq1$ for all $i$, let $t^{\new}=\left(1-\frac{\kappa}{\sqrt{ \nu }}\right)t$ where $\nu = \sum_{i=1}^m \nu_i$,
we have
\[
\Phi^{t^{\new}}(x,s)\leq\Phi^{t}(x,s)+10\kappa\lambda\left(\sum_{i=1}^{m}\exp(2\lambda\gamma_{i}^{t}(x,s))\right)^{1/2}.
\]
\end{lemma} \begin{proof}
Note that 
\begin{align*}
\gamma_{i}^{t^{\new}}(x,s)= & ~\left\Vert \frac{s}{t^{\new}}+\nabla\phi_{i}(x_{i})\right\Vert _{x_{i}}^{*}\\
= & ~\left\Vert \frac{s}{t(1-\kappa/\sqrt{ \nu })}+\nabla\phi_{i}(x_{i})\right\Vert _{x_{i}}^{*}\\
\leq & ~(1+2\kappa/\sqrt{\nu})\gamma_{i}^{t}(x,s)+2\|(\kappa/\sqrt{\nu})\nabla\phi_{i}(x_{i})\|_{x_{i}}^{*}\\
\leq & ~(1+2\kappa/\sqrt{\nu})\gamma_{i}^{t}(x,s)+3\kappa\sqrt{\nu_{i}}/\sqrt{\nu}\\
\leq & \ \gamma_{i}^{t}(x,s)+5\kappa\sqrt{\nu_{i}}/\sqrt{\nu}
\end{align*}
where the first step follows from definition, the second step follows
from $t^{\new}=t(1-\kappa/\sqrt{ \nu })$, the second last step
follows from the fact that our barriers are $\nu_{i}$-self-concordant
and the last step used $\gamma_{i}^{t}(x,s)\leq1$ and $\nu_{i}\geq1$.
Using that $5\kappa\leq\frac{1}{10\lambda}$ and $\gamma_{i}^{t}(x,s)\leq1$,
we have by simple taylor expansion,
\begin{align*}
\Phi^{t^{\new}}(x,s) \leq & ~ \sum_{i=1}^m \exp(\lambda\gamma_{i}^{t}(x,s))+2\lambda\sum_{i=1}^m \exp(\lambda\gamma_{i}^{t}(x,s))\left(5\kappa\sqrt{ \nu_i / \nu }\right)\\
= & ~ \sum_{i=1}^m \exp(\lambda\gamma_{i}^{t}(x,s))+10\kappa\lambda\sum_{i=1}^m \exp(\lambda\gamma_{i}^{t}(x,s))\left( \sqrt{ \nu_i / \nu }\right)\\
 \leq & ~ \sum_{i=1}^m \exp(\lambda\gamma_{i}^{t}(x,s))+ 10\kappa\lambda\left(\sum_{i=1}^{m}\exp(2\lambda\gamma_{i}^{t}(x,s))\right)^{1/2} \left( \sum_{i=1}^m \frac{\nu_i}{\nu} \right)^{1/2} \\
 = & ~ \sum_{i=1}^m \exp(\lambda\gamma_{i}^{t}(x,s))+ 10\kappa\lambda\left(\sum_{i=1}^{m}\exp(2\lambda\gamma_{i}^{t}(x,s))\right)^{1/2},
\end{align*}
where the third step follows from Cauchy-Schwarz, and the last step follows from $\sum_{i=1}^m \nu_i = \nu$.
\end{proof}

\subsection{Potential Maintenance}\label{sec:robust_central_path_changes_in_potential}

Putting it all together, we can show that our potential $\Phi^t$ can be maintained to be small throughout our algorithm.

\begin{lemma}[Potential Maintenance]\label{lem:robust_central_path_part3}If
$\Phi^{t}(x,s)\leq80\frac{m}{\alpha}$, then
\[
\Phi^{t^{\new}}(x^{\new},s^{\new})\leq \left( 1-\frac{\alpha\lambda}{40\sqrt{m}} \right) \Phi^{t}(x,s)+\sqrt{m}\lambda\cdot\exp(192\lambda\sqrt{\alpha}).
\]
In particularly, we have $\Phi^{t^{\new}}(x^{\new},s^{\new})\leq80\frac{m}{\alpha}$.

\end{lemma}

\begin{proof}Let 
\[
\zeta(x,s)=\left(\sum_{i=1}^{m}\exp(2\lambda\gamma_{i}^{t}(x,s))\right)^{1/2}.
\]
By combining our previous lemmas, 
\begin{align}
 & \Phi^{t^{\new}}(x^{\new},s^{\new})\nonumber \\
\leq & \Phi^{t}(x^{\new},s^{\new})+10\kappa\lambda\cdot\zeta(x^{\new},s^{\new})\nonumber \\
\leq & \Phi^{t}(x^{\new},x,s^{\new})+12\alpha\lambda(\|\gamma^{t}(x,s)\|_{\infty}+3\alpha)\cdot\zeta(x,s)+10\kappa\lambda\cdot\zeta(x^{\new},s^{\new})\nonumber \\
\leq & \Phi^{t}(x,x,s)-\frac{\alpha\lambda}{5}\zeta(\ov{x},\ov{s})+\sqrt{m}\lambda\cdot\exp(192\lambda\sqrt{\alpha})\nonumber \\
 & +12\alpha\lambda(\|\gamma^{t}(x,s)\|_{\infty}+3\alpha)\cdot\zeta(x,s)+10\kappa\lambda\cdot\zeta(x^{\new},s^{\new})\label{eq:part3_phi}
\end{align}
where the first step follows from Lemma~\ref{lem:changes_in_t},
the second step follows from Lemma~\ref{lem:changes_in_z}, and the
last step follows from Lemma~\ref{lem:changes_in_xs}. We note that
in all lemma above, we used that fact that $\|\gamma^{t}\|_{\infty}\leq1$
(for different combination of $x$, $\ov{x}$, $x^{\new}$,
$s$, $\ov{s}$, $s^{\new}$) which we will show later.

We can upper bound $\gamma_{i}^{t}(x^{\new},s^{\new})$ in the following sense,
\begin{equation}
\gamma_{i}^{t}(x^{\new},s^{\new})\leq\gamma_{i}^{t}(x^{\new},x,s^{\new})+2\alpha\leq\gamma_{i}^{t}(x,x,s)+5\alpha.\label{eq:part3_gamma_new}
\end{equation}
where the first step follows from self-concordance and $\gamma_i \leq 1$, the second step follows from Lemma~\ref{lem:changes_in_gamma}. %\Zhao{It is not obvious}.

Hence, since $\zeta$ changes multiplicatively when $\gamma$ changes additively, $\zeta(x^{\new},s^{\new})\leq2\zeta(x,s).$ %\Zhao{why?}

Lemma \ref{lem:changes_in_gamma_help} shows that $\|\mu_{i}^{t}(x,s)-\mu_{i}^{t}(\ov{x},\ov{s})\|_{x_{i}}^{*}\leq4\alpha$
and hence
\begin{align}
\zeta(\ov{x},\ov{s}) 
\geq & ~ \frac{2}{3}\left(\sum_{i=1}^{m}\exp(2\lambda\gamma_{i}^{t}(\ov{x},x,\ov{s}))\right)^{1/2}\nonumber \\
\geq & ~ \frac{2}{3}\left(\sum_{i=1}^{m}\exp(2\lambda\gamma_{i}^{t}(x,x,s)-8\alpha\lambda)\right)^{1/2}\nonumber \\
\geq & ~ \frac{1}{2}\zeta(x,s).\label{eq:part3_gamma_bar}
\end{align}
Combining (\ref{eq:part3_gamma_new}) and (\ref{eq:part3_gamma_bar})
into (\ref{eq:part3_phi}) gives
\begin{align*}
 & ~ \Phi^{t^{\new}}(x^{\new},s^{\new})\\
\geq & ~ \Phi^{t}(x,s)+\left(12\alpha\lambda(\| \gamma^{t}(x,s) \|_{\infty}+3\alpha)+20\kappa\lambda-\frac{\alpha\lambda}{10}\right)\cdot\zeta(x,s)+\sqrt{m}\lambda\cdot\exp(192\lambda\sqrt{\alpha})\\
\geq & ~ \Phi^{t}(x,s)+\left(12\alpha\lambda\| \gamma^{t}(x,s) \|_{\infty}-\frac{\alpha\lambda}{20}\right)\cdot\zeta(x,s)+\sqrt{m}\lambda\cdot\exp(192\lambda\sqrt{\alpha})
\end{align*}
where the last step follows from $\kappa\leq\frac{\alpha}{1000}$ and $\alpha\leq\frac{1}{10000}$.

Finally, we need to bound $\|\gamma^{t}(x,s)\|_{\infty}$. The
bound for other $\|\gamma^{t}\|_{\infty}$, i.e. for different combination of $x$, $\ov{x}$, $x^{\new}$,
$s$, $\ov{s}$, $s^{\new}$, are similar. We note
that 
\[
\Phi^{t}(x,s)\leq80\frac{m}{\alpha}
\]
implies that $\|\gamma^{t}(x,s)\|_{\infty}\leq\frac{\log ( 80\frac{m}{\alpha} )}{\lambda}$.
Hence, by our choice of $\lambda$ and $\alpha$, we have that $\lambda\geq480\log(80\frac{m}{\alpha})$
and hence
\[
12\alpha\lambda\|\gamma^{t}(x,s)\|_{\infty}\leq\frac{\alpha\lambda}{40}.
\]
Finally, using $\Phi^{t}(x,s)\leq\sqrt{m}\cdot\zeta(x,s)$, we have
\begin{align*}
\Phi^{t^{\new}}(x^{\new},s^{\new}) & \geq\Phi^{t}(x,s)-\frac{\alpha\lambda}{40}\zeta(x,s)+\sqrt{m}\lambda\cdot\exp(192\lambda\sqrt{\alpha})\\
 & \geq \left(1-\frac{\alpha\lambda}{40\sqrt{m}}\right)\Phi^{t}(x,s)+\sqrt{m}\lambda\cdot\exp(192\lambda\sqrt{\alpha}).
\end{align*}
Since $\lambda\leq\frac{1}{400\sqrt{\alpha}}$, we have $\Phi^{t}(x,s)\leq80\frac{m}{\alpha}$
implies $\Phi^{t^{\new}}(x^{\new},s^{\new})\leq80\frac{m}{\alpha}$.

\end{proof}

\section{Central Path Maintenance}\label{sec:central_path_maintenance}

\begin{table}
\begin{center}
  \begin{tabular}{ | l | l | l | l | l | l | }
    \hline
    Name & Type & Statement & Algorithm & Input & Output \\ \hline
    \textsc{Initialize} & public & Lemma~\ref{lem:central_path_maintenance_initialize} & Alg.~\ref{alg:central_path_maintenance_main} & $A,x,s,\ov{W},\epsilon_{mp},a,b$ & $\emptyset$ \\ \hline
    \textsc{Update} & public & Lemma~\ref{lem:central_path_maintenance_update} & Alg.~\ref{alg:central_path_maintenance_update} & $\ov{W}$ & $\emptyset$ \\ \hline
    \textsc{FullUpdate} & private & Lemma~\ref{lem:central_path_maintenance_fullupdate} & Alg.~\ref{alg:central_path_maintenance_fullupdate} & $\ov{W}$ & $\emptyset$ \\ \hline
    \textsc{PartialUpdate} & private & Lemma~\ref{lem:central_path_maintenance_partialupdate} & Alg.~\ref{alg:central_path_maintenance_update} & $\ov{W}$ & $\emptyset$ \\ \hline
    \textsc{Query} & public & Lemma~\ref{lem:central_path_maintenance_query} & Alg.~\ref{alg:central_path_maintenance_main} & $\emptyset$ & $\ov{x},\ov{s}$ \\ \hline
    \textsc{MultiplyMove} & public & Lemma~\ref{lem:central_path_maintenance_multiply_move} & Alg.~\ref{alg:central_path_maintenance_multiply_move} & $h,t$ & $\emptyset$ \\ \hline
    \textsc{Multiply} & private & Lemma~\ref{lem:central_path_maintenance_multiply} & Alg.~\ref{alg:central_path_maintenance_multiply_move} & $h,t$ & $\emptyset$ \\ \hline
    \textsc{Move} & private & Lemma~\ref{lem:central_path_maintenance_move} & Alg.~\ref{alg:central_path_maintenance_multiply_move} & $\emptyset$ & $\emptyset$ \\ \hline
  \end{tabular}
\end{center}\caption{Summary of data structure \textsc{CentralPathMaintenance}}\label{tab:central_path_maintenance}
\end{table}

The goal of this section is to present a data-structure to perform our centering steps in $\widetilde{O}(n^{\omega-1/2})$ amortized time and prove a theoretical guarantee of it. The original idea of inverse maintenance is from Michael B. Cohen \cite{l17}, then \cite{cls18} used it to get faster running time for solving Linear Programs. 
 Because a simple matrix vector product would require $O(n^2)$ time, our speedup comes via a low-rank embedding that provides $\ell_\infty$ guarantees, which is unlike the sparse vector approach of \cite{cls18}. In fact, we are unsure if moving in a sparse direction $h$ can have sufficiently controlled noise to show convergence. Here, we give a stochastic version that is faster for dense direction $h$. 
 
\begin{theorem}[Central path maintenance]\label{thm:central_path_maintenance}
Given a full rank matrix $A \in \R^{d \times n}$ with $n \geq d$, a tolerance parameter $0 < \epsilon_{mp} < 1/4$ and a block diagonal structure $n = \sum_{i=1}^m n_i$. Given any positive number $a$ such $a \leq \alpha$ where $\alpha$ is the dual exponent of matrix multiplication. Given any linear sketch of size $b$, there is a randomized data structure \textsc{CentralPathMaintenance} (in Algorithm~\ref{alg:central_path_maintenance_main}, \ref{alg:central_path_maintenance_update}, \ref{alg:central_path_maintenance_multiply_move}) that approximately maintains the projection matrices
\begin{align*}
\sqrt{W} A^\top ( A W A^\top )^{-1} A \sqrt{W}
\end{align*}
for positive block diagonal psd matrix $W \oplus_i \R^{n_i \times n_i}$;
exactly implicitly maintains central path parameters $(x,s)$ and approximately explicitly maintains path parameters through the following five operations:

1. $\textsc{Initialize}(\ov{W}^{(0)}, \cdots)$ : Assume $\ov{W}^{(0)} \in \otimes_i \R^{n_i \times n_i}$. Initialize all the parameters in $O(n^{\omega})$ time.

2. $\textsc{Update}( \ov{W} )$ : Assume $\ov{W} \in \oplus_i \R^{n_i \times n_i} $. Output a block diagonal matrix $\wt{V} \oplus_i \R^{n_i \times n_i}$ such that
\begin{align*}
( 1 - \epsilon_{mp} ) \wt{v}_i \preceq \ov{w}_i \preceq ( 1 + \epsilon_{mp} ) \wt{v}_i .
\end{align*}

3. $\textsc{Query}()$ : Output $(\ov{x}, \ov{s})$ such that $\| \ov{x} - x \|_{\wt{V}^{-1}} \leq \epsilon_{mp}$ and $\| \ov{s} - s \|_{\wt{V}} \leq t \epsilon_{mp}$ where $t$ is the last $t$ used in \textsc{MultiplyMove}, where $\epsilon_{mp} = \alpha \log^2 (nT) \frac{ n^{1/4} }{ \sqrt{b} }$ and the success probability is $1-1/\poly(nT)$. This step takes $O(n)$ time.

4. $\textsc{MultiplyMove}(h,t)$ : It outputs nothing. It implicitly maintains:
\begin{align*}
x = x + \wt{V}^{1/2} ( I - \wt{P} ) \wt{V}^{1/2} h,
s = s + t \wt{V}^{-1/2} \wt{P} \wt{V}^{1/2} h.
\end{align*}
where $\wt{P} = \wt{V}^{1/2} A^\top ( A \wt{V} A^\top )^{-1} A \wt{V}^{1/2}$.
It also explicitly maintains $\ov{x}, \ov{s}$. Assuming $t$ is decreasing, each call takes $O( nb + n^{a \omega + o(1)} + n^a \| h \|_0 + n^{1.5} )$ amortized time.

Let $\ov{W}^{(0)}$ be the initial matrix and $\ov{W}^{(1)}, \cdots, \ov{W}^{(T)}$ be the (random) update sequence. Under the assumption that there is a sequence of matrix $W^{(0)}, \cdots, W^{(T)} \in \oplus_{i=1}^{m} \R^{n_i \times n_i}$ satisfies for all $k$
\begin{align*}
\left\| w_i^{-1/2} ( \ov{w}_i - w_i ) w_i^{-1/2} \right\|_F \leq & ~ \epsilon_{mp}, \\
\sum_{ i = 1 }^{ m } \left\| (w_i^{(k)})^{-1/2} ( \E [ w_i^{(k+1)} ] - w_i^{(k)} ) ( w_i^{(k)} )^{-1/2} \right\|_F^2 \leq & ~ C_1^2, \\
\sum_{ i = 1 }^{ m } \left( \E \left[ \left\| (w_i^{(k)})^{-1/2} ( w_i^{(k+1)}  - w_i^{(k)} ) ( w_i^{(k)} )^{-1/2} \right\|_F^2 \right] \right)^2 \leq & ~ C_2^2, \\
\left\| (w_i^{(k)})^{-1/2} ( w_i^{(k+1)}  - w_i^{(k)} ) ( w_i^{(k)} )^{-1/2}  \right\|_F \leq & ~ \frac{1}{4}.
\end{align*}
where $w_i^{(k)}$ is the $i$-th block of $W^{(k)}$, $\forall i \in [m]$.

Then, the amortized expected time per call of \textsc{Update}$(w)$ is
\begin{align*}
( C_1 / \epsilon_{mp} + C_2 / \epsilon_{mp}^2 ) \cdot ( n^{\omega - 1/2 + o(1)} + n^{2- a/2 +o(1)}) .
\end{align*}
\end{theorem}

\begin{remark}
For our algorithm, we have $C_1 = O(1/\log^2 n)$, $C_2 = O(1/\log^4 n)$ and $\epsilon_{mp}= O(1/\log^2 n)$. Note that the input of $\textsc{Update}$ $\ov{W}$ can move a lot. It is working as long as $\ov{W}$ is close to some $W$ that is slowly moving. In our application, our $W$ satisfies $C_1,C_2$ deterministically. We keep it for possible future applications. %not necessarily to be the same as $W$. It is working as long as $\ov{W}$ is close $W$. In our application, we actually don't need $\E$ in the assumption of sequence $W$ when we use this data structure. We keep it for fun.
\end{remark}

%\begin{remark}
%For FastJL matrices, there are mainly two categories : one category takes input sparsity time and it basically combines random Gaussian matrix with Count-Sketch matrix \cite{cw13,nn13}; the second category is subsampled Hadamard/Fourier matrices, the running time is dominated by fast Fourier transform. In our scenario, the second category is always better. For more detailed the trade-off between those two categories and applications of those matrices, we refer the readers to \cite{w14,cemmp15,psw17,swz17,swz19}.
%\end{remark}

\begin{algorithm}[!t]\caption{Central Path Maintenance Data Structure - Initial, Query, Move}\label{alg:central_path_maintenance_main} {%\footnotesize
\begin{algorithmic}[1]
\State {\bf datastructure} \textsc{CentralPathMaintenance} \Comment{Theorem~\ref{thm:central_path_maintenance}}
\\
\State {\bf private : members}
  \State \hspace{4mm} $\ov{W} \in \otimes_{i\in [m]} \R^{n_i \times n_i}$  \Comment{Target vector, $\ov{W}$ is $\epsilon_w$-close to $W$}
  \State \hspace{4mm} $V, \wt{V} \in \otimes_{i\in [m]} \R^{n_i \times n_i}$   \Comment{Approximate vector}
  \State \hspace{4mm} $A \in \R^{d \times n}$ \Comment{Constraints matrix}
  \State \hspace{4mm} $M \in \R^{n \times n}$ \Comment{Approximate Projection Matrix}
  \State \hspace{4mm} $\epsilon_{mp} \in (0,1/4)$ \Comment{Tolerance}
  \State \hspace{4mm} $a \in (0,\alpha]$ \Comment{Batch Size for Update ($n^a$)}
  \State \hspace{4mm} $b \in \mathbb{Z}_+$ \Comment{Sketch size of one sketching matrix}
  \State \hspace{4mm} $R \in \R^{n^{1+o(1)} \times n}$ \Comment{A list of sketching matrices}
  \State \hspace{4mm} $Q \in \R^{b \times n}$ \Comment{Sketched matrices}
  \State \hspace{4mm} $u_1 \in \R^n, F \in \R^{n \times n}, u_2 \in \R^n$ \Comment{Implicit representation of $x$, $x = u_1 + F \cdot u_2$}
  \State \hspace{4mm} $u_3 \in \R^n, G \in \R^{n \times n}, u_4 \in \R^n$ \Comment{Implicit representation of $s$, $s = u_3 + G \cdot u_4$}
   \State \hspace{4mm} $\ov{x}$, $\ov{s} \in \R^n$ \Comment{Central path parameters, maintain explicitly}
  \State \hspace{4mm} $l \in \mathbb{Z}_+$ \Comment{Randomness counter, $R_l \in \R^{b \times n}$}
  \State \hspace{4mm} $t^{\pre} \in \R_+$ \Comment{Tracking the changes of $t$}
\State {\bf end members}
\State
  \State {\bf public : procedure }\textsc{Initialize}{$(A,x,s,W,\epsilon_{mp},a,b)$} \Comment{Lemma~\ref{lem:central_path_maintenance_initialize}}
  \State \hspace{4mm} \Comment{parameters will {\it never change} after initialization}
  \State \hspace{4mm} $A \leftarrow A$, $a \leftarrow a$, $b \leftarrow b$, $\epsilon_{mp} \leftarrow \epsilon_{mp}$
  \State \hspace{4mm} \Comment{parameters will {\it still change} after initialization}
  \State \hspace{4mm} $\ov{W} \leftarrow W$, $V \leftarrow W$, $\wt{V} \leftarrow V$
  \State \hspace{4mm} Choose $R_{l} \in \R^{b \times n}$ to be sketching matrix, $\forall l \in [ \sqrt{n} ]$ \Comment{Lemma~\ref{lem:sketch_vector}}
  \State \hspace{4mm} $R \leftarrow [ R_1^\top , R_2^\top , \cdots ]^\top$ \Comment{Batch them into one matrix $R$}
  \State \hspace{4mm} $M \leftarrow A^\top ( A V A^\top)^{-1} A $, $Q \leftarrow R \sqrt{ \wt{V} } M$ \Comment{Initialize projection matrices}
  \State \hspace{4mm} $u_1 \leftarrow x$, $u_2 \leftarrow 0$, $u_3 \leftarrow s$, $u_4 \leftarrow 0$ \Comment{Initialize $x$ and $s$}
  \State \hspace{4mm} $\ov{x} \leftarrow x$, $\ov{s} \leftarrow s$
  \State \hspace{4mm} $l \leftarrow 1$
  \State {\bf end procedure}
  \State
  \State {\bf public : procedure }\textsc{Query}{$()$} \Comment{Lemma~\ref{lem:central_path_maintenance_query}}
    \State \hspace{4mm} \Return $(\ov{x}, \ov{s})$
  \State {\bf end procedure}
  \State
\State {\bf end datastructure}
\end{algorithmic}}
\end{algorithm}

\begin{algorithm}[!t]\caption{Central Path Maintenance Data Structure - Update and PartialUpdate}\label{alg:central_path_maintenance_update}{%\scriptsize
\begin{algorithmic}[1]
\State {\bf datastructure} \textsc{CentralPathMaintenance} \Comment{Theorem~\ref{thm:central_path_maintenance}}
  \State
  \State {\bf public : procedure }\textsc{Update}{$(\ov{W}^{\new})$} \Comment{Lemma~\ref{lem:central_path_maintenance_update}, $\ov{W}^{\new}$ is close to $W^{\new}$}
    \State \hspace{4mm} $\ov{y}_i \leftarrow v_i^{-1/2} \ov{w}^{\new}_i v_i^{-1/2}  - 1$, $\forall i \in [ m ]$
    \State \hspace{4mm} $r \leftarrow$ the number of indices $i$ such that $\| \ov{y}_i \|_F \geq \epsilon_{mp}$
    \State \hspace{4mm} {\bf if} {$r < n^a$} {\bf then}
    \State \hspace{8mm} \textsc{PartialUpdate}($\ov{W}^{\new}$)
    \State \hspace{4mm} {\bf else} %%%%%% above is if, and blow is else
    \State \hspace{8mm} \textsc{FullUpdate}($\ov{W}^{\new}$) \Comment{Algorithm~\ref{alg:central_path_maintenance_fullupdate}}
    \State \hspace{4mm} {\bf end if}
  %\State \hspace{4mm} \Return $\tilde{v}$
  \State {\bf procedure}
  
  \State
  \State {\bf private : procedure }\textsc{PartialUpdate}{$(\ov{W}^{\new})$} \Comment{Lemma~\ref{lem:central_path_maintenance_partialupdate}}
      \State \hspace{4mm} $\ov{W} \leftarrow \ov{W}^{\new}$
      \State \hspace{4mm} $\tilde{v}^{\new}_{i} \leftarrow \begin{cases} v_{i} & \text{if } (1-\epsilon_{mp}) v_i \preceq \ov{w}_i \preceq (1+\epsilon_{mp}) v_i\\ 
w_{i} & \text{otherwise} \end{cases}$
      \State \hspace{4mm} $F^{\new} \leftarrow F + ( ( \wt{V}^{\new} )^{ 1/2} - ( \wt{V} )^{ 1/2} ) M$ \Comment{only takes $n^{1+a}$ time, instead of $n^2$}
      \State \hspace{4mm} $G^{\new} \leftarrow G + ( ( \wt{V}^{\new} )^{-1/2} - ( \wt{V} )^{-1/2} ) M$
      \State \hspace{4mm} $u_1 \leftarrow u_1 + (F - F^{\new}) u_2$, $u_3 \leftarrow u_3 + (G - G^{\new}) u_4$
      \State \hspace{4mm} $F \leftarrow F^{\new}$, $G \leftarrow G^{\new}$
      \State \hspace{4mm} Let $\wh{S}$ denote the blocks where $\wt{V}$ and $\wt{V}^{\new}$ are different
      \State \hspace{4mm} $\ov{x}_{\wh{S}} \leftarrow (u_1)_{\wh{S}} + (F u_2)_{\wh{S}}$, $\ov{s}_{\wh{S}} \leftarrow (u_3)_{\wh{S}} + (G u_2)_{\wh{S}}$ \Comment{make sure $x$ and $\ov{x}$ are close, similarly for $s$ and $\ov{s}$}
  \State {\bf end procedure}
  \\
\State {\bf end datastructure}
\end{algorithmic}}
\end{algorithm}

\begin{algorithm}[!t]\caption{Central Path Maintenance Data Structure - Full Update}\label{alg:central_path_maintenance_fullupdate}{%\small
\begin{algorithmic}[1]
\State {\bf datastructure} \textsc{CentralPathMaintenance} \Comment{Theorem~\ref{thm:central_path_maintenance}}
  \State
  \State {\bf private : procedure }\textsc{FullUpdate}{$(\ov{W}^{\new})$} \Comment{Lemma~\ref{lem:central_path_maintenance_fullupdate}}
     \State \hspace{4mm} $\ov{y}_i \leftarrow v_i^{-1/2} \ov{w}^{\new}_i v_i^{-1/2}  - 1$, $\forall i \in [m]$
    \State \hspace{4mm} $r \leftarrow$ the number of indices $i$ such that $\| \ov{y}_i \|_F \geq \epsilon_{mp}$
    %\State \hspace{4mm} {\bf if} $r \geq 1$ {\bf do}
    \State \hspace{4mm} Let $\ov{\pi} : [m] \rightarrow [m]$ be a sorting permutation such that $\| \ov{y}_{ \ov{\pi} (i)} \|_F \geq \| \ov{y}_{ \ov{\pi} (i+1)} \|_F$
      \State \hspace{4mm} {\bf while} {$1.5 \cdot r < m$ and $\| \ov{y}_{ \ov{\pi} (1.5 r)} \|_F \geq ( 1 - 1 / \log m ) \| \ov{y}_{ \ov{\pi} (r)} \|_F $}
        \State \hspace{8mm} $r \leftarrow \min(\lceil 1.5 \cdot r  \rceil, m)$
      \State \hspace{4mm} {\bf end while}
      \State \hspace{4mm} $v^{\new}_{ \ov{\pi} (i)} \leftarrow \begin{cases} \ov{w}^{\new}_{ \ov{\pi} (i)} & i \in \{1,2,\cdots,r\} \\ v_{ \ov{\pi} (i)} & i \in \{r+1, \cdots, m\} \end{cases}$
      \\
      \Comment{Compute $M^{\new} = A^\top ( A V^{\new} A^\top )^{-1} A$ via Matrix Woodbury}
      \State \hspace{4mm} $\Delta \leftarrow V^{\new} - V$ \Comment{$\Delta \in \R^{n \times n}$ and $\| \Delta \|_0 = r$}
      \State \hspace{4mm} $\Gamma \leftarrow \sqrt{ V^{\new} } - \sqrt{ V } $
      \State \hspace{4mm} Let $S \leftarrow \ov{\pi} ( [r] )$ be the first $r$ indices in the permutation
      \State \hspace{4mm} Let $M_{*,S}\in \R^{n \times O(r)}$ be the $r$ column-blocks from $S$ of $M$
      \State \hspace{4mm} Let $M_{S,S},\Delta_{S,S} \in \R^{O(r) \times O(r)}$ be the $r$ row-blocks and column-blocks from $S$ of $M$, $\Delta$
      \State \hspace{4mm} $M^{\new} \leftarrow M - M_{*,S} \cdot ( \Delta^{-1}_{S,S} + M_{S,S} )^{-1} \cdot (M_{*,S})^\top$  \label{lin:offlinewoodburry} \Comment{Update $M$}
      \State \hspace{4mm} $Q^{\new} \leftarrow Q + R \cdot ( \Gamma \cdot M^{\new} ) + R \cdot \sqrt{V} \cdot ( M^{\new} - M )$ \Comment{Update $Q$}
      \State \hspace{4mm} $\ov{W} \leftarrow \ov{W}^{\new}$, $V \leftarrow V^{\new}$, $M \leftarrow M^{\new}$, $Q \leftarrow Q^{\new}$ \Comment{Update in memory}
      %\State \hspace{4mm} {\bf end if} 
      %\State \hspace{4mm}  \Comment{Update $W$ in memory}
      \State \hspace{4mm} $\tilde{v}_{i} \leftarrow \begin{cases} v_{i} & \text{if } (1-\epsilon_{mp}) v_i \preceq \ov{w}_i \preceq (1+\epsilon_{mp}) v_i\\ 
w_{i} & \text{otherwise} \end{cases}$\label{lin:vtilde}
      %\State \hspace{4mm} {\bf if} $l > \sqrt{n}$
      %\State \hspace{8mm} Create $R$ to be a list of sketching matrices \Comment{Lemma~\ref{lem:sketch_vector}}
      %\State \hspace{8mm} $Q \leftarrow R \sqrt{ \wt{V} } M$
      %\State \hspace{8mm} $l \leftarrow 1$  \Comment{Reset randomness counter, since we create a list of new sketching matrices}
      %\State \hspace{4mm} {\bf end if}
      \State \hspace{4mm} $F^{\new} \leftarrow \sqrt{ \wt{V} } M$, $G^{\new} \leftarrow \frac{1}{ \sqrt{ \wt{V} } } M$ 
      \State \hspace{4mm} $u_1 \leftarrow u_1 + (F - F^{\new}) u_2$, $u_3 \leftarrow u_3 + (G - G^{\new}) u_4$
      \State \hspace{4mm} $F \leftarrow F^{\new}$, $G \leftarrow G^{\new}$
      \State \hspace{4mm} Let $\wh{S}$ denote the blocks where $\wt{V}$ and $\wt{V}^{\new}$ are different
      \State \hspace{4mm} $\ov{x}_{\wh{S}} \leftarrow (u_1)_{\wh{S}} + (F u_2)_{\wh{S}}$, $\ov{s}_{\wh{S}} \leftarrow (u_3)_{\wh{S}} + (G u_2)_{\wh{S}}$ \Comment{make sure $x$ and $\ov{x}$ are close, similarly for $s$ and $\ov{s}$}
      \State \hspace{4mm} $t^{\pre} \leftarrow t$
  \State {\bf end procedure}
  \State
\State {\bf end datastructure}
\end{algorithmic}}%\vspace{1mm}
\end{algorithm}

\begin{algorithm}[!t]\caption{Central Path Maintenance Data Structure - Multiply and Move}\label{alg:central_path_maintenance_multiply_move} {
\begin{algorithmic}[1]
\State {\bf datastructure} \textsc{CentralPathMaintenance} \Comment{Theorem~\ref{thm:central_path_maintenance}}
  \State
  \State {\bf public : procedure} \textsc{MultiplyAndMove}{$(h,t)$} \Comment{Lemma~\ref{lem:central_path_maintenance_multiply_move}}
    \State \hspace{4mm} \textsc{Multiply}{$(h,t)$} 
    \State \hspace{4mm} \textsc{Move}$()$
  \State {\bf end procedure}
  \State
  \State {\bf private : procedure} \textsc{Multiply} {$(h,t)$} \Comment{Lemma~\ref{lem:central_path_maintenance_multiply}}
    \State \hspace{4mm} Let $\wt{S}$ be the indices $i$ such that $ (1-\epsilon_{mp}) v_i \preceq \ov{w}_i \preceq (1+\epsilon_{mp}) v_i$ is false.
    \State \hspace{4mm} $\wt{\Delta} \leftarrow \wt{V} - V$
    \State \hspace{4mm} $\wt{\Gamma} \leftarrow \sqrt{ \wt{V} } - \sqrt{ V }$
    \State \hspace{4mm} $\delta_m \leftarrow ( ( \wt{\Delta}_{\wt{S},\wt{S}}^{-1} + M_{\wt{S},\wt{S}} )^{-1}  \cdot ( (M_{\wt{S},*})^\top \sqrt{\wt{V}} h )  ) $ \Comment{$|\wt{S}| \leq n^a$}
    \State \hspace{4mm} \Comment{Compute $\wt{\delta}_x = \wt{V}^{1/2} (I - R^\top R \wt{P}) \wt{V}^{1/2} h$}
    \State \hspace{4mm} $ \wt{\delta}_x \leftarrow \wt{V} h - \Big(  ( R_l^\top \cdot (  (Q_l + R_l \cdot \wt{\Gamma} \cdot M )\cdot \sqrt{\wt{V}} \cdot h ) ) -  ( R_l^\top \cdot( ( Q_{l,\wt{S}} + R_l \cdot \wt{\Gamma} \cdot M_{\wt{S},*} )  \cdot \delta_m ) )  \Big) $
     \State \hspace{4mm} \Comment{Compute $\wt{\delta}_s = t \wt{V}^{-1/2} R^\top R \wt{P} \wt{V}^{1/2} h$}
    \State \hspace{4mm} $ \wt{\delta}_s \leftarrow t \cdot \wt{V}^{-1} \cdot \Big(  ( R_l^\top \cdot ( ( Q + R_l \cdot \wt{\Gamma} \cdot M  ) \cdot \sqrt{\wt{V}} \cdot h ) ) -  ( R_l^\top \cdot( ( Q_{l,\wt{S}} + R_l \cdot \wt{\Gamma} \cdot M_{\wt{S},*} ) \cdot \delta_m ) )  \Big) $
    \State \hspace{4mm} $l \leftarrow l + 1$ \Comment{Increasing the randomness counter, and using the new randomness next time}
    \State \hspace{4mm} \Comment{Implicitly maintain $x = x + \wt{V}^{1/2} ( I - \wt{P} ) \wt{V}^{1/2} h$}
    \State \hspace{4mm} $u_1 \leftarrow u_1 + \wt{V} h$ 
    \State \hspace{4mm} $u_2 \leftarrow u_2 - \sqrt{ \wt{V} } h + {\bf 1}_{\wt{S}} \delta_m $ %( ( \wt{\Delta}_{ \wt{S} , \wt{S} }^{-1} + M_{ \wt{S} , \wt{S} } )^{-1} M_{\wt{S}}^\top \sqrt{ \wt{V} } h )_{\wt{S}} $
    \State \hspace{4mm} \Comment{Implicitly maintain $s = s + t \wt{V}^{-1/2} \wt{P} \wt{V}^{1/2} h$}
    \State \hspace{4mm} $u_3 \leftarrow u_3 + 0$ 
    \State \hspace{4mm} $u_4 \leftarrow u_4 - t \sqrt{ \wt{V} } h + t {\bf 1}_{\wt{S}} \delta_m  $
  \State {\bf end procedure}
  \State
  \State {\bf private : procedure} {\textsc{Move}}{$()$} \Comment{Lemma~\ref{lem:central_path_maintenance_move}}
    \State \hspace{4mm} {\bf if} {$ l > \sqrt{n} $ or $t \geq t^{\pre} / 2$} \Comment{Variance is large enough} %\Comment{we actually update only for the bad block, not all of them}
        \State \hspace{8mm} $x \leftarrow u_1 + F u_2$, $s \leftarrow u_3 + F u_4$
        \State \hspace{8mm} \textsc{Initialize}($A,x,s,\ov{W},\epsilon_{mp},a,b$) \Comment{Algorithm~\ref{alg:central_path_maintenance_main}} %\textsc{FullUpdate}($\ov{W}$) \Comment{Algorithm~\ref{alg:central_path_maintenance_fullupdate}}
    \State \hspace{4mm} {\bf else}
        \State \hspace{8mm} $\ov{x} \leftarrow \ov{x} + \wt{\delta}_x$, $\ov{s} \leftarrow \ov{s} + \wt{\delta}_s$ \Comment{Update $\ov{x},\ov{s}$}
    \State \hspace{4mm} {\bf end if}
    \State \Return $( \ov{x} , \ov{s} )$
  \State {\bf end procedure}
  \State
\State {\bf end datastructure}
\end{algorithmic}}%\vspace{10mm}
\end{algorithm}

\subsection{Proof of Theorem~\ref{thm:central_path_maintenance}}

We follow the proof-sketch as \cite{cls18}. The proof contains four parts : 1) Definition of $X$ and $Y$, 2) We need to assume sorting, 3) We provide the definition of potential function, 4) We write the potential function.

\paragraph{Definition of matrices $X$ and $Y$.} Let us consider the $k$-th round of the algorithm. For all $i\in [m]$, matrix $\ov{y}_i^{(k)} \in \R^{n_i \times n_i}$ is constructed based on procedure \textsc{Update} (Algorithm~\ref{alg:central_path_maintenance_update}) :
\begin{align*}
\ov{y}_i^{(k)} = \frac{ \ov{w}_i^{ ( k + 1) } }{ v_i^{(k)} } - I.
\end{align*}
and $\ov{\pi}$ is a permutation such that $\| \ov{y}_{ \ov{\pi}(i) }^{(k)} \|_F \geq \| \ov{y}_{ \ov{\pi}(i + 1) }^{(k)} \|_F$.

For the purpose of analysis : for all $i \in [m]$, we define $x_i^{(k)}$, $x_i^{(k)}$ and $y_i^{(k)} \in \R^{n_i \times n_i}$ as follows:
\begin{align*}
x_i^{(k)} = \frac{ w_i^{(k)} }{ v_i^{(k)} } - I, ~~~ y_i^{(k)} = \frac{ w_i^{ ( k + 1) } }{ v_i^{(k)} } - I, ~~~ x_i^{ (k+1) } = \frac{ w_i^{ (k+1) } }{ v_i^{(k+1)} } - I,
\end{align*}
where $\frac{ w_i^{(k)} }{ v_i^{(k)} }$ denotes $ ( v_i^{(k)} )^{-1/2} w_i^{(k)} ( v_i^{(k)} )^{-1/2} $.

It is not hard to observe the difference between $x_i^{(k)}$ and $y_i^{(k)}$ is that $w$ is changing. We call it ``$w$ move''. Similarly, the difference between $y_i^{(k)}$ and $x_i^{(k+1)}$ is that $v$ is changing. We call it ``$v$ move''.

For each $i$, we define $\beta_i$ as follows
\begin{align*}
\beta_i = \| ( w_i^{(k)} )^{-1/2} ( \E[ w_i^{(k+1)} ] - w_i^{(k)} ) ( w_i^{(k)} )^{-1/2} \|_F,
\end{align*}
then one of assumption becomes 
\begin{align*}
\sum_{i=1}^m \beta_i^2 \leq C_1^2.
\end{align*}

\paragraph{Assume sorting for diagonal blocks.} Without loss of generality, we can assume the diagonal blocks of matrix $x^{(k)} \in \oplus_{i=1}^m \R^{n_i \times n_i}$ are sorted such that $\| x_i^{(k)} \|_F \geq \| x_{i+1}^{(k)} \|_F$. In \cite{cls18}, $x_i^{(k)}$ is a scalar. They sorted the sequence based on absolute value. In our situation, $x_i^{(k)}$ is a matrix. We sort the sequence based on Frobenius norm. Let $\tau$ permutation such that $\| x_{ \tau(i) }^{ (k + 1) } \|_F \geq \| x_{ \tau(i+1) }^{ (k + 1) } \|_F$. Let $\pi$ denote the permutation such that $\| y_{\pi(i)}^{(k)} \|_F \geq \| y_{ \pi(i+1) }^{(k)} \|_F$.

\paragraph{Definition of Potential function.} 
We define three functions $g$, $\psi$ and $\Phi_k$ here. The definition of $\psi$ is different from \cite{cls18}, since we need to handle matrix. The definitions of $g$ and $\Phi_k$ are the same as \cite{cls18}.

For the completeness, we still provide a definition of $g$. Let $g$ be defined as
\begin{align*}
g_i 
=
\begin{cases}
n^{-a}, & \text{~if~} i < n^a ; \\
i^{ \frac{ \omega - 2 }{ 1 - a } } n^{ - \frac{ a (\omega - 2) }{ 1 - a } }, & \text{~otherwise}.
\end{cases}
\end{align*}

In \cite{cls18}, the input of function $\psi : \R \rightarrow \R$ has to be a number. We allow matrix here. Let $\psi$ : square matrix $\rightarrow \R$ be defined by
\begin{align}\label{eq:def_psi_matrix}
\psi(x) = \begin{cases}
\frac{ \| x \|_F^2 }{ 2 \epsilon_{mp} }, & \| x \|_F \in [0, \epsilon_{mp} ] ; \\
\epsilon_{mp} - \frac{ ( 4 \epsilon_{mp}^2 - \| x \|_F^2 )^2 }{ 18 \epsilon_{mp}^3 }, & \| x \|_F \in ( \epsilon_{mp} , 2 \epsilon_{mp} ] ; \\
\epsilon_{mp}, & \| x \|_F \in (2\epsilon_{mp}, + \infty).
\end{cases}
\end{align}
where $\| x \|_F$ denotes the Frobenius norm of square matrix $x$, and let $L_1 = \max_x D_x \psi [h] / \| H \|_F $, $L_2 = \max_{x} D_{x}^2 \psi[ h , h ] / \| H \|_F^2$ where $h$ is the vectorization of matrix $H$.

For the completeness, we define the potential at the $k$-th round by
\begin{align*}
\Phi_k = \sum_{i=1}^m g_i \cdot \psi( x_{ \tau_k(i) }^{(k)} )
\end{align*}
where $\tau_k(i)$ is the permutation such that $\| x_{ \tau_k(i) }^{(k)} \|_F \geq \| x_{ \tau_k(i+1) }^{(k)} \|_F$. (Note that in \cite{cls18} $\| \cdot \|_F$ should be $| \cdot |$.)

\paragraph{Rewriting the potential, and bounding it.}
Following the ideas in \cite{cls18}, we can rewrite $\Phi_{k+1} - \Phi_k$ into two terms: the first term is $w$ move, and the second term is $v$ move. For the completeness, we still provide a proof.
\begin{align*}
\Phi_{k+1} - \Phi_k = & ~ \sum_{i=1}^m g_i \cdot \left( \psi( x_{\tau(i)}^{(k+1)} ) - \psi( x_i^{(k)} )  \right) \\
= & ~ \sum_{i=1}^m g_i \cdot \underbrace{ \left( \psi( y_{\pi(i)}^{(k)} ) - \psi ( x_i^{(k)} ) \right) }_{W~\text{move}} - \sum_{i=1}^m g_i \cdot \underbrace{ \left( \psi( y_{\pi(i)}^{(k)} ) - \psi( x_{\tau(i)}^{(k+1)} ) \right) }_{V~\text{move}}
\end{align*}
Using Lemma~\ref{lem:W_move}, we can bound the first term. Using Lemma~\ref{lem:v_move}, we can bound the second term.

\subsection{Initialization time, update time, query time, move time, multiply time}

\begin{remark}
In terms of implementing this data-structure, we only need three operations \textsc{Initialize}, \textsc{Update}, and \textsc{Query}. However, in order to make the proof more understoodable, we split \textsc{Update} into many operations : \textsc{FullUpdate}, \textsc{PartialUpdate}, \textsc{Multiply} and \textsc{Move}. We give a list of operations in Table~\ref{tab:central_path_maintenance}.
\end{remark}

\begin{lemma}[Initialization]\label{lem:central_path_maintenance_initialize}
The initialization time of data-structure \textsc{CentralPathMaintenance} (Algorithm~\ref{alg:central_path_maintenance_main}) is $O(n^{\omega+o(1)})$.
\end{lemma}
\begin{proof}
The running time is mainly dominated by two parts, the first part is computing $A^\top (A V A^\top)^{-1} A$, this takes $O(n^2 d^{\omega-2})$ time.

The second part is computing $R \sqrt{ \wt{V} } M$. This takes $O( n^{\omega + o(1)} )$ time.

\end{proof}

\begin{lemma}[Update time]\label{lem:central_path_maintenance_update}
The update time of data-structure \textsc{CentralPathMaintenance} (Algorithm~\ref{alg:central_path_maintenance_update}) is $O(r g_r n^{2 + o(1)})$ where $r$ is the number of indices we updated in $V$.
\end{lemma}
\begin{proof}
It is trivially follows from combining Lemma~\ref{lem:central_path_maintenance_partialupdate} and Lemma~\ref{lem:central_path_maintenance_fullupdate}.
\end{proof}

\begin{lemma}[Partial Update time]\label{lem:central_path_maintenance_partialupdate}
The partial update time of data-structure \textsc{CentralPathMaintenance} (Algorithm~\ref{alg:central_path_maintenance_update}) is $O(n^{1+a})$.
\end{lemma}
\begin{proof}

We first analyze the running time of $F$ update, the update equation of $F$ in algorithm is
\begin{align*}
F^{\new} \leftarrow & ~ F + ( (\wt{V}^{\new})^{1/2} - (\wt{V})^{1/2} ) M \\
F \leftarrow & F^{\new}
\end{align*}
which can be implemented as
\begin{align*}
F \leftarrow F + ( (\wt{V}^{\new})^{1/2} - (\wt{V})^{1/2} ) M
\end{align*}
where we only need to change $n^a$ row-blocks of $F$. It takes $O(n^{1+a})$ time.

Similarly, for the update time of $G$.

Next we analyze the update time of $u_1$, the update equation of $u_1$ is
\begin{align*}
u_1 \leftarrow u_1 + (F - F^{\new}) u_2
\end{align*}
Note that the difference between $F$ and $F^{\new}$ is only $n^a$ row-blocks, thus it takes $n^{1+a}$ time to update.

Finally we analyze the update time of $\ov{x}$. Let $\wh{S}$ denote the blocks where $\wt{V}$ and $\wt{V}^{\new}$ are different.
\begin{align*}
\ov{x}_{\wh{S}} \leftarrow (u_1)_{\wh{S}} + (F u_2)_{\wh{S}}
\end{align*}
This also can be done in $n^{1+a}$ time, since  $\wh{S}$ indicates only $n^a$ blocks.

Therefore, the overall running time is $O(n^{1+a})$.

\end{proof}

\begin{lemma}[Full Update time]\label{lem:central_path_maintenance_fullupdate}
The full update time of data-structure \textsc{CentralPathMaintenance} (Algorithm~\ref{alg:central_path_maintenance_fullupdate}) is $O(r g_r n^{2 + o(1)})$ where $r$ is the number of indices we updated in $V$.
\end{lemma}
\begin{proof}
The update equation we use for $Q$ is 
\begin{align*}
Q^{\new} \leftarrow Q + R \cdot ( \Gamma \cdot M^{\new} ) + R \cdot \sqrt{V} \cdot ( M^{\new} - M ).
\end{align*}
It can be re-written as 
\begin{align*}
Q^{\new} \leftarrow Q + R \cdot ( \Gamma \cdot M^{\new} ) + R \cdot \sqrt{V} \cdot ( - M_{*,S} \cdot ( \Delta_{S,S}^{-1} + M_{S,S} )^{-1} \cdot (M_{*,S})^\top )
\end{align*}
The running time of computing second term is multiplying a $n \times r$ matrix with another $r \times n$ matrix. The running time of computing third term is also dominated by multiplying a $n \times r$ matrix with another $r \times n$ matrix. 

Thus running time of processing $Q$ update is the same as the processing $M$ update. 

For the running time of other parts, it is dominated by the time of updating $M$ and $Q$.

Therefore, the rest of the proof is almost the same as Lemma 5.4 in \cite{cls18}, we omitted here.
\end{proof}

\begin{lemma}[Query time]\label{lem:central_path_maintenance_query}
The query time of data-structure \textsc{CentralPathMaintenance} (Algorithm~\ref{alg:central_path_maintenance_main}) is $O(n)$ time.
\end{lemma}
\begin{proof}
This takes only $O(n)$ time, since we stored $\ov{x}$ and $\ov{s}$.
\end{proof}

\begin{lemma}[Move time]\label{lem:central_path_maintenance_move}
The move time of data-structure \textsc{CentralPathMaintenance} (Algorithm~\ref{alg:central_path_maintenance_multiply_move}) is $O(n^{\omega+o(1)})$ time in the worst case, and is $O(n^{\omega-1/2+o(1)})$ amortized cost per iteration.
\end{lemma}
\begin{proof}
In one case, it takes only $O(n)$ time. For the other case, the running time is dominated by \textsc{Initialize}, which takes $n^{\omega+o(1)}$ by Lemma~\ref{lem:central_path_maintenance_fullupdate}. % multiplying a $n\times n^a$ matrix with another $n^a \times n$ matrix which is $n^{2+o(1)}$ time.

\end{proof}

\begin{lemma}[Multiply time]\label{lem:central_path_maintenance_multiply}
The multiply time of data-structure \textsc{CentralPathMaintenance} (Algorithm~\ref{alg:central_path_maintenance_multiply_move}) is $O(nb + n^{1+a + o(1)})$ for dense vector $\|h\|_0 = n$, and is $O(nb + n^{a\omega +o(1)} + n^a \| h \|_0 )$ for sparse vector $h$.
\end{lemma}

\begin{proof}

We first analyze the running time of computing vector $\delta_m$, the equation is
\begin{align*}
\delta_m \leftarrow \left( ( ( \wt{\Delta}_{\wt{S},\wt{S}} )^{-1} + M_{\wt{S},\wt{S}} )^{-1}  \cdot ( M_{ \wt{S} , * } ) ^\top \sqrt{ \wt{V} } h \right)
\end{align*}
where $\wt{\Delta} = \wt{V} - V $. Let $\wt{r} = \sum_{i \in \wt{S}}n_i = O(r)$ where $r$ is the number of blocks are different in $\wt{V}$ and $V$.

It contains several parts:

1. Computing $\wt{M}_{\wt{S}}^\top \cdot ( \sqrt{ \wt{V} } h ) \in \R^{\wt{r}}$ takes $O( \wt{r} ) \| h \|_0$. 

2. Computing $ ( \wt{\Delta}_{\wt{S} , \wt{S} }^{-1} + M_{\wt{S} , \wt{S}} )^{-1} \in \R^{ O( \wt{r} ) \times O( \wt{r} ) }$ that is the inverse of a $O( \wt{r} ) \times O( \wt{r} )$ matrix takes $O( \wt{r}^{\omega + o(1)} )$ time.

3. Computing matrix-vector multiplication between $O(\wt{r}) \times O(\wt{r})$ matrix $( ( \wt{\Delta}_{\wt{S}, \wt{S} } + M_{\wt{S}, \wt{S}} )^{-1} ) $ and $O(\wt{r}) \times 1$ vector $( (\wt{M}_{\wt{S},*})^\top \sqrt{ \wt{V} } h)$ takes $O(\wt{r}^2)$ time.

Thus, the running time of computing $\delta_m$ is
\begin{align*}
O(\wt{r} \| h \|_0 + \wt{r}^{\omega + o(1)} + \wt{r}^2 ) = O(\wt{r} \| h \|_0 + \wt{r}^{\omega + o(1)} ).
\end{align*}

Next, we want to analyze the update equation of $\wt{\delta}_x$
\begin{align*}
\wt{\delta}_x \leftarrow \wt{V} h - \Big(  ( R_l^\top \cdot (  (Q_l + R_l \sqrt{ \wt{ \Delta } } M )\cdot \sqrt{\wt{V}} \cdot h ) ) -  ( R_l^\top \cdot( ( Q_{l,\wt{S}} + R_l \wt{\Gamma} M_{\wt{S}} )  \cdot \delta_m ) )  \Big)
\end{align*}
where $\wt{\Gamma} = \sqrt{ \wt{V} } - \sqrt{ V }$ has $O(r)$ non-zero blocks.

It is clear that the running time is dominated by the second term in the equation. We only focus on that term.

1. Computing $R_l^\top Q_l \sqrt{ \wt{V} } h$ takes $O(bn)$ time, because $Q_l, R_l \in \R^{b \times n}$.

2. Computing $R_l^\top R_l \sqrt{ \wt{ \Delta } } M \sqrt{ \wt{V} } h$ takes $O(bn + b \wt{r} + \wt{r} \| h \|_0)$ time. The reason is, computing $\sqrt{\wt{\Delta}} M \sqrt{\wt{V}} h$ takes $\wt{r} \| h \|_0$ time, computing $R_l \cdot (\sqrt{\wt{\Delta}} M \sqrt{\wt{V}} h) $ takes $b \wt{r}$, then finally computing $R_l^\top \cdot ( R_l \sqrt{ \wt{ \Delta } } M \sqrt{ \wt{V} } h )$ takes $nb$.

Last, the update equation of $u_1,u_2,u_3,u_4$ only takes the $O(n)$ time.

Finally, we note that $r \leq O(n^a)$ due to the guarantee of \textsc{FullUpdate} and \textsc{PartialUpdate}.

Thus, overall the running time of the \textsc{Multiply} is
\begin{align*}
  O(\wt{r} \| h \|_0 + \wt{r}^{\omega + o(1)} +  \wt{r}^2 + b \wt{r} + nb ) 
= & ~ O( \wt{r} \| h \|_0 + \wt{r}^{\omega + o(1)} + nb ) \\
= & ~ O( r \| h \|_0 + r^{\omega + o(1)} + nb ) \\
= & ~ O( n^a \| h \|_0 + n^{a\omega + o(1)} + nb )
\end{align*}
where the first step follows from $b \wt{r} \leq nb$ and $\wt{r}^2 \leq \wt{r}^{\omega + o(1)}$, and the second step follows from $\wt{r} = O(r)$, and the last step follows from $r = O(n^a)$.
 
 If $h$ is the dense vector, then the overall time is
\begin{align*}
O(nb + n^{1+a} + n^{a \omega + o(1)}).
\end{align*}

Based on Lemma 5.5 in \cite{cls18}, we know that $a \omega \leq 1+a$. Thus, it becomes $O(nb + n^{1+a+o(1)})$ time.

If $h$ is a sparse vector, then the overall time is
\begin{align*}
O(nb + n^a \| h \|_0 + n^{a \omega + o(1)}).
\end{align*}

\end{proof}

\begin{lemma}[\textsc{MultiplyMove}]\label{lem:central_path_maintenance_multiply_move}
The running time of \textsc{MultiplyMove} (Algorithm~\ref{lem:central_path_maintenance_multiply_move}) is the \textsc{Multiply} time plus \textsc{Move} time.
\end{lemma}

\subsection{Bounding $W$ move}
The goal of this section is to analyze the movement of $W$. \cite{cls18} provided a scalar version of $W$ move, here we provide a matrix version.
\begin{lemma}[$W$ move, matrix version of Lemma 5.7 in \cite{cls18}]\label{lem:W_move}
\begin{align*}
\sum_{i=1}^m g_i \cdot \E \Big[ \psi( y_{\pi(i)}^{ (k) } ) - \psi ( x_{\pi(i)}^{ (k) } ) \Big] = O( C_1 + C_2 / \epsilon_{mp} ) \cdot \sqrt{\log n} \cdot ( n^{-a/2} + n^{\omega - 5/2} ).
\end{align*}
\end{lemma}
\begin{proof}
In scalar version, \cite{cls18} used absolute ($| \cdot |$) to measure each $x_i^{(k)}$. In matrix version, we use Frobenius norm ($\| \|_F$) to measure each $x_i^{(k)}$. 
Let $I \subseteq [m]$ be the set of indices such that $\| x_i^{(k)} \|_F \leq 1$. We separate the term into two :
\begin{align*}
\sum_{i=1}^m g_i \cdot \E[ \psi( y_{\pi(i)}^{(k)} ) - \psi ( x_{\pi(i)}^{(k)} ) ] = \sum_{i \in I} g_{ \pi^{-1}(i) } \cdot \E[ \psi( y_i^{(k)} ) - \psi ( x_i^{(k)} ) ] + \sum_{i \in I^c} g_{ \pi^{-1}(i) } \cdot \E [ \psi( y_i^{(k)} ) - \psi( x_i^{(k)} ) ] .
\end{align*}

{\bf Case 1.} Let us consider the terms from $I$.

Let $\text{vec}(y_i^{(k)})$ denote the vectorization of matrix $y_i^{(k)}$. Similarly, $\text{vec}(x_i^{(k)})$ denotes the vectorization of $x_i^{(k)}$. Mean value theorem shows that
\begin{align*}
\psi( y_i^{(k)} ) - \psi ( x_i^{(k)} ) 
= & ~ \langle \psi'( x_i^{(k)} ) , y_i^{(k)} - x_i^{(k)} \rangle + \frac{1}{2} \text{vec}( y_i^{(k)} - x_i^{(k)} )^\top \psi''(\zeta) \text{vec}( y_i^{(k)} - x_i^{(k)} ) \\
\leq & ~ \langle \psi'( x_i^{(k)} ) , y_i^{(k)} - x_i^{(k)}  \rangle + \frac{L_2}{2} \|  y_i^{(k)} - x_i^{(k)} \|_F^2 \\
= & ~ \langle \psi'( x_i^{(k)} ) , ( v_i^{(k)} )^{-1/2} ( w_i^{(k+1)} - w_i^{(k)} ) ( v_i^{(k)} )^{-1/2} \rangle \\
& ~ + \frac{L_2}{2} \| ( v_i^{(k)} )^{-1/2} ( w_i^{(k+1)} - w_i^{(k)} ) ( v_i^{(k)} )^{-1/2} \|_F^2
\end{align*}
where the second step follows from definition of $L_2$ (see Part 4 of Lemma~\ref{lem:def_psi_matrix}).

Taking conditional expectation given $w^{(k)}$ on both sides
\begin{align}\label{eq:write_expectation_into_beta_gamma}
\E[ \psi( y_i^{(k)} ) - \psi( x_i^{(k)} ) ]
\leq & ~ \langle \psi'( x_i^{(k)} ) , ( v_i^{(k)} )^{-1/2} ( \E[w_i^{(k+1)}] - w_i^{(k)} ) ( v_i^{(k)} )^{-1/2} \rangle \notag \\
& ~ + \frac{L_2}{2} \E[ \| ( v_i^{(k)} )^{-1/2} ( w_i^{(k+1)} - w_i^{(k)} ) ( v_i^{(k)} )^{-1/2} \|_F^2 ] \notag\\
\leq & ~ L_1 \| ( v_i^{(k)} )^{-1/2} ( \E[w_i^{(k+1)}] - w_i^{(k)} ) ( v_i^{(k)} )^{-1/2} \|_F \notag\\
& ~ + \frac{L_2}{2} \E[ \| ( v_i^{(k)} )^{-1/2} ( w_i^{(k+1)} - w_i^{(k)} ) ( v_i^{(k)} )^{-1/2} \|_F^2 ] \notag\\
\leq & ~ L_1 \| (v_i^{(k)})^{-1/2} ( w_i^{(k)} )^{1/2} \|^2 \cdot \| ( w_i^{(k)} )^{-1/2} ( \E[w_i^{(k+1)}] - w_i^{(k)} ) ( w_i^{(k)} )^{-1/2} \|_F \notag\\
& ~ + \frac{L_2}{2} \| (v_i^{(k)})^{-1/2} ( w_i^{(k)} )^{1/2} \|^4 \cdot \E[ \| ( w_i^{(k)} )^{-1/2} ( w_i^{(k+1)} - w_i^{(k)} ) ( w_i^{(k)} )^{-1/2} \|_F^2 ] \notag\\
= & ~ L_1 \| (v_i^{(k)})^{-1/2} ( w_i^{(k)} )^{1/2} \|^2 \cdot \beta_i + \frac{L_2}{2} \| (v_i^{(k)})^{-1/2} ( w_i^{(k)} )^{1/2} \|^4 \cdot \gamma_i
\end{align}
where the second step follows from definition of $L_2$ (see Part 4 of Lemma~\ref{lem:def_psi_matrix}), the third step follows from $\| A B \|_F \leq \| A \|_F \cdot \| B \|$, and the last step follows from defining $\beta_i$ and $\gamma_i$ as follows:
\begin{align*}
\beta_i = & ~ \Big\| ( w_i^{(k)} )^{-1/2} ( \E[w_i^{(k+1)}] - w_i^{(k)} ) ( w_i^{(k)} )^{-1/2} \Big\|_F \\
\gamma_i = & ~ \E \Big[ \Big\| ( w_i^{(k)} )^{-1/2} ( w_i^{(k+1)} - w_i^{(k)} ) ( w_i^{(k)} )^{-1/2} \Big\|_F^2 \Big] .
\end{align*}

To upper bound $\sum_{i \in I} g_{ \pi^{-1} (i) } \E [ \psi( y_i^{(k)} ) - \psi ( x_i^{(k)} ) ] $, we need to bound the following two terms,
\begin{align}\label{eq:two_terms_beta_gamma}
\sum_{i \in I} g_{\pi^{-1}(i)} L_1 \Big\| (v_i^{(k)})^{-1/2} (w_i^{(k)})^{1/2} \Big\|^2 \beta_i , \text{~and~} \sum_{i \in I} g_{\pi^{-1}(i)} \frac{L_2}{2} \Big\| (v_i^{(k)})^{-1/2} (w_i^{(k)})^{1/2} \Big\|^4 \gamma_i .
\end{align}

For the first term (which is related to $\beta$) in Eq.~\eqref{eq:two_terms_beta_gamma}, we have
\begin{align}\label{eq:upper_bound_beta_part}
\sum_{i \in I} g_{ \pi^{-1}(i) } L_1 \| (v_i^{(k)})^{-1/2} (w_i^{(k)})^{1/2} \|^2 \beta_i 
\leq & ~ \left( \sum_{i \in I} \left( g_{\pi^{-1}(i)} L_1 \| (v_i^{(k)})^{-1/2} (w_i^{(k)})^{1/2} \|^2 \right)^2 \sum_{i \in I} \beta_i^2 \right)^{1/2} \notag \\
\leq & ~ O(L_1) \left( \sum_{i=1}^n g_i^2 \cdot C_1^2 \right)^{1/2} \notag \\
= & ~ O( C_1 L_1 \| g \|_2 ).
\end{align}
where the first step follows from Cauchy-Schwarz inequality, the second step follows from $n_i = O(1)$ and $\| (v_i^{(k)})^{-1/2} (w_i^{(k)})^{1/2} \|^2 = O(1)$.

For the second term (which is related to $\gamma$) in Eq.~\eqref{eq:two_terms_beta_gamma}, we have
\begin{align}\label{eq:upper_bound_gamma_part}
\sum_{i \in I} g_{\pi^{-1}(i)} \frac{L_2}{2} \| (v_i^{(k)})^{-1/2} (w_i^{(k)})^{1/2} \|^4 n_i \gamma_i \leq O(L_2) \cdot \sum_{i=1}^m g_i \cdot \gamma_i = O(C_2 L_2 \| g \|_2).
\end{align}

Putting Eq.~\eqref{eq:write_expectation_into_beta_gamma}, Eq.~\eqref{eq:upper_bound_beta_part} and Eq.~\eqref{eq:upper_bound_gamma_part} together, and using several facts $L_1 = O(1)$, $L_2 = O(1 / \epsilon_{mp} )$ (from part 4 of Lemma~\ref{lem:def_psi_matrix}) and $\| g \|_2 \leq \sqrt{\log n} \cdot O( n^{-a/2} + n^{\omega - 5/2} ) $ (from Lemma~\ref{lem:lemma_5.8_in_cls18}) gives us
\begin{align*}
\sum_{i \in I} g_{\pi^{-1}(i)} \cdot \E[ \psi( y_i^{k)} ) - \psi( x_i^{(k)} ) ] \leq O(C_1 + C_2 / \epsilon_{mp}) \cdot \sqrt{\log n} \cdot ( n^{-a/2} + n^{\omega-5/2} ).
\end{align*}
(Note that, the above Equation is the same as \cite{cls18}.)

{\bf Case 2.} Let us consider the terms from $I^c$.

For each $i \in I^c$, we know $\| x_i^{(k)} \|_F \geq 1$. We observe that $\psi(x)$ is constant for $\|x \|_F^2 \geq ( 2 \epsilon_{mp} )^2$, where $\epsilon_{mp} \leq 1/4$. 
If $\| y_i^{(k)} \|_F \geq 1/2$, then $\psi( y_i^{(k)} ) - \psi( x_i^{(k)} ) = 0 $. Therefore, we only need to focus on the $i \in I^c$ such that $\| y_i^{(k)} \|_F < 1/2$. 

For each $i \in I^c$ with $\| y_i^{(k)} \|_F < 1/2$, we have
\begin{align}\label{eq:1_over_2_less_than_3_over_2}
\frac{1}{2} 
< & ~ \| y_i^{(k)} - x_i^{(k)} \|_F \notag \\
= & ~ \| (v_i^{(k)})^{-1/2} ( w_i^{(k+1)} - w_i^{(k)} ) (v_i^{(k)})^{-1/2} \|_F \notag \\
= & ~ \| (v_i^{(k)})^{-1/2} ( w_i^{(k)} )^{1/2} \cdot ( w_i^{(k)} )^{-1/2} ( w_i^{(k+1)} - w_i^{(k)} ) ( w_i^{(k)} )^{-1/2} \cdot ( w_i^{(k)} )^{1/2} (v_i^{(k)})^{-1/2} \|_F \notag \\
\leq & ~ \| (v_i^{(k)})^{-1/2} ( w_i^{(k)} )^{1/2} \| \cdot \| ( w_i^{(k)} )^{-1/2} ( w_i^{(k+1)} - w_i^{(k)} ) ( w_i^{(k)} )^{-1/2} \|_F \cdot \| ( w_i^{(k)} )^{1/2} (v_i^{(k)})^{-1/2} \| \notag \\
= & ~ \| (v_i^{(k)})^{-1/2}  w_i^{(k)} (v_i^{(k)})^{-1/2} \| \cdot \| ( w_i^{(k)} )^{-1/2} ( w_i^{(k+1)} - w_i^{(k)} ) ( w_i^{(k)} )^{-1/2} \|_F \notag \\
\leq & ~ \frac{3}{2}\| ( w_i^{(k)} )^{-1/2} ( w_i^{(k+1)} - w_i^{(k)} ) ( w_i^{(k)} )^{-1/2} \|_F 
\end{align}
where the last step follows from $\| y_i^{(k)} \|_F = \| \frac{ w_i^{(k+1)} }{ v_i^{(k)} } - I \|_F \leq 1/2$. 

It is obvious that Eq.~\eqref{eq:1_over_2_less_than_3_over_2} implies
\begin{align*}
\| ( w_i^{(k)} )^{-1/2} ( w_i^{(k+1)} - w_i^{(k)} ) ( w_i^{(k)} )^{-1/2} \|_F  > 1/3 > 1/4.
\end{align*}
But this is impossible, since we assume it is $\leq 1/4$.

Thus, we have
\begin{align*}
\sum_{i\in I^c} g_{\pi^{-1} (i) } \cdot \E[ \psi( y_i^{(k)} ) - \psi( x_i^{(k)} ) ] = 0.
\end{align*}
\end{proof}

We state a Lemma that was proved in previous work \cite{cls18}.
\begin{lemma}[Lemma 5.8 in \cite{cls18}]\label{lem:lemma_5.8_in_cls18}
\begin{align*}
\left( \sum_{i=1}^n g_i^2 \right)^{1/2} \leq \sqrt{\log n} \cdot O( n^{-a/2} + n^{\omega - 5/2} )
\end{align*}
\end{lemma}

\subsection{Bounding $V$ move}\label{sec:v_move}
In previous work, \cite{cls18} only handled the movement of $V$ in scalar version. Here, the goal of is to understand the movement of $V$ in matrix version. We start to give some definitions about block diagonal matrices.
\begin{definition}
We define block diagonal matrices $X^{(k)}$, $Y^{(k)}$, $X^{(k+1)}$ and $\ov{Y}^{(k)} \otimes_{i\in [m]} \R^{n_i \times n_i}$ as follows
\begin{align*}
x_i^{(k)} = \frac{ w_i^{(k)} }{ v_i^{(k)} } - I, ~~~ y_i^{(k)} = \frac{ w_i^{(k+1)} }{ v_i^{(k)} } - I, ~~~ x_i^{(k+1)} = \frac{ w_i^{(k+1)} }{ v_i^{(k+1)} } - I, ~~~ \ov{y}_i^{ (k) } = \frac{ \ov{w}_i^{(k+1)} }{ v_i^{(k)} } - I .
\end{align*}
Let $\epsilon_w$ denote the error between $W$ and $\ov{W}$
\begin{align*}
\| W_i^{-1/2} ( \ov{W}_i - W_i ) W_i^{-1/2} \|_F \leq \epsilon_w.
\end{align*}
\end{definition}

\begin{lemma}[$V$ move, matrix version of Lemma 5.9 in \cite{cls18}]\label{lem:v_move}
We have,
 \begin{align*}
 \sum_{i = 1}^n g_i \cdot \left( \psi(  y^{(k)}_{\pi(i)} ) - \psi ( x^{(k+1)}_{\tau(i)} ) \right) \geq \Omega( \epsilon_{mp} r_k g_{r_k} / \log n).
\end{align*}
\end{lemma}
\begin{proof}
To prove the Lemma, similarly as \cite{cls18}, we will split the proof into two cases.

Before getting into the details of each case, let us first understand several simple facts which are useful in the later proof. Note that from the definition of the algorithm, we only change the block if $\| \ov{y}_i^{(k)} \|_F$ is larger than the error between $w_i$ and $\ov{w}_i$. Hence, all the changes only decreases the norm, namely $\psi(y_i^{(k)}) \geq \psi( x_i^{(k+1)} ) $ for all $i$. So is their sorted version $\psi( y_{\pi(i)}^{(k)} ) \geq \psi( x_{ \tau(i) } )^{(k+1)}$ for all $i$.

\paragraph{Case 1.} The procedure exits the while loop when $1.5 r_k \geq n$.

Let $u^*$ denote the largest $u$ such that $\| \ov{y}^{(k)}_{ \ov{\pi} (u)} \|_F \geq \epsilon_{mp}$. 

If $u^* = r_k$, we have that 
\begin{align*} 
\| \ov{y} ^{(k)}_{ \ov{\pi} (r_k)} \|_F \geq \epsilon_{mp} \geq \epsilon_{mp}/100.
\end{align*}

If $u^* \neq r_k$, using the condition of the loop, we have that
\begin{align*}
 \| \ov{y}^{(k)}_{ \ov{\pi} (r_k)} \|_F 
 \geq & ~ ( 1 - 1 / \log n)^{\log_{1.5} r_k - \log_{1.5} u^*} \cdot \| \ov{y}^{(k)}_{ \ov{\pi} (u^*)} \|_F \\
 \geq & ~ ( 1 - 1 / \log n )^{\log_{1.5} n} \cdot \epsilon_{mp} \\
 \geq & ~ \epsilon_{mp}/100 .
\end{align*}
where the last step follows from $n \geq 4$. 

Recall the definition of $x^{(k+1)}_{\tau(i)}$. We can lower bound the $\LHS$ in the Lemma statement in the following sense,
\begin{align*}
\sum_{i = 1}^n g_i \cdot \Big( \psi ( y^{(k)}_{\pi(i)} ) - \psi (  x^{(k+1)}_{\tau(i)} ) \Big) 
\geq & ~ \sum_{i = n/3+1}^{ 2n / 3 } g_i \cdot \Big( \psi ( y^{(k)}_{\pi(i)} ) - \psi (  x^{(k+1)}_{\tau(i)} ) \Big) \\
\geq & ~ \sum_{i = n/3+1}^{ 2n / 3 } g_i \cdot ( \Omega( \epsilon_{mp} ) - O(\epsilon_w) ) \\
\geq & ~ \sum_{i = n/3+1}^{ 2n / 3 } g_i \cdot \Omega( \epsilon_{mp} ) \\
= & ~ \Omega (r_k g_{r_k} \epsilon_{mp} ).
\end{align*}
where the second step follows from $\| y^{(k)}_{\pi(i)} \|_F \geq \| y^{(k)}_{\pi(r_k)} \|_F \geq (1-O(\epsilon_w)) \|  \ov{y}_{ \ov{\pi}(r_k) }^{(k)} \|_F \geq \epsilon_{mp}/200$ for all $i<2n/3$.

\paragraph{Case 2.}
The procedure exits the while loop when $1.5 r_k < n$ and  $ \| \ov{y}^{(k)}_{ \ov{\pi} (1.5 r_k)} \|_F < (1- 1 / \log n) \| \ov{y}^{(k)}_{ \ov{\pi} (r_k)} \|_F $.

Using the same argument as Case 1, we have 
\begin{align*}
\| \ov{y}^{(k)}_{ \ov{\pi} (r_k)} \|_F \geq \epsilon_{mp} / 100.
\end{align*}
Using Part 3 of Lemma~\ref{lem:def_psi_matrix} and the following fact
\begin{align*}
\| \ov{y}^{(k)}_{ \ov{\pi} (1.5 r) } \|_F < \min \left( \epsilon_{mp}, \| \ov{y}^{(k)}_{ \ov{\pi}(r)} \|_F \cdot (1 - 1/\log n) \right),
\end{align*}
we can show that
\begin{align}\label{eq:r1.5_at_most_r_ov_pi_version}
\psi( \ov{y}^{(k)}_{ \ov{\pi} (1.5 r) } )  - \psi( \ov{y}^{(k)}_{ \ov{\pi} (r) } ) = \Omega(\epsilon_{mp} / \log n).
\end{align}

Now the question is, how to relax $\psi( \ov{y}^{(k)}_{ \ov{\pi}(1.5r) } )$ to $\psi( {y}^{(k)}_{ {\pi}(1.5r) } )$ and how to relax $\psi( \ov{y}^{(k)}_{ \ov{\pi}(r) } )$ to $\psi( {y}^{(k)}_{ {\pi}(r) } )$

Note that $\| y_i^{(k)} \|_F \geq \| x_i^{(k+1)} \|_F$ for all $i$. Hence, we have $\psi ( y_{ \pi( i ) }^{(k)} ) \geq \psi ( x_{ \tau( i ) }^{ ( k + 1 ) } )$ for all $i$.

 Recall the definition of $y$, $\ov{y}$, $\pi$ and $\ov{\pi}$,
\begin{align*}
y_i^{(k)} = \frac{ w_i^{(k+1)} }{ v_i^{(k)} } - I, ~~~ \ov{y}_i^{(k)} = \frac{ \ov{w}_i^{(k+1)} }{ v_i^{(k)} } - I .
\end{align*}
and $\pi$ and $\ov{\pi}$ denote the permutations such that $\| y_{ \pi(i) }^{ ( k ) } \|_F \geq \| y_{ \pi (i+1) }^{(k)} \|_F$ and $\| \ov{y}_{ \ov{ \pi } (i) }^{(k)} \|_F \geq \| \ov{y}_{ \ov{\pi} (i+1) }^{(k)} \|_F$.

Using Fact~\ref{fac:two_sorted_sequences} and $\| \cdot \|_2 = \Theta(1) \| \cdot \|_F$ when the matrix has constant dimension
\begin{align*}
\| y_{\pi(i)}^{(k)} - \ov{y}_{\ov{\pi}(i)}^{(k)} \|_F \leq O(\epsilon_w). 
\end{align*}
where $\epsilon_w$ is the error between $W$ and $\ov{W}$.

Next, $\forall i$, we have
\begin{align}\label{eq:relax_y_to_ov_y_lower}
\psi( y_{ \pi(i) }^{(k)} ) 
= & ~ \psi( \ov{y}_{ \ov{\pi}(i) }^{(k)} ) \pm O(\epsilon_w \epsilon_{mp})
\end{align}
Next, we note that all the blocks the algorithm updated must lies in the range $i = 1, \cdots, \frac{ 3 r_k }{ 2 } - 1$. After the update, the error of $r_k$ of these block becomes so small that its rank will much higher than $r_k$. Hence, $r_k/2$ of the unchanged blocks in the range $i = 1, \cdots, \frac{3r_k}{2}$ will move earlier in the rank. Therefore, the $r_k/2$-th element in $x^{(k+1)}$ must be larger than the $\frac{3}{2}r_k$-th element in $y^{(k)}$. In short, we have $\psi( x_{\tau(i)}^{(k+1)} ) \leq \psi ( y_{ \pi(1.5 r_k) }^{(k)} )$ for all $i \geq r_k / 2$.

%(Note that, the reason is: there are two cases, for the case $ \| y_{\pi(i)}^{(k)} \|_F \gg \epsilon_{mp}$, the function value remains to be the same. For the other case, it can decrease at most $O(\epsilon_w \epsilon_{mp})$.)

Putting it all together, we have 
 \begin{align*}
& ~ \sum_{i = 1}^n g_i \cdot \Big( \psi( y^{(k)}_{\pi(i)} ) - \psi ( x^{(k+1)}_{\tau(i)} ) \Big) \\
\geq & ~ \sum_{i = r_k/2}^{r_k} g_i \cdot \Big( \psi( y^{(k)}_{\pi(i)} ) - \psi ( x^{(k+1)}_{\tau(i)} ) \Big) \\
\geq & ~ \sum_{i = r_k/2}^{r_k} g_i \cdot \Big( \psi( y^{(k)}_{\pi(i)} ) - \psi ( y^{(k+1)}_{\pi(1.5 r_k)} ) \Big) & \text{~by~} \psi( x_{\tau(i)}^{(k+1)} ) \leq \psi ( y_{ \pi(1.5 r_k) }^{(k)} ), \forall i \geq r_k / 2 \\
\geq & ~ \sum_{i = r_k/2}^{r_k} g_i \cdot \Big( \psi( \ov{y}^{(k)}_{\ov{\pi}(i)} ) - \psi ( \ov{y}^{(k+1)}_{\ov{\pi}(1.5 r_k)} ) - O(\epsilon_w \epsilon_{mp} ) \Big) & \text{~by~\eqref{eq:relax_y_to_ov_y_lower}}  \\
\geq & ~ \sum_{i = r_k/2}^{r_k} g_i \cdot \Big( \psi( \ov{y}^{(k)}_{\ov{\pi}(r_k)} ) - \psi ( \ov{y}^{(k+1)}_{ \ov{\pi}(1.5 r_k)} ) - O(\epsilon_w \epsilon_{mp} ) \Big) & \text{~by~} \psi( \ov{y}^{(k)}_{\ov{\pi}(i)} ) \geq \psi( \ov{y}^{(k)}_{\ov{\pi}(r_k)} ), \forall i \in [r_k/2,r_k] \\
\geq & ~ \sum_{i= r_k/2}^{ r_k } g_{r_k} \cdot \Big( \Omega(\frac{\epsilon_{mp}}{\log n}) - O(\epsilon_w \epsilon_{mp}) \Big) & \text{~by~\eqref{eq:r1.5_at_most_r_ov_pi_version} } \\
\geq & ~ \sum_{i= r_k/2}^{ r_k } g_{r_k} \cdot \Omega(\frac{\epsilon_{mp}}{\log n}) & \text{~by~} \epsilon_w < O( 1 / \log n) \\
= & ~ \Omega\left( \epsilon_{mp} r_k  g_{r_k} / \log n \right) .
\end{align*}
Therefore, we complete the proof.
\end{proof}

\begin{comment}
Putting it all together, we have 
 \begin{align*}
& ~ \sum_{i = 1}^n g_i \cdot \Big( \psi( y^{(k)}_{\pi(i)} ) - \psi ( x^{(k+1)}_{\tau(i)} ) \Big) \\
\lesssim & ~ \sum_{i = 1}^n g_i \cdot \Big( \psi( y^{(k)}_{\pi(i)} ) - \psi ( y^{(k)}_{\pi(i+r_k)} ) \Big) & \text{~by~} x^{(k+1)}_{\tau(i)} \approx_{\epsilon_{w}} y^{(k)}_{\pi(i+r_k)} \\
\geq & ~ \sum_{i = r_k/2 }^{ r_k }  g_i \cdot \Big( \psi ( y^{(k)}_{\pi(i)} ) - \psi ( y^{(k)}_{\pi(i+r_k)} ) \Big) & \text{~by~}  \psi(y^{(k)}_{\pi(i)}) - \psi(y^{(k)}_{\pi(i+r_k)}) \geq 0 \\
\geq & ~ \sum_{i = r_k/2 }^{ r_k }  g_i \cdot \Big( \psi ( \ov{y}^{(k)}_{ \ov{\pi} (r_k)} ) - \psi ( \ov{y}^{(k)}_{ \ov{\pi} (1.5r_k)} ) - O(\epsilon_w \epsilon_{mp}) \Big)  & \text{~by~\eqref{eq:relax_y_to_ov_y_lower} and \eqref{eq:relax_y_to_ov_y_upper}} \\
\geq & ~ \sum_{i= r_k/2 }^{ r_k } g_i \cdot \Big( \Omega(\frac{\epsilon_{mp}}{\log n}) - O( \epsilon_w \epsilon_{mp}) \Big) & \text{~by~ \eqref{eq:r1.5_at_most_r_ov_pi_version}} \\
\geq & ~ \sum_{i= r_k/2 }^{ r_k } g_i \cdot \Omega(\frac{\epsilon_{mp}}{\log n})  & \text{~by~$\epsilon_w \leq O(1/\log n)$} \\
\geq & ~ \sum_{i= r_k/2}^{ r_k } g_{r_k} \cdot \Omega(\frac{\epsilon_{mp}}{\log n}) & \text{~by~}g_i\text{~is~decreasing} \\
= & ~ \Omega\left( \epsilon_{mp} r_k  g_{r_k} / \log n \right) .
\end{align*}
Therefore, we complete the proof.
%where the third step follows by $\| y^{(k)}_{\pi(i)} \|_F$ is decreasing and $\psi$ is non-decreasing (from part 2 of Lemma~\ref{lem:def_psi_matrix}).
\end{proof}
\end{comment}

\begin{fact}\label{fac:two_sorted_sequences}
Given two length $n$ positive vectors $a,b$. Let $a$ be sorted such that $a_i \geq a_{i+1}$. Let $\pi$ denote the permutation such that $b_{ \pi(i) } \geq b_{\pi(i+1)}$. If for all $i \in [n]$, $|a_i - b_i| \leq \epsilon a_i$. Then for all $i \in [n]$, $| a_i - b_{\pi(i)} | \leq \epsilon a_i$.
\end{fact}
\begin{proof}
Case 1. $\pi(i) = i$. This is trivially true.

Case 2. $\pi(i) < i$. 
We have
\begin{align*}
b_{\pi(i)} \geq b_i \geq (1-\epsilon) a_i
\end{align*}
Since $\pi(i) < i$, we know that there exists a $j > i$ such that $\pi(j) < \pi(i)$. Then we have
\begin{align*}
b_{\pi(i)} \leq b_{ \pi(j) } \leq (1+\epsilon) a_j \leq (1+\epsilon) a_i
\end{align*}
Combining the above two inequalities, we have $(1-\epsilon) a_i \leq b_{\pi(i)} \leq (1+\epsilon) a_i$.

Case 3. $\pi(i) > i$.
We have
\begin{align*}
b_{\pi(i)} \leq b_i \leq (1+\epsilon) a_i
\end{align*}
Since $\pi > i$, we know that there exists $j < i$ such that $\pi(j) > \pi(i)$. Then we have
\begin{align*}
b_{\pi(i)} \geq b_{\pi(j)} \geq (1-\epsilon) a_j \geq (1-\epsilon) a_i.
\end{align*}
Combining the above two inequalities gives us $(1-\epsilon) a_i \leq b_{\pi(i)} \leq (1+\epsilon) a_i$.

Therefore, putting all the three cases together completes the proof.

\end{proof}

\subsection{Potential function $\psi$}
\cite{cls18} used a scalar version potential function. Here, we generalize it to the matrix version.
\begin{lemma}[Matrix version of Lemma 5.10 in \cite{cls18}]\label{lem:def_psi_matrix}
Let function $\psi : \mathrm{square~matrix} \rightarrow \R$ (defined as Eq.~\eqref{eq:def_psi_matrix}) satisfies the following properties :\\
1. Symmetric $( \psi(x) = \psi(-x) )$ and $\psi(0) = 0$\\
2. If $\| x \|_F \geq \| y \|_F$, then $\psi(x) \geq \psi(y)$ \\% Let $f(\| x \|_F^2) = \psi( x )$, then $f( \| x \|_F^2 )$ is non-decreasing\\
3. $| f'(x) | = \Omega(1/\epsilon_{mp}), \forall x \in [ (0.01 \epsilon_{mp})^2 , \epsilon_{mp}^2 ]$ \\ %3. $| D_x \psi[H] | = \Omega(1) \| H \|_F, \forall \| x \|_F \in [ 0.01 \epsilon_{mp} , \epsilon_{mp}]$ \\
4. $L_1 \overset{\mathrm{def}}{=} \max_x \frac{ D_x \psi[H] }{ \| H \|_F } = 2 $ and $L_2 \overset{ \mathrm{def} }{=} \max_x \frac{ D_x^2 \psi[H,H] }{ \| H \|_F^2 } = 10 /\epsilon_{mp}$ \\
5. $\psi( x  )$ is a constant for $\| x \|_F \geq 2 \epsilon_{mp}$
\end{lemma}
\begin{proof}
Let $f : \R_{+} \rightarrow \R$ be defined as
\begin{align*}
f(x) = 
\begin{cases}
\frac{ x^2 }{ 2 \epsilon_{mp}^3 }, & ~ x \in [0, \epsilon_{mp}^2 ] ; \\
\epsilon_{mp} - \frac{ ( 4 \epsilon_{mp}^2 - x )^2 }{18 \epsilon_{mp}^3 }, & ~ x \in ( \epsilon_{mp}^2, 4 \epsilon_{mp}^2 ] ; \\
\epsilon_{mp}, & ~ x \in (4 \epsilon_{mp}^2, + \infty) .
\end{cases}
\end{align*}
We can see that 

\begin{align*}
f(x)' = 
\begin{cases}
\frac{ x }{ \epsilon_{mp}^3 }, & ~ x \in [0, \epsilon_{mp}^2 ] ; \\
\frac{ 4 \epsilon_{mp}^2 - x  }{9 \epsilon_{mp}^3 }, & ~ x \in ( \epsilon_{mp}^2, 4 \epsilon_{mp}^2 ] ; \\
0, & ~ x \in (4 \epsilon_{mp}^2, + \infty) .
\end{cases}
\text{~and~}
f(x)'' = 
\begin{cases}
\frac{ 1 }{ \epsilon_{mp}^3 }, & ~ x \in [0, \epsilon_{mp}^2 ] ; \\ %\cup [-4\epsilon_{mp}^2, - \epsilon_{mp}^2] ; \\
- \frac{ 1 }{ 9 \epsilon_{mp}^3 }, & ~ x \in ( \epsilon_{mp}^2, 4 \epsilon_{mp}^2 ] ; \\ % \cup [ - \epsilon_{mp}^2 , 0 ] ; \\
0, & ~ |x| \in (4 \epsilon_{mp}^2, + \infty) .
\end{cases}
\end{align*}
It implies that $\max_x |f(x)'| \leq \frac{ 1 }{ \epsilon_{mp} }$ and $\max_x |f(x)''| \leq \frac{1}{ \epsilon_{mp}^3 }$. Let $\psi(x) = f( \| X \|_F^2 )$.

{\bf Proof of Part 1,2 and 5.}
These proofs are pretty standard from definition of $\psi$.

{\bf Proof of Part 3.}
This is trivially following from definition of scalar function $f$.
%By chain rule, we have
%\begin{align*}
%\| D_x \psi \|_2 = 2 \| f'( \| X \|_F^2 ) x \|_2 = 2 | f'( \| X \|_F^2 ) | \cdot \| X \|_F \geq 2 \| X \|_F^3 / \epsilon_{mp}^3 \geq \frac{2}{10^6 \epsilon_{mp^3}}
%\end{align*}
%where the third step follows from $\| X \|_F \in [0,\epsilon_{mp}]$ and the last step follows from $\| X\|_F \geq 0.01 \epsilon_{mp} $.

{\bf Proof of Part 4.}
 By chain rule, we have
\begin{align*}
D_x \psi [h] = & ~  2 f'( \| X \|_F^2 ) \cdot \tr[ X H ] \\
D_x^2 \psi [h,h] = & ~ 2 f''( \| X \|_F^2 ) \cdot ( \tr[ X H] )^2 + 2 f'( \| x \|_F^2 ) \cdot \tr[ H^2 ]
\end{align*}
where $x$ is the vectorization of matrix $X$ and $h$ is the vectorization of matrix $H$. We can upper bound
\begin{align*}
| D_x \psi [ h ] | \leq 2 | f'( \| X \|_F^2 ) | \cdot | \tr[ X H ] | 
\leq  2 | f'( \| X \|_F^2 ) | \cdot \| X \|_F \cdot \| H \|_F 
\end{align*}
Then, we have
\begin{align*}
 | f'( \| X \|_F^2 ) | \cdot \| X \|_F
=
\begin{cases}
\| X \|_F^3 / \epsilon_{mp}^3 \leq 1, & \| X \|_F \in [0, \epsilon_{mp}] \\
( 4 \epsilon_{mp}^2 - \| X \|_F^2 ) \| X \|_F / 9 \epsilon_{mp} \leq 2/3, & \| X \|_F \in (\epsilon_{mp}, 2\epsilon_{mp}] \\
0, & \| X \|_F \in (2\epsilon_{mp}, +\infty)
\end{cases}
\end{align*}
It implies that $|D_x \psi[h]| \leq 2 \| H \|_F$, $\forall x$.

By case analysis, we have
\begin{align*}
|f''( \| X \|_F^2 ) | \cdot \| X \|_F^2 \leq
\begin{cases}
\frac{1}{\epsilon_{mp}^3} \| X \|_F^2 \leq 4 / \epsilon_{mp}, & \| X \|_F^2 \in [0, 4 \epsilon_{mp}^2] \\
0, & ~ \| X \|_F^2 \in ( 4 \epsilon_{mp} , + \infty)
\end{cases}
\end{align*}

We can also upper bound
\begin{align*}
| D_x^2 \psi[ h,h ] | \leq & ~ 2 | f''( \| X \|_F^2 ) | \cdot ( \tr[ X H ] )^2 + 2 |f'( \| X \|_F^2 )| \cdot \tr[ H^2 ] \\
\leq & ~ 2 | f''( \| X \|_F^2 ) | \cdot ( \| X \|_F \| H \|_F )^2 + 2 |f'( \| X \|_F^2 )| \cdot \| H \|_F^2 \\
\leq & ~ 2 \cdot \frac{ 4 }{ \epsilon_{mp} } \| H \|_F^2 + 2 \cdot \frac{1}{\epsilon_{mp}} \| H \|_F^2 \\
= & ~ \frac{10}{ \epsilon_{mp} } \| H \|_F^2. 
\end{align*}

\end{proof}

\subsection{$x$ and $\ov{x}$ are close}
\begin{lemma}[$x$ and $\ov{x}$ are close in term of $\wt{V}^{-1}$]\label{lem:accuracy_of_ov_x_respect_to_x}
With probability $1-\delta$ over the randomness of sketching matrix $R \in \R^{b \times n}$, we have 
\begin{align*}
\| \ov{x}_i - x_i \|_{ \wt{V}^{-1}_i } \leq \epsilon_x
\end{align*}
$\epsilon_x = O( \alpha \log^2 ( n / \delta ) \cdot \frac{ n^{1/4} }{ \sqrt{b} } )$, $b$ is the size of sketching matrix.
\end{lemma}
\begin{proof}

Recall the definition of $\wt{\delta}_x$ and $\delta_x$, we have
\begin{align*}
\wt{\delta}_{x,i} - \delta_{x,i} = \wt{V}_i^{1/2} ( I - R^\top R \wt{P} ) \wt{V}^{1/2} h - \wt{V}_i^{1/2} ( I - \wt{P} ) \wt{V}^{1/2} h = \wt{V}_i^{1/2} ( \wt{P} - R^\top R \wt{P} ) \wt{V}^{1/2} h
\end{align*}
For iteration $t$, the definition should be
\begin{align*}
\wt{\delta}_{x,i}^{(t)} - \delta_{x,i}^{(t)} = ( \wt{V}_i^{(t)} )^{1/2} ( \wt{P}^{(t)} - ( R^{(t)} )^\top R^{(t)} \wt{P}^{(t)} ) ( \wt{V}^{(t)} )^{1/2} h .
\end{align*}

For any $i$, let $k$ be the current iteration, $k_i$ be the last when we changed the $\wt{V}_i$. Then, we have that
\begin{align*}
x_i^{(k)} - \ov{x}_i^{(k)} = \sum_{t = k_i}^k \wt{\delta}_{x,i}^{(t)} - \delta_{x,i}^{(t)}
\end{align*}
because we have $x_i^{(k_i)} = \ov{x}_i^{(k_i)}$ (guaranteed by our algorithm). Since $\wt{V}_i^{(t)}$ did not change during iteration $k_i$ to $k$ for the block $i$. (However, the whole other parts of matrix $\wt{V}$ could change). We consider
\begin{align*}
( x_i^{(k)} - \ov{x}_i^{(k)} )^\top \cdot ( \wt{V}_i^{(k)} )^{-1} \cdot ( x_i^{(k)} - \ov{x}_i^{(k)} )
= & ~ \left( \sum_{t = k_i}^k \wt{\delta}_{x,i}^{(t)} - \delta_{x,i}^{(t)} \right)^\top \cdot ( \wt{V}_i^{(k)} )^{-1} \cdot \left( \sum_{t = k_i}^k \wt{\delta}_{x,i}^{(t)} - \delta_{x,i}^{(t)} \right) \\
= & ~ \left\| \sum_{t = k_i}^k \left( ( I - ( R^{(t)})^\top R^{(t)} ) \wt{P}^{(t)} ( \wt{V}^{(t)} )^{1/2} h^{(t)} \right)_i \right\|_2^2 .
\end{align*}

We consider block $i$ and a coordinate $j \in $ block $i$. We define random vector $X_t \in \R^{n_i}$ as follows:
\begin{align*}
X_t = \left( ( I - R^{(t) \top} R^{(t)} ) \wt{P}^{(t)} ( \wt{V}^{(t)} )^{1/2} h^{(t)} \right)_i.
\end{align*}  
Let $(X_t)_j$ denote the $j$-th coordinate of $X_t$, for each $j \in [n_i]$.

By Lemma~\ref{lem:sketch_vector} in Section~\ref{sec:fastjl}, we have for each $t$,
\begin{align*}
\E[ X_t ] = 0, \text{~~~and~~~} \E[ ( X_t )_j^2 ] = \frac{1}{b} \| ( \wt{P}^{(t)} ( \wt{V}^{(t)} )^{1/2} h^{(t)} )_i \|_2^2
\end{align*}
and with probability $1-\delta$,
\begin{align*}
| (X_t)_j | \leq \| ( \wt{P}^{(t)} ( \wt{V}^{(t)} )^{1/2} h^{(t)} )_i \|_2 \frac{ \log ( n / \delta ) }{ \sqrt{b} } := M.
\end{align*}

Now, we apply Bernstein inequality (Lemma~\ref{lem:bernstein_inequality}),
\begin{align*}
\Pr \left[ \sum_t (X_t)_j > \tau \right] \leq \exp \left( - \frac{ \tau^2 / 2 }{ \sum_{t} \E[ (X_t)_j^2 ] + M \tau / 3} \right)
\end{align*}

Choosing $\tau = 10^3 \frac{ \sqrt{T} }{ \sqrt{b} } \log^2 (n/\delta) \cdot \| ( \wt{P}^{(t)} ( \wt{V}^{(t)} )^{1/2} h^{(t)} )_i \|_2$
\begin{align*}
 & ~ \Pr \left[ \sum_t (X_t)_j > 10^3 \frac{ \sqrt{T} }{ \sqrt{b} } \log^2(n/\delta) \cdot \| ( \wt{P}^{(t)} ( \wt{V}^{(t)} )^{1/2} h^{(t)} )_i \|_2 \right] \\
\leq & ~ \exp \left( - \frac{ 10^6 \frac{T}{b} \log^4 (n/\delta) \cdot \| ( \wt{P}^{(t)} ( \wt{V}^{(t)} )^{1/2} h^{(t)} )_i \|_2^2 / 2 }{ \frac{T}{b} \| ( \wt{P}^{(t)} ( \wt{V}^{(t)} )^{1/2} h^{(t)} )_i \|_2^2 + 10^3 \frac{ \sqrt{T} }{ b } \log^3 (n/\delta) \| ( \wt{P}^{(t)} ( \wt{V}^{(t)} )^{1/2} h^{(t)} )_i \|_2^2 / 3 } \right) \\
\leq & ~ \exp( - 100 \log(n/\delta) )
\end{align*}

Now, taking a union, we have
\begin{align*}
\left\| \sum_{t = k_i}^k \left( ( I - ( R^{(t)})^\top R^{(t)} ) \wt{P}^{(t)} ( \wt{V}^{(t)} )^{1/2} h^{(t)} \right)_i \right\|_2 
= & ~ O \left( \frac{\sqrt{T}}{ \sqrt{b} } \log^2 (n/\delta) \left\| ( \wt{P}^{(t)} ( \wt{V}^{(t)} )^{1/2} h^{(t)} )_i \right\|_2 \right) \\
\leq & ~ O \left( \frac{ \sqrt{T} }{ \sqrt{b} } \log^2 ( n / \delta ) \alpha \right)
\end{align*}
where we use that $\| (\wt{P}^{(t)} ( \wt{V}^{(t)} )^{1/2} h^{(t)} )_i \|_2 \leq \| ( ( \wt{V}^{(t)} )^{1/2} h^{(t)} )_i \|_2 =O(\alpha )$, $n_i=O(1)$.

Finally, we use the fact that the algorithm reset $\ov{x} = x$, $\ov{s} = s$ in less than $\sqrt{n}$ iterations. 
% and choosing $\delta = 1/\poly(n)$. Therefore, we complete the proof.

\end{proof}

\subsection{$s$ and $\ov{s}$ are close}

\begin{lemma}[$s$ and $\ov{s}$ are close]\label{lem:accuracy_of_ov_s_respect_to_s}
With probability $1 - \delta $ over the randomness of sketching matrix $R \in \R^{b \times n}$, we have
\begin{align*}
t^{-1} \| \ov{s}_i - s_i \|_{ \wt{V}_i } \leq \epsilon_s,
\end{align*}
$\epsilon_s = O( \alpha \log^2 (n / \delta) \cdot  \frac{ n^{1/4} }{ \sqrt{b} } $, and $b$ is the size of sketching matrix.
\end{lemma}
\begin{proof}

Recall the definition of $\wt{\delta}_s$, $\delta_s$, we have
\begin{align*}
\wt{\delta}_{s,i} - \delta_{s,i} = t \wt{V}_i^{-1/2} ( R^\top R - I ) \wt{P} \wt{V}^{1/2} h
\end{align*}
The rest of the proof is identical to Lemma~\ref{lem:accuracy_of_ov_x_respect_to_x} except we use also the fact we make $\ov{s} = s$ whenever our $t$ changed by a constant factor. We omitted the details here.
\end{proof}

\subsection{Data structure is maintaining $(x,s)$ implicitly over all the iterations}

\begin{lemma}
Over all the iterations, $u_1 + F u_2$ is always maintaining $x$ implicitly, $u_3 + G u_4$ is always maintaining $s$ implicitly.
\end{lemma}
\begin{proof}

We only focus on the \textsc{PartialUpdate}. The \textsc{FullUpdate} is trivial, we ignore the proof. 

{\bf For $x$.}

Note that $M$ is not changing. Let's assume that $u_1 + F u_2 = x$, we want to show that
\begin{align*}
u_1^{\new} + F^{\new} u_2^{\new} = x^{\new}.
\end{align*}
which is equivalent to prove
\begin{align*}
u_1^{\new} + F^{\new} u_2^{\new} - (u_1 + F u_2) = \delta_x
\end{align*}

Let $\Delta u_1 = u_1^{\new} - u_1$ be the change of $u_1$ over iteration $t$, then
\begin{align*}
\Delta u_1 
= & ~ \wt{V}^{\new} h + ( F - F^{\new} ) u_2
\end{align*}
Let $\Delta u_2 = u_2^{\new} - u_2$ be the change of $u_2$ over iteration $t$, then
\begin{align*}
\Delta u_2 =  - ( \wt{V}^{\new} )^{1/2} h + {\bf 1}_{ \wt{S} } ( \wt{\Delta}_{ \wt{S} , \wt{S} }^{-1} + M_{ \wt{S} , \wt{S} } )^{-1} M_{\wt{S}}^\top ( \wt{V}^{\new} )^{1/2} h .
\end{align*} 

By definition of $\delta_x$ at iteration $t$, we have
\begin{align*}
\delta_x = \wt{V}^{\new} h - \left( \sqrt{ \wt{V}^{\new} } M \sqrt{ \wt{V}^{\new} } h - \sqrt{\wt{V}^{\new} } M_{\wt{S}} ( \wt{\Delta}_{\wt{S},\wt{S}}^{-1} + M_{\wt{S},\wt{S}} )^{-1} ( M_{\wt{S}} )^\top \sqrt{ \wt{V}^{\new} } h \right).
\end{align*} 

We can compute 
\begin{align*}
  & ~ u_1^{\new} + F^{\new} u_2^{\new} - (u_1 + F u_2) \\
 = & ~ \Delta u_1 + (F^{\new} u_2^{\new} - F u_2  )\\
 = & ~ \wt{V}^{\new} h + ( F - F^{\new} ) u_2 + (F^{\new} u_2^{\new} - F u_2 ) \\
 = & ~ \wt{V}^{\new} h + F^{\new} (u_2^{\new} - u_2) \\
 = & ~ \wt{V}^{\new} h + F^{\new} \Delta u_2 \\
 = & ~ \wt{V}^{\new} h - F^{\new} \sqrt{ \wt{V}^{\new} } h + F^{\new} {\bf 1}_{\wt{S}} ( \wt{\Delta}_{\wt{S},\wt{S}}^{-1} + M_{\wt{S},\wt{S}} )^{-1} ( M_{\wt{S}} )^\top \sqrt{ \wt{V}^{\new} } h \\
 = & ~ \delta_x
\end{align*}
where we used $F^{\new} = \sqrt{ \wt{V}^{\new} } M$ in the last step.

{\bf For $s$.}

We have
\begin{align*}
G^{\new} = \frac{1}{ \sqrt{ \wt{V}^{\new} } } M, G = \frac{1}{\sqrt{ \wt{V} }} M
\end{align*}

Let $\Delta u_3 = u_3^{\new} - u_3$ be the change of $u_3$ over iteration $t$, then
\begin{align*}
\Delta u_3 = (G - G^{\new}) u_4
\end{align*}
Let $\Delta u_4 = u_4^{\new} - u_4$ be the change of $u_4$ over iteration $t$, then
\begin{align*}
\Delta u_4 = t \cdot \Delta u_2
\end{align*}

By definition of $\delta_s$ in iteration $t$,

\begin{align*}
\delta_s =  \left( \frac{1}{ \sqrt{ \wt{V}^{\new} } } M \sqrt{ \wt{V}^{\new} }  (t h) - \frac{1}{ \sqrt{\wt{V}^{\new} } } M_{\wt{S}} ( \wt{\Delta}_{\wt{S},\wt{S}}^{-1} + M_{\wt{S},\wt{S}} )^{-1} ( M_{\wt{S}} )^\top \sqrt{ \wt{V}^{\new} } (t h) \right)
\end{align*}

We can compute
\begin{align*}
 ( u_3^{\new} + G^{\new} u_4^{\new} ) - ( u_3 + G u_4 ) 
= & ~ \Delta u_3  + ( G^{\new} u_4^{\new} - G u_4 ) \\
= & ~ (G - G^{\new}) u_4 + ( G^{\new} u_4^{\new} - G u_4 ) \\
= & ~ G^{\new} ( u_4^{\new} - u_4 ) \\
= & ~ G^{\new} t \Delta u_2 \\
= & ~ \delta_s
\end{align*}
where the last step follows by definition of $\Delta u_2$.

\end{proof}

\begin{comment}
x and y defined in ideal W, pi and tau defined idea W

algorithm only use ov{W}, ov{y}

the first equality is wrong,

y_i^{k} = w_i^{k+1} / v_i^{k} - I
x_i^{k} = w_i^{k} / v_i^{k} -I

tau is sorting x^{(k+1)}, pi is sorting y^{(k)}

Hence, we have psi( y )

err(W,ov{W})
\end{comment}

\newpage
\section{Combining Robust Central Path with Data Structure}\label{sec:combine}
The goal of this section is to combine Section~\ref{sec:robust_central_path} and Section~\ref{sec:central_path_maintenance}.

\begin{table}[!h]
\begin{center}
    \begin{tabular}{ | l | l | l | l | }
    \hline
    Notation & Choice of Parameter & Statement & Comment \\ \hline
    $C_1$ & $\Theta(1/\log^2 n)$ & Lem.~\ref{lem:guarantee_of_a_sequence_of_W}, Thm.~\ref{thm:central_path_maintenance} & $\ell_2$ accuracy of $W$ sequence \\ \hline
    $C_2$ & $\Theta(1/\log^4 n)$ & Lem.~\ref{lem:guarantee_of_a_sequence_of_W}, Thm.~\ref{thm:central_path_maintenance} & $\ell_4$ accuracy of $W$ sequence \\ \hline
    $\epsilon_{mp}$ & $\Theta(1/\log^2 n)$ & \textsc{RobustIPM} Alg in Sec.~\ref{sec:robust_central_path} & accuracy for data structure \\ \hline
    $T$ & $\Theta(\sqrt{n} \log^2 n \log (n/\delta) )$ & Thm.~\ref{thm:roubst_IPM} & \#iterations \\ \hline
    $\alpha$ & $\Theta(1/\log^2 n)$ & \textsc{RobustIPM} Alg in Sec.~\ref{sec:robust_central_path} & step size in Hessian norm \\ \hline
    $b$ & $\Theta(\sqrt{n} \log^6 (nT)$ & Lem.~\ref{lem:accuracy_of_ov_x_respect_to_x}, Lem.~\ref{lem:accuracy_of_ov_s_respect_to_s}, Lem.~\ref{lem:accuracy_of_ov_W_respect_to_W} & sketch size \\ \hline
    $\epsilon_x$ & $\Theta( 1/\log^3 n )$ & Lem.~\ref{lem:accuracy_of_ov_x_respect_to_x} & accuracy of $\ov{x}$ (respect to $x$)  \\ \hline
    $\epsilon_s$ & $\Theta( 1/\log^3 n )$ & Lem.~\ref{lem:accuracy_of_ov_s_respect_to_s} & accuracy of $\ov{s}$ (respect to $s$) \\ \hline
    $\epsilon_w$ & $\Theta(1/\log^3 n)$ & Lem.~\ref{lem:accuracy_of_ov_W_respect_to_W} & accuracy of $\ov{W}$ (respect to $W$) \\ \hline
    $a$ & $\min (2/3, \alpha_m)$ & $\alpha_m$ is the dual exponent of MM & batch size \\ \hline
    \end{tabular}
\end{center}\caption{Summary of parameters}\label{tab:summary_of_parameters_in_combine}
\end{table}

\subsection{Guarantee for $W$ matrices}

\begin{lemma}[Guarantee of a sequence of $W$]\label{lem:guarantee_of_a_sequence_of_W}
Let $x^{\new} = x + \delta_x$. Let $W^{\new} = ( \nabla^2 \phi ( x^{\new} ) )^{-1}$ and $W = ( \nabla^2 \phi ( x ) )^{-1} $. Then we have
\begin{align*}
\sum_{i=1}^m \left\| w_i^{-1/2} ( w_i^{\new} - w_i ) w_i^{-1/2}  \right\|_F^2 \leq & ~ C_1^2 , \\
\sum_{i=1}^m \left\| w_i^{-1/2} ( w_i^{\new} - w_i ) w_i^{-1/2}  \right\|_F^4 \leq & ~ C_2^2 , \\
\left\| w_i^{-1/2} ( w_i^{\new} - w_i ) w_i^{-1/2} \right\|_F \leq & ~ \frac{1}{4} .
\end{align*}
where $C_2 = \Theta(\alpha^2)$ and $C_1 = \Theta( \alpha )$.
\end{lemma}

\begin{proof}

For each $i\in [m]$, we have
\begin{align*}
& ~ \left\| W_i^{-1/2} ( W_i^{\new} - W_i ) W_i^{-1/2}  \right\|_F^2 \\
= & ~ n_i \left\| W_i^{-1/2} ( W_i^{\new} - W_i ) W_i^{-1/2}  \right\|^2 \\
= & ~ n_i \left\| (\nabla^2 \phi (x_i) )^{1/2} ( \nabla^2 \phi (x_i^{\new})^{-1} - \nabla^2 \phi ( x_i )^{-1}  ) (\nabla^2 \phi (x_i) )^{1/2} \right\|^2 \\
\leq & ~ \left( \frac{1}{ ( 1 -  \| x_i^{\new} - x_i \|_{\nabla^2 \phi(x_i)} )^2 } - 1 \right)^2 \cdot \left\| (\nabla^2 \phi ( x_i ) )^{1/2}  \nabla^2 \phi ( x_i )^{-1} (\nabla^2 \phi (x_i) )^{1/2} \right\|^2 \\
= & ~ n_i \left( \frac{1}{ ( 1 -  \| x_i^{\new} - x_i \|_{\nabla^2 \phi(x_i)} )^2 } - 1 \right)^2 \\
\leq & ~ 100 n_i \| x_i^{\new} - x_i \|_{ \nabla^2 \phi( x_i ) }^2,
\end{align*}
where the second step follows by Theorem~\ref{thm:hessiansc}.

In our problem, we assume that $n_i = O(1)$. It remains to bound
\begin{align*}
\| x_i^{\new} - x_i \|_{ \nabla^2 \phi (x_i) }^2 
=  \| \delta_{x,i} \|_{ \nabla^2 \phi (x_i) }^2
\lesssim \| \delta_{x,i} \|_{\ov{x}_i}^2 
= \alpha_i^2
\end{align*}
where the last step follows from definition $\alpha_i = \| \delta_{x,i} \|_{\ov{x}_i}$.

Then, we have
\begin{align*}
\sum_{i=1}^m \| x_i^{\new} - x_i \|_{ \nabla^2 \phi (x_i) }^2 \leq \sum_{i=1}^m O( \alpha_i^2 ) \leq O( \alpha^2 ).
\end{align*}
where the last step follows by Lemma~\ref{lem:alpha_i}.
\end{proof}

\begin{lemma}[Accuracy of $\ov{W}$]\label{lem:accuracy_of_ov_W_respect_to_W}
Let $x$ and $\ov{x}$ be the vectors maintained by data-structure \textsc{StochasticProjectionMaintenance}. Let $W = ( \nabla^2 \phi(x) )^{-1}$ and $\ov{W} = ( \nabla^2 \phi ( \ov{x} ) )^{-1}$. Then we have
\begin{align*}
\| w_i^{-1/2} ( \ov{w}_i - w_i ) w_i^{-1/2} \|_F \leq \epsilon_w,
\end{align*}
where $\epsilon_w = O \Big( \alpha \log^2 ( n T ) \cdot \frac{ n^{1/4} }{ \sqrt{b} } \Big) $, $b$ is the size of sketching matrix.
\end{lemma}
\begin{proof}
By similar calculation, we have
\begin{align*}
\| w_i^{-1/2} ( \ov{w}_i - w_i ) w_i^{-1/2} \|_F = O(1) \cdot \| \ov{x}_i - x_i \|_{ \nabla^2 \phi( x_i ) }.
\end{align*}
Then, using Lemma~\ref{lem:accuracy_of_ov_x_respect_to_x} with $\delta = 1/T$
\begin{align*}
\| \ov{x}_i - x_i \|_{ \nabla^2 \phi( x_i ) } 
\leq  O \left( \alpha \log^2 ( n T ) \cdot \frac{ \sqrt{ n^{1/4} } }{ \sqrt{ b } } \right) .
\end{align*}
\end{proof}

\begin{algorithm}[t]\caption{Robust Central Path}\label{alg:robust_central_path}
\begin{algorithmic}[1]
\Procedure{\textsc{CentralPathStep}}{$\ov{x},\ov{s},t,\lambda,\alpha$} 
    
    \State{\bf for} {$i = 1 \to m$} {\bf do} \Comment{Figure out direction $h$}
        \State \hspace{4mm} $\mu_i^t \leftarrow \ov{s}_i/t + \nabla \phi_i ( \ov{x}_i )$ \Comment{According to Eq.~\eqref{eq:def_mu_i_t}}
        \State \hspace{4mm} $\gamma_i^t \leftarrow \| \mu_i^t \|_{\nabla^2 \phi_i(\ov{x}_i)^{-1} } $ \Comment{According to Eq.~\eqref{eq:def_gamma_i_t}}
        \State \hspace{4mm} $c_i^t \leftarrow  \frac{ \exp(\lambda \gamma_i^t ) / \gamma_i^t }{ ( \sum_{i=1}^m \exp( 2 \lambda \gamma_i^t ) )^{1/2} }$ if $\gamma_i^t \geq 96 \sqrt{\alpha}$ and $c_i^t \leftarrow 0$ otherwise \Comment{According to Eq.~\eqref{eq:def_c_i_t}}
        \State \hspace{4mm} $h_i \leftarrow - \alpha \cdot c_i^t  \cdot \mu_i^t  $ \Comment{According to Eq.~\eqref{eq:def_h_i_t}}
    \State {\bf end for}
    \State $\ov{W} \leftarrow ( \nabla^2 \phi ( \ov{x} ) )^{-1}$ \Comment{Computing block-diagonal matrix $\ov{W}$}
    \State \Return $h,\ov{W}$
\EndProcedure
\State
\Procedure{\textsc{RobustCentralPath}}{$\text{mp},t,\lambda,\alpha$} \Comment{Lemma~\ref{lem:robust_central_path_part3}}
    \State \Comment{Standing at $(x,s)$ implicitly via data-structure}
    \State \Comment{Standing at $(\ov{x},\ov{s})$ explicitly via data-structure}
    \State $(\ov{x}, \ov{s}) \leftarrow \text{mp}.\textsc{Query}()$ \Comment{ Algorithm~\ref{alg:central_path_maintenance_main}, Lemma~\ref{lem:central_path_maintenance_query}}
    \State
    \State $h, \ov{W} \leftarrow \textsc{CentralPathStep}(\ov{x},\ov{s},t,\lambda,\alpha)$
    \State
    \State $\text{mp}.\textsc{Update} ( \ov{W} )$ \Comment{Algorithm~\ref{alg:central_path_maintenance_update}, Lemma~\ref{lem:central_path_maintenance_update}}
    \State $\text{mp}.\textsc{MultiplyMove}( h , t )$ \Comment{Algorithm~\ref{alg:central_path_maintenance_multiply_move}, Lemma~\ref{lem:central_path_maintenance_multiply}, Lemma~\ref{lem:central_path_maintenance_move}}
    \State \Comment{$x \leftarrow x + \delta_x$, $s \leftarrow s + \delta_s$, achieved by data-structure implicitly}
    \State \Comment{$\ov{x} \leftarrow \ov{x} + \wt{\delta}_x$, $\ov{s} \leftarrow \ov{s} + \wt{\delta}_s$, achieved by data-structure explicitly}
    \State \Comment{If $x$ is far from $\ov{x}$, then $\ov{x} \leftarrow x$}
\EndProcedure
\end{algorithmic}
\end{algorithm}

\begin{algorithm}[!h]\caption{Our main algorithm (More detailed version of \textsc{RobustIPM} in Section~\ref{sec:robust_central_path_intro})}\label{alg:main}
\begin{algorithmic}[1]
\Procedure{\textsc{Main}}{$A,b,c,\phi,\delta$} \Comment{Theorem~\ref{thm:main_result_informal}, Theorem~\ref{thm:main_result_formal}}
    \State $\lambda \leftarrow 2^{16} \log (m)$, $\alpha \leftarrow 2^{-20} \lambda^{-2}$ , $\kappa \leftarrow 2^{-10} \alpha$
    \State $\delta \leftarrow \min( \frac{1}{ \lambda } , \delta )$ \Comment{Choose the target accuracy}
    \State $a \leftarrow \min (2/3,\alpha_m)$ \Comment{Choose the batch size}
    \State $b_{\text{sketch}} \leftarrow 2^{10}  \sqrt{\nu} \log^6 (n/\delta) \cdot \log \log (1/\delta) $ \Comment{Choose the size of sketching matrix}
    \State Modify the ERM($A,b,c,\phi$) and obtain an initial $x$ and $s$
    \State \textsc{CentralPathMaintenance} $\text{mp}$ \Comment{Algorithm~\ref{alg:central_path_maintenance_main}, Theorem~\ref{thm:central_path_maintenance}}
    \State $\text{mp}.\textsc{Initialize}(A,x,s,\alpha,a,b_{\text{sketch}})$ \Comment{Algorithm~\ref{alg:central_path_maintenance_main}, Lemma~\ref{lem:central_path_maintenance_initialize}}
    \State $\nu \leftarrow \sum_{i=1}^m \nu_i$ \Comment{$\nu_i$ are the self-concordant parameters of $\phi_i$}
    \State $t \leftarrow 1$
    \State {\bf while} {$t > \delta^2 / (4 \nu)$} {\bf do}
        \State \hspace{8mm} $t^{\new} \leftarrow (1 - \frac{ \kappa }{  \sqrt{\nu} } ) t$
        \State \hspace{8mm} $\textsc{RobustCentralPath}(\text{mp},t,\lambda,\alpha)$ \Comment{Algorithm~\ref{alg:robust_central_path}}
        \State \hspace{8mm} $t \leftarrow t^{\new}$
    \State {\bf end while}
    \State Return an approximate solution of the original ERM according to Section~\ref{sec:initial_point_termination_condition} %Lemma A.6 in \cite{cls18}
\EndProcedure
\end{algorithmic}
\end{algorithm}

\subsection{Main result}
The goal of this section is to prove our main result.
\begin{theorem}[Main result, formal version of Theorem~\ref{thm:main_result_informal}]\label{thm:main_result_formal} 
Consider a convex problem
\begin{align*}
\min_{Ax=b,x \in \prod_{i=1}^{m} K_{i} } c^{\top} x
\end{align*}
where $K_{i}$ are
compact convex set. For each $i \in [m]$, we are given a $\nu_{i}$-self
concordant barrier function $\phi_{i}$ for $K_{i}$. Also, we are
given $x^{(0)}=\arg\min_{x}\sum_{i}\phi_{i}(x_{i})$. Assume that
\begin{enumerate}
\item Diameter of the set: For any $x\in\prod_{i=1}^m K_{i}$, we have that $\|x\|_{2}\leq R$.
\item Lipschitz constant of the program: $\|c\|_{2}\leq L$.
\end{enumerate}
Then, the algorithm $\textsc{Main}$ finds a vector $x$ such
that 
\begin{align*}
c^{\top}x & \leq\min_{Ax=b, x \in \prod_{i=1}^m K_{i}} c^{\top}x+LR\cdot\delta,\\
\|Ax-b\|_{1} & \leq3\delta\cdot\left(R\sum_{i,j}|A_{i,j}|+\|b\|_{1}\right),\\
x & \in \prod_{i=1}^m K_{i}.
\end{align*}
in time %$O(\sqrt{\nu}\log m\log(\frac{\nu}{\delta}))$ iterations.
\begin{align*}
O  ( n^{\omega+o(1)} + n^{2.5-\alpha/2 + o(1)} + n^{2+1/6 + o(1) } ) \cdot \wt{O}( \log ( n / \delta ) ).
\end{align*}
where $\omega$ is the exponent of matrix multiplication \cite{w12,l14}, and $\alpha$ is the dual exponent of matrix multiplication \cite{gu18}.
\end{theorem}
\begin{proof}
The number of iterations is 
\begin{align*}
O( \sqrt{\nu} \log^2 (m) \log (\nu /\delta) ) = O( \sqrt{n} \log^2 (n) \log ( n / \delta ) ).
\end{align*} For each iteration, the amortized cost per iteration is
\begin{align*}
& ~ O(nb + n^{1+a} + n^{1.5}) + O( C_1 / \epsilon_{mp} + C_2 / \epsilon_{mp}^2 ) \cdot ( n^{\omega-1/2 + o(1)} + n^{2-a/2+o(1)} ) + O(n^{\omega-1/2+o(1)}) \\
= & ~ O(nb + n^{1+a} + n^{1.5}) + O( \alpha + \alpha^2 ) \cdot ( n^{\omega-1/2 + o(1)} + n^{2-a/2+o(1)} ) + O(n^{\omega-1/2+o(1)}) \\
= & ~ O(nb + n^{1+a} + n^{1.5}) + O( 1 / \log^4 n ) \cdot ( n^{\omega-1/2 + o(1)} + n^{2-a/2+o(1)} ) + O(n^{\omega-1/2+o(1)}) \\
= & ~ O( n^{1.5 + o(1)} \log^6 \log (1/\delta) + n^{1+a + o(1)} ) + O( n^{\omega-1/2 + o(1)} + n^{2-a/2+o(1)} ) .
\end{align*}
where the last step follows from choice of $b$ (see Table~\ref{tab:summary_of_parameters_in_combine}).

Finally, we have
\begin{align*}
 & ~ \text{total~time} \\
= & ~ \text{\#iterations} \cdot \text{cost~per~iteration} \\
= & ~ \underbrace{ O\left( \sqrt{n} \log^2 n \log (n / \delta) \right) }_{ \text{\#iterations} } \cdot \underbrace{ O \left( n^{1.5+o(1)} \log^6 \log(1/\delta) + n^{1+a + o(1)} +  n^{\omega-1/2 + o(1)} + n^{2-a/2+o(1)} \right) }_{ \text{cost~per~iteration} } \\
= & ~ O \left( n^{1.5+a + o(1)} + n^{\omega+o(1)} + n^{2.5-a/2+o(1)} \right) \cdot \log (n/\delta) \cdot \log^6 \log (1/\delta) \\
= & ~ O \left( n^{2+1/6 + o(1)} + n^{\omega+o(1)} + n^{2.5-\alpha_m/2+o(1)} \right) \cdot \log (n/\delta) \cdot \log^6 \log (1/\delta)
\end{align*}
where we pick $a=\min (2/3,\alpha_m)$ and $\alpha_m$ is the dual exponent of matrix multiplication\cite{gu18}.
 
Thus, we complete the proof.
\end{proof}

%\Zhao{The following corollary and its proof is new, Yintat needs to check.}
%\Richard{Think $f_i$ can just need to admit self-concordant barrier and not be smooth}
\begin{corollary}[Empirical risk minimization]
Given convex function $f_i (y) : \R \rightarrow \R$. Suppose the solution $x^* \in \R^d$ lies in $\ell_{\infty}$-$\mathrm{Ball}(0,R)$. Suppose $f_i$ is $L$-Lipschitz in region $\{ y : |y| \leq 4 \sqrt{n} \cdot M \cdot R \}$. Given a matrix $A \in \R^{d \times n }$ with $\| A \| \leq M$ and $A$ has no redundant constraints, and a vector $b \in \R^d$ with $\| b \|_2 \leq M \cdot R$. We can find $x \in \R^d$ s.t.
\begin{align*}
\sum_{i=1}^n f_i(a_i^\top x + b_i) \leq \min_{x\in \R^d} \sum_{i=1}^n f_i (a_i^\top x + b_i) + \delta M R
\end{align*}
in time
\begin{align*}
O  ( n^{\omega+o(1)} + n^{2.5-\alpha/2 + o(1)} + n^{2+1/6 + o(1) } ) \cdot \wt{O}( \log ( n / \delta ) ).
\end{align*}
where $\omega$ is the exponent of matrix multiplication \cite{w12,l14}, and $\alpha$ is the dual exponent of matrix multiplication \cite{gu18}.
\end{corollary}
\begin{proof}
It follows from applying Theorem~\ref{thm:main_result_formal} on convex program \eqref{eq:apply_main_theorem_and_get_ERM} with an extra constraint $x^*$ lies in $\ell_{\infty}$-$\mathrm{Ball}(0,R)$. Note that in program \eqref{eq:apply_main_theorem_and_get_ERM}, $n_i = 2$. Thus $m = O(n)$.
\end{proof}

\section{Initial Point and Termination Condition}\label{sec:initial_point_termination_condition}

We first need some result about self concordance.

\begin{lemma}[Theorem 4.1.7, Lemma 4.2.4 in \cite{n98}]\label{lem:self_concordant}
Let $\phi$ be any $\nu$-self-concordant barrier. Then, for any $x,y\in\mathrm{dom}\phi$,
we have
\begin{align*}
\left\langle \nabla\phi(x),y-x\right\rangle  & \leq\nu,\\
\left\langle \nabla\phi(y)-\nabla\phi(x),y-x\right\rangle  & \geq\frac{\|y-x\|_{x}^{2}}{1+\|y-x\|_{x}}.
\end{align*}
Let $x^{*}=\arg\min_{x}\phi(x)$. For any $x\in\R^{n}$ such that
$\|x-x^{*}\|_{x^{*}}\leq1$, we have that $x\in\mathrm{dom}\phi$.
\[
\|x^{*}-y\|_{x^{*}}\leq\nu+2\sqrt{\nu}.
\]

\end{lemma}

\begin{lemma}\label{lem:feasible_LP} Consider a convex problem $\min_{Ax=b,x\in\prod_{i=1}^{m}K_{i}}c^{\top}x$
where $K_{i}$ are compact convex set. For each $i \in [m]$, we are given
a $\nu_{i}$-self concordant barrier function $\phi_{i}$ for $K_{i}$.
Also, we are given $x^{(0)}=\arg\min_{x}\sum_{i}\phi_{i}(x_{i})$.
Assume that
\begin{enumerate}
\item Diameter of the set: For any $x\in\prod_{i=1}^m K_{i}$, we have that $\|x\|_{2}\leq R$.
\item Lipschitz constant of the program: $\|c\|_{2}\leq L$.
\end{enumerate}
For any $\delta>0$, the modified program $\min_{\overline{A}\overline{x}=\overline{b},\overline{x}\in \prod_{i=1}^m K_{i}\times\R_{+}}\overline{c}^{\top}\overline{x}$
with 
\[
\overline{A}=[A\ |\ b-Ax^{(0)}],\overline{b}=b\text{, and }\overline{c}=\left[\begin{array}{c}
\frac{\delta}{LR}\cdot c\\
1
\end{array}\right]
\]
satisfies the following:
\begin{enumerate}
\item $\overline{x}=\left[\begin{array}{c}
x^{(0)}\\
1
\end{array}\right]$, $\overline{y}=0_{d}$ and $\overline{s}=\left[\begin{array}{c}
\frac{\delta}{LR}\cdot c\\
1
\end{array}\right]$ are feasible primal dual vectors with $\|\overline{s}+\nabla\overline{\phi}(\overline{x})\|_{\overline{x}}^{*}\leq\delta$
where $\overline{\phi}(\overline{x})=\sum_{i=1}^{m}\phi_{i}(\overline{x}_{i})-\log(\overline{x}_{m+1})$.
\item For any $\overline{x}$ such that $\overline{A}\overline{x}=\overline{b},\overline{x}\in\prod_{i=1}^m K_{i}\times\R_{+}$
and $\overline{c}^{\top}\overline{x}\leq\min_{\overline{A}\overline{x}=\overline{b},\overline{x}\in\prod_{i=1}^m K_{i}\times\R_{+}}\overline{c}^{\top}\overline{x}+\delta^{2}$,
the vector $\overline{x}_{1:n}$ ($\overline{x}_{1:n}$ is the first
$n$ coordinates of $\overline{x}$) is an approximate solution to
the original convex program in the following sense 
\begin{align*}
c^{\top}\overline{x}_{1:n} & \leq\min_{Ax=b,x\in \prod_{i=1}^m K_{i}}c^{\top}x+LR\cdot\delta,\\
\|A\overline{x}_{1:n}-b\|_{1} & \leq3\delta\cdot\left(R\sum_{i,j}|A_{i,j}|+\|b\|_{1}\right),\\
\overline{x}_{1:n} & \in\prod_{i=1}^m K_{i}.
\end{align*}
\end{enumerate}
\end{lemma}

\begin{proof} For the first result, straightforward calculations
show that $(\overline{x},\overline{y},\overline{s})$ are feasible.

To compute $\|\overline{s}+\nabla\overline{\phi}(\overline{x})\|_{\overline{x}}^{*}$,
note that
\[
\|\overline{s}+\nabla\overline{\phi}(\overline{x})\|_{\overline{x}}^{*}=\|\frac{\delta}{LR}\cdot c\|_{\nabla^{2}\phi(x^{(0)})^{-1}}.
\]
Lemma \ref{lem:self_concordant} shows that $x\in\R^{n}$ such that
$\|x-x^{(0)}\|_{x^{(0)}}\leq1$, we have that $x \in \prod_{i=1}^m K_{i}$ because
$x^{(0)}=\arg\min_{x}\sum_{i}\phi_{i}(x_{i})$. Hence, for any $v$
such that $v^{\top}\nabla^{2}\phi(x^{(0)})v\leq1$, we have that $x^{(0)}\pm v\in\prod_{i=1}^m K_{i}$
and hence $\|x^{(0)}\pm v\|_{2}\leq R$. This implies $\|v\|_{2}\leq R$
for any $v^{\top}\nabla^{2}\phi(x^{(0)})v\leq1$. Hence, $(\nabla^{2}\phi(x^{(0)}))^{-1}\preceq R^{2}\cdot I$.
Hence, we have
\[
\|\overline{s}+\nabla\overline{\phi}(\overline{x})\|_{\overline{x}}^{*}=\|\frac{\delta}{LR}\cdot c\|_{\nabla^{2}\phi(x^{(0)})^{-1}}\leq\|\frac{\delta}{L}\cdot c\|_{2}\leq\delta.
\]

For the second result, we let $\text{OPT}=\min_{Ax=b,x\in\prod_{i=1}^m K_{i}}c^{\top}x$
and $\overline{\text{OPT}}=\min_{\overline{A}\overline{x}=b,\overline{x}\in \prod_{i=1}^m K_{i}\times\R_{+}}\overline{c}^{\top}\overline{x}$.
For any feasible $x$ in the original problem, $\overline{x}=\left[\begin{array}{c}
x\\
0
\end{array}\right]$ is a feasible in the modified problem. Therefore, we have that 
\[
\overline{\text{OPT}}\leq\frac{\delta}{LR}\cdot c^{\top}x=\frac{\delta}{LR}\cdot\text{OPT}.
\]

Given a feasible $\overline{x}$ with additive error $\delta^{2}$.
Write $\overline{x}=\left[\begin{array}{c}
\overline{x}_{1:n}\\
\tau
\end{array}\right]$ for some $\tau\geq0$. We can compute $\ov c^{\top}\ov x$ which
is $\frac{\delta}{LR}\cdot c^{\top}\ov x_{1:n}+\tau$. Then, we have
\begin{equation}
\frac{\delta}{LR}\cdot c^{\top}\overline{x}_{1:n}+\tau\leq\overline{\text{OPT}}+\delta^{2}\leq\frac{\delta}{LR}\cdot\text{OPT}+\delta^{2}.\label{eq:eq1_appendix}
\end{equation}
Hence, we can upper bound the OPT of the transformed program as follows:
\[
c^{\top}\overline{x}_{1:n}=\frac{LR}{\delta}\cdot\frac{\delta}{LR}c^{\top}\ov x_{1:n}\leq\frac{LR}{\delta}\left(\frac{\delta}{LR}\cdot\text{OPT}+\delta^{2}\right)=\text{OPT}+LR\cdot\delta,
\]
where the second step follows by \eqref{eq:eq1_appendix}.

For the feasibility, we have that $\tau\leq-\frac{\delta}{LR}\cdot c^{\top}\overline{x}_{1:n}+\frac{\delta}{LR}\cdot\text{OPT}+\delta^{2}\leq\delta+\delta+\delta$
because $\text{OPT}=\min_{Ax=b,x\geq0}c^{\top}x\leq LR$ and that
$c^{\top}\overline{x}_{1:n}\leq LR$. The constraint in the new polytope
shows that 
\[
A\overline{x}_{1:n}+(b-Ax^{(0)})\tau=b.
\]
Rewriting it, we have $A\overline{x}_{1:n}-b=(Ax^{(0)}-b)\tau$ and
hence 
\[
\|A\overline{x}_{1:n}-b\|_{1}\leq\|Ax^{(0)}-b\|_{1}\cdot\tau.
\]
\end{proof}

\begin{lemma}\label{lem:apx_center_imply_gap}Let $\phi_{i}(x_{i})$
be a $\nu_{i}$-self-concordant barrier. Suppose we have $\frac{s_{i}}{t}+\nabla\phi_{i}(x_{i})=\mu_{i}$
for all $i\in[m]$, $A^{\top}y+s=c$ and $Ax=b$. Suppose that $\|\mu_{i}\|_{x,i}^{*}\leq1$
for all $i$, we have that 
\[
\left\langle c,x\right\rangle \leq\left\langle c,x^{*}\right\rangle +4t\nu
\]
where $x^{*}=\arg\min_{Ax=b,x\in \prod_{i=1}^m K_{i}}c^{\top}x$ and $\nu = \sum_{i=1}^m \nu_{i}$.

\end{lemma}

\begin{proof}

Let $x_{\alpha}=(1-\alpha)x+\alpha x^{*}$ for some $\alpha$ to be
chosen. By Lemma \ref{lem:self_concordant}, we have that $\left\langle \nabla\phi(x_{\alpha}),x^{*}-x_{\alpha}\right\rangle \leq\nu$.
Hence, we have $\frac{\nu}{1-\alpha}\geq\left\langle \nabla\phi(x_{\alpha}),x^{*}-x\right\rangle $.
Hence, we have
\begin{align*}
\frac{\nu\alpha}{1-\alpha} & \geq\left\langle \nabla\phi(x_{\alpha}),x_{\alpha}-x\right\rangle \\
 & =\left\langle \nabla\phi(x_{\alpha})-\nabla\phi(x),x_{\alpha}-x\right\rangle +\left\langle \mu-\frac{s}{t},x_{\alpha}-x\right\rangle \\
 & \geq\sum_{i=1}^m \frac{\|x_{\alpha,i}-x_{i}\|_{x_{i}}^{2}}{1+\|x_{\alpha,i}-x_{i}\|_{x_{i}}}+\left\langle \mu,x_{\alpha}-x\right\rangle -\frac{1}{t}\left\langle c-A^{\top}y,x_{\alpha}-x\right\rangle \\
 & \geq\sum_{i=1}^m \frac{\alpha^{2}\|x_{i}^{*}-x_{i}\|_{x_{i}}^{2}}{1+\alpha\|x_{i}^{*}-x_{i}\|_{x_{i}}}-\alpha \sum_{i=1}^m \|\mu_{i}\|_{x_{i}}^{*}\|x_{i}^{*}-x_{i}\|_{x_{i}}-\frac{\alpha}{t}\left\langle c,x^{*}-x\right\rangle .
\end{align*}
where we used Lemma \ref{lem:self_concordant} on the second first,
$Ax_{\alpha}=Ax$ on the second inequality. Hence, we have
\[
\frac{\left\langle c,x\right\rangle }{t}\leq\frac{\left\langle c,x^{*}\right\rangle }{t}+\frac{\nu}{1-\alpha}+\sum_{i=1}^m \|\mu_{i}\|_{x_{i}}^{*}\|x_{i}^{*}-x_{i}\|_{x_{i}}-\sum_{i=1}^m \frac{\alpha\|x_{i}^{*}-x_{i}\|_{x_{i}}^{2}}{1+\alpha\|x_{i}^{*}-x_{i}\|_{x_{i}}}.
\]
Using $\|\mu_{i}\|_{x_{i}}^{*}\leq1$ for all $i$, we have
\[
\frac{\left\langle c,x\right\rangle }{t}\leq\frac{\left\langle c,x^{*}\right\rangle }{t}+\frac{\nu}{1-\alpha} + \sum_{i=1}^m \frac{\|x_{i}^{*}-x_{i}\|_{x_{i}}}{1+\alpha\|x_{i}^{*}-x_{i}\|_{x_{i}}}\leq\frac{\left\langle c,x^{*}\right\rangle }{t}+\frac{\nu}{1-\alpha}+\frac{m}{\alpha}.
\]
Setting $\alpha=\frac{1}{2}$, we have $\left\langle c,x\right\rangle \leq\left\langle c,x^{*}\right\rangle +2t(\nu+m)\leq\left\langle c,x^{*}\right\rangle +4t\nu$
because the self-concordance $\nu_{i}$ is always larger than $1$.

\end{proof}

\section{Basic Properties of Subsampled Hadamard Transform Matrix}\label{sec:fastjl}
This section provides some standard calculations about sketching matrices, it can be found in previous literatures \cite{psw17}. Usually, the reason for using subsampled randomized Hadamard/Fourier transform \cite{ldfu13} is multiplying the matrix with $k$ vectors only takes $k n \log n$ time. Unfortunately, in our application, the best way to optimize the running is using matrix multiplication directly (without doing any fast Fourier transform \cite{ct65}, or more fancy sparse Fourier transform \cite{hikp12b,hikp12a,price13,ikp14,ik14,ps15,ckps16,k16,k17,nsw19}). In order to have an easy analysis, we still use subsampled randomized Hadamard/Fourier matrix.

\subsection{Concentration inequalities}
We first state a useful for concentration,
\begin{lemma}[Lemma 1 on page 1325 of \cite{lm00}]\label{lem:chi_squared}
Let $X \sim {\cal X}_k^2$ be a chi-squared distributed random variable with $k$ degrees of freedom. Each one has zero mean and $\sigma^2$ variance. Then
\begin{align*}
& ~ \Pr[ X - k \sigma^2 \geq (2 \sqrt{ k t } + 2 t ) \sigma^2 ] \leq \exp( - t ) \\
& ~ \Pr[ k \sigma^2 - X \geq 2 \sqrt{ k t } \sigma^2 ] \leq \exp( - t )
\end{align*}
\end{lemma}

\begin{lemma}[Khintchine's Inequality]
Let $\sigma_1, \cdots, \sigma_n$ be i.i.d. sign random variables, and let $z_1, \cdots, z_n$ be real numbers. Then there are constants $C, C' > 0$ so that
\begin{align*}
\Pr \left[ \left| \sum_{i=1}^n z_i \sigma_i \right| \geq C t \| z \|_2 \right] \leq \exp( - C' t^2 )
\end{align*}
\end{lemma}

\begin{lemma}[Bernstein Inequality]\label{lem:bernstein_inequality}
Let $X_1, \cdots, X_n$ be independent zero-mean random variables. Suppose that $|X_i| \leq M$ almost surely, for all $i$. Then, for all positive $t$,
\begin{align*}
\Pr \left[ \sum_{i=1}^n X_i > t \right] \leq \exp \left( - \frac{ t^2 / 2 }{ \sum_{j=1}^n \E [X_j^2] + M t / 3 } \right)
\end{align*}
\end{lemma}

\subsection{Properties obtained by random projection}

\begin{remark}
The Subsampled Randomized Hadamard Transform \cite{ldfu13} can be defined as $R = S H_n \Sigma \in \R^{b \times}$, where $\Sigma$ is an $n \times n$ diagonal matrix with i.i.d. diagonal entries $\Sigma_{i,i}$ in which $\Sigma_{i,i}=1$ with probability $1/2$, and $\Sigma_{i,i} = -1$ with probability $1/2$. $H_n$ refers to the Hadamard matrix of size $n$, which we assume is a power of $2$. The $b \times n$ matrix $S$ samples $b$ coordinates of $n$ dimensional vector uniformly at random. If we replace the definition of sketching matrix in Lemma~\ref{lem:sketch_vector} by Subsampled Randomized Hadamard Transform and let $\ov{R} = S H_n$, then the same proof will go through.
\end{remark}

\begin{lemma}[Expectation, variance, absolute guarantees for sketching a fixed vector]\label{lem:sketch_vector}
Let $h \in \R^n$ be a fixed vector. Let $\ov{R} \in \R^{b \times n}$ denote a random matrix where each entry is i.i.d. sampled from $+1/\sqrt{b}$ with probability $1/2$ and $-1/\sqrt{b}$ with probability $1/2$. Let $\Sigma \in \R^{n \times n}$ denote a diagonal matrix where each entry is $1$ with probability $1/2$ and $-1$ with probability $1/2$. Let $R = \ov{R} \Sigma$, then we have
\begin{align*}
\E[ R^\top R h ] = h, ~~~ \E[ (R^\top R h)_i^2 ] \leq h_i^2 + \frac{1}{b} \| h \|_2^2, ~~~ \Pr \left[ | ( R^\top R h )_i - h_i | > \| h \|_2 \frac{\log ( n / \delta ) }{ \sqrt{b} } \right] \leq \delta .
\end{align*}
\end{lemma}
\begin{proof}

Let $R_{i,j}$ denote the entry at $i$-th row and $j$-th column in matrix $R \in \R^{b \times n}$. Let $R_{*,i} \in \R^b$ denote the vector in $i$-th column of $R$.

We first show expectation,
\begin{align*}
\E[ (R^\top R h)_i ] = & ~ \E[ \langle R, R_{*,i} h^\top \rangle ] \\
= & ~ \E \left[ \sum_{j=1}^b \sum_{l=1}^n R_{j,l} R_{j,i} h_l \right] \\
= & ~ \E \left[ \sum_{j=1}^b R_{j,i}^2 h_i \right] + \E \left[ \sum_{j=1}^b \sum_{l \in [n] \backslash i} R_{j,l} R_{j,i} h_l \right] \\
= & ~ h_i + 0 \\
= & ~ h_i
\end{align*}

Secondly, we prove the variance is small
\begin{align*}
\E [ ( R^\top R h_i )^2 ]
= & ~ \E [ \langle R , R_{*,i} h^\top \rangle^2 ] \\
= & ~ \E \left[ \left( \sum_{j=1}^b \sum_{l=1}^n R_{j,l} R_{j,i} h_l \right)^2 \right] \\
= & ~ \E \left[ \left( \sum_{j=1}^b R_{j,i}^2 h_i + \sum_{j=1}^b \sum_{l\in [n] \backslash i} R_{j,l} R_{j,i} h_l \right)^2 \right] \\
= & ~ \E \left[ \left( \sum_{j=1}^b R_{j,i}^2 h_i \right)^2 \right] + 2 \E \left[ \sum_{j'=1}^b R_{j',i}^2 h_i \sum_{j=1}^b \sum_{l\in [n] \backslash i} R_{j,l} R_{j,i} h_l  \right] \\
& ~ + \E \left[ \left( \sum_{j=1}^b \sum_{l\in [n] \backslash i} R_{j,l} R_{j,i} h_l \right)^2 \right] \\
= & ~ C_1 + C_2 + C_3,
\end{align*}
where the last step follows from defining those terms to be $C_1, C_2$ and $C_3$. For the term $C_1$, we have
\begin{align*}
C_1 = h_i^2 \E \left[ \left( \sum_{j=1}^b R_{j,i}^2 \right)^2 \right] = h_i^2 \E \left[ \sum_{j=1}^b R_{j,i}^4 +  \sum_{j'\neq j} R_{j,i}^2 R_{j',i}^2 \right] = h_i^2 \left( b \cdot \frac{1}{b^2} + b(b-1) \cdot \frac{1}{b^2} \right) = h_i^2
\end{align*}
For the second term $C_2$,
\begin{align*}
C_2 = 0.
\end{align*}
For the third term $C_3$,
\begin{align*}
C_3 = & ~ \E \left[ \left( \sum_{j=1}^b \sum_{ l \in [n] \backslash i } R_{j,l} R_{j,i} h_l \right)^2 \right] \\
= & ~ \E \left[ \sum_{j=1}^b \sum_{ l \in [n] \backslash i } R_{j,l}^2 R_{j,i}^2 h_l^2 \right] + \E \left[ \sum_{j=1}^b \sum_{ l \in [n] \backslash i } R_{j,l} R_{j,i} h_l \sum_{j' \in [b] \backslash j} \sum_{l' \in [n] \backslash i \backslash l} R_{j',l'} R_{j',i} h_{l'} \right] \\
= & ~ \sum_{j=1}^b \sum_{l\in [n] \backslash i} \frac{1}{d} \frac{1}{d} h_l^2 + 0 \leq \frac{1}{d} \| h \|_2^2
\end{align*}
Therefore, we have
\begin{align*}
\E[ ( R^\top R h )_i^2 ] \leq C_1 + C_2 + C_3 \leq h_i^2 + \frac{1}{b} \| h \|_2^2.
\end{align*}

Third, we prove the worst case bound with high probability. We can write $(R^\top R h)_i - h_i$ as follows
\begin{align*}
(R^\top R h)_i - h_i 
= & ~ \langle R, R_{*,i} h^\top \rangle - h_i \\
= & ~ \sum_{j=1}^b \sum_{l=1}^n R_{j,l} \cdot R_{j,i} \cdot h_l - h_i \\
= & ~ \sum_{j=1}^b R_{j,i}^2 h_i - h_i + \sum_{j=1}^b \sum_{l \in [n] \backslash i} R_{j,l} R_{j,i} \cdot h_l \\
= & ~ \sum_{j=1}^b \sum_{l \in [n] \backslash i} R_{j,l} R_{j,i} \cdot h_l & \text{~by~} R_{j,i}^2 = 1/b \\
= & ~ \sum_{l \in [n] \backslash i} h_l \langle R_{*,l}, R_{*,i} \rangle \\
= & ~ \sum_{l \in [n] \backslash i} h_l  \cdot  \langle \sigma_l \ov{R}_{*,l}, \sigma_i \ov{R}_{*,i} \rangle & \text{~by~} R_{*,l} = \sigma_l \ov{R}_{*,l}
\end{align*}

First, we apply Khintchine's inequality, we have
\begin{align*}
\Pr \left[ \left| \sum_{l \in [n] \backslash i} h_l \cdot \sigma_l \cdot \langle \ov{R}_{*,l}, \sigma_i \ov{R}_{*,i} \rangle  \right| \geq C t \left( \sum_{l \in [n] \backslash i} h_l^2 ( \langle \ov{R}_{*,l}, \sigma_i \ov{R}_{*,i} \rangle )^2 \right)^{1/2} \right] \leq \exp(-C't^2 )
\end{align*}
and choose $t = \sqrt{ \log ( n / \delta ) }$.

For each $l\neq i$, using \cite{ldfu13} we have
\begin{align*}
\Pr \left[ | \langle \ov{R}_{*,l} , \ov{R}_{*,i} \rangle | \geq \frac{ \sqrt{ \log ( n / \delta ) } }{ \sqrt{b} } \right] \leq \delta / n .
\end{align*}

Taking a union bound over all $l \in [n] \backslash i$, we have
\begin{align*}
| (R^\top R h)_i - h_i | \leq \| h \|_2 \frac{ \log ( n / \delta)  }{ \sqrt{b} }
\end{align*}
with probability $1 - \delta$.
\end{proof}

\section*{Acknowledgments}
The authors would like to thank Haotian Jiang, Swati Padmanabhan, Ruoqi Shen, Zhengyu Wang, Xin Yang and Peilin Zhong for useful discussions.
\end{document}